\PassOptionsToPackage{colorlinks=true,linkcolor={blue!66!black},citecolor={blue!66!black},urlcolor=magenta}{hyperref}
\documentclass[numsec,webpdf,modern,medium,namedate]{oup-authoring-template}

\newif\ifsubmissionlayout
\submissionlayoutfalse

\onecolumn
\graphicspath{{figures/}}

\usepackage{amsfonts}
\usepackage{mathtools}
\usepackage{bm}
\usepackage{booktabs}
\usepackage{graphicx}
\usepackage{multirow}
\usepackage{tabularx}
\usepackage{longtable}
\usepackage{placeins}
\usepackage{tikz-cd}
\ifsubmissionlayout
  \usepackage[doublespacing]{setspace}
  \usepackage[fontsize=12pt]{fontsize}
  \usepackage[nomarkers,tablesonly]{endfloat}
\else
  \usepackage{setspace}
  \usepackage{needspace}
  \usepackage[final]{microtype}
\fi

\makeatletter
\AtBeginDocument{%
  \renewenvironment{proof}[1][\proofname]{\par\removelastskip%
    \pushQED{\ensuremath{\qed}}%
    \normalfont \topsep7.5\p@\@plus7.5\p@\relax%
    \trivlist%
    \item[\hskip\labelsep\bfseries\itshape #1\@addpunct{.}\enspace]\ignorespaces%
  }{%
    \popQED\endtrivlist\@endpefalse\par\smallskip%
  }%
}
\makeatother

\theoremstyle{thmstyleone}
\newtheorem{theorem}{Theorem}
\newtheorem{lemma}{Lemma}
\newtheorem{proposition}{Proposition}
\newtheorem{corollary}{Corollary}
\theoremstyle{thmstylethree}

\theoremstyle{thmstyletwo}
\newtheorem{remark}{Remark}

\newcolumntype{L}[1]{>{\raggedright\arraybackslash}p{#1}}
\newcolumntype{Y}{>{\raggedright\arraybackslash}X}

\algrenewcommand\algorithmicrequire{\textbf{Input:}}
\algrenewcommand\algorithmicensure{\textbf{Output:}}
\algtext*{EndIf}
\algtext*{EndFor}

\newcommand{\R}{\mathbb{R}}

\newcommand{\KL}{\operatorname{KL}}

\newcommand{\VR}{\operatorname{VR}}

\newcommand{\parlabel}[1]{\par\noindent\textbf{\textit{#1}}\quad}

\newcommand{\norm}[1]{\left\lVert #1 \right\rVert}

\newcommand{\pfstep}[1]{\medskip\noindent\textbf{\textit{#1.}}\quad}
\newcommand{\alttext}[1]{}

\begin{document}

\journaltitle{}
\DOI{}
\hidepubinfo
\copyrightyear{XXXX}
\pubyear{XXXX}
\access{}
\appnotes{}

\firstpage{1}

% Keep the running head within 50 characters.
\title[Persistent Laplacian Change-Point Detection]{Online Change-Point Detection with Persistent Laplacian Features}

\author[1]{Shiheng Nie}
\author[1,$\ast$]{Yunguang Yue}

\authormark{Nie and Yue}

\address[1]{\orgdiv{College of Science}, \orgname{Shihezi University}, \orgaddress{\street{Shihezi}, \postcode{832003}, \state{Xinjiang}, \country{China}}}

\corresp[$\ast$]{Address for correspondence. Yunguang Yue, College of Science, Shihezi University, Shihezi, 832003, Xinjiang, China. Email: \href{mailto:guangyy@shzu.edu.cn}{guangyy@shzu.edu.cn}}

% Fill these only when publisher metadata is needed.
% \received{Date}{0}{Year}
% \revised{Date}{0}{Year}
% \accepted{Date}{0}{Year}

% The submission abstract should be a single paragraph of at most 200 words,
% without citations or unexplained abbreviations.
\abstract{
Online change-point detection in high-dimensional nonlinear time series faces
two challenges. The underlying distributions are difficult to model, and
state changes are difficult to characterize. We propose persistent Laplacian
cumulative sum (PL-CUSUM) to address these challenges. PL-CUSUM maps
delay-embedded sliding windows to point clouds. It extracts persistent Betti
vectors and the positive spectra of persistent Laplacians from
the same Vietoris--Rips filtration. A ridge-whitened projection converts these
features into a scalar score. Page's CUSUM recursion then accumulates this score
over time. The positive spectra capture within-scale connectivity and geometric
information that persistent Betti vectors do not record. Under a finite-support
local model, we prove that the oracle upper bound on detection delay and a local
minimax lower bound have the same order. This common order is the logarithm of
the average run length constraint divided by the squared ridge-whitened
separation. We also establish finite-horizon false-alarm and expected-delay
bounds for plug-in whitened scores under weak dependence within each state.
The method has two phases.
Phase~I estimates the projection parameters, selects the feature configuration,
and calibrates the control limit. Phase~II updates the resulting CUSUM statistic
online. Experiments on simulated and real monitoring data show stable
false-alarm control and competitive detection performance.
}

\maketitle

\section{Introduction}
\label{sec:intro}

Online change-point detection seeks to minimize detection delay subject to a
prescribed false-alarm constraint. The average run length (ARL) is the standard
measure of false-alarm control. Lorden's criterion considers the worst-case
conditional delay, whereas Pollak's criterion considers the worst conditional
average delay~\citep{lorden1971,pollak1985}. These criteria give two standard
minimax formulations subject to an ARL constraint. When the pre- and
post-change distributions are known, Page's cumulative sum (CUSUM) procedure
recursively accumulates log-likelihood-ratio increments~\citep{page1954}. The
likelihood-ratio CUSUM is optimal under Lorden's
criterion~\citep{moustakides1986}. When the distributions are not fully known,
however, the log-likelihood ratio cannot be computed directly, and this
optimality result no longer applies. Modern monitoring data often exhibit
cross-channel dependence, serial dependence, and nonlinear
coupling~\citep{colosimo2024,qiu2022}. Specifying a parametric joint distribution
for such data is difficult and prone to misspecification. This difficulty
motivates nonparametric extensions of classical CUSUM procedures.

Existing methods address this problem in different ways. Window-limited CUSUM
updates its statistic using information from the most recent fixed-length
window~\citep{xie2023window_limited_cusum}. Online kernel CUSUM constructs
monitoring scores from kernel statistics~\citep{wei2022online_kernel_cusum}.
Graph-based methods form detection statistics from adjacency relations among
observations~\citep{chen2022graph}. These methods represent changes through a
window statistic, a kernel function, or a graph structure. In nonlinear
dynamical systems, however, the differences between the pre- and post-change
regimes can be subtle and complex. A feature that summarizes only one type of
discrepancy may miss changes in both the global and local geometry of the
observations within a window.

Topological data analysis (TDA) provides a multiscale representation of such
structural changes. For each sliding window, delay embedding maps the local
observations to a point cloud~\citep{perea2015}. As the scale increases, the
Vietoris--Rips filtration gradually connects nearby points. It represents the
geometric relations in the point cloud as a nested family of simplicial
complexes. Persistent homology tracks connected components, loops, and
higher-dimensional holes across this family. It yields persistent Betti numbers
and persistence diagrams. The persistent Laplacian (PL) further associates
spectral operators with the same filtration~\citep{wang2020psg}. These operators
capture within-scale connectivity and geometric information beyond homological
counts~\citep{memoli2022}.

We propose PL-CUSUM, a CUSUM procedure based on persistent Laplacian spectral
features, for online change-point detection in high-dimensional nonlinear time
series. Existing topological change-point methods typically construct detection
statistics from barcodes, persistence diagrams, or their discretized summaries.
In contrast, PL-CUSUM uses persistent Betti vectors and the positive spectra of
persistent Laplacians to construct monitoring scores within a common sequential
detection framework. PL-CUSUM separates filtration scale from monitoring time.
The PL features summarize how the structure changes with filtration scale within
each window. The CUSUM recursion accumulates detection evidence across the
sequence of windows. The tree-family construction gives point clouds with the
same persistent Betti vector but different positive PL spectra. Thus, the
positive spectrum can contain state information not captured by the persistent
Betti vector. We study two standard performance measures. The average
run length measures false-alarm control before the change, whereas Lorden's
criterion measures the worst-case detection delay after the change. Under a
finite-support local model, we derive a delay upper bound for an oracle CUSUM
based on log-likelihood-ratio increments. We also establish a local minimax
lower bound. Both bounds are of order $\log\mathcal A/\gamma_\rho^2$. Here,
$\mathcal A$ is the average run length constraint, and $\gamma_\rho$ is the
ridge-whitened separation. We also derive finite-horizon false-alarm and
expected-delay bounds for plug-in whitened scores under weak dependence within
each state. Following
\citet{qiu2022}, PL-CUSUM uses a two-phase procedure. Phase~I selects the feature
configuration, estimates the projection parameters, and calibrates the control
limit. Phase~II applies the resulting CUSUM recursion online. The synthetic
experiments examine spectral separation, oracle delay scaling, and the effect
of the Phase~I sample size. Experiments on SWaT and Electric Motor Vibrations
evaluate false-alarm control and detection performance for a high-dimensional
process and mechanical vibration signals.

\section{Related Work}
\label{sec:related_work}

\parlabel{Topology in change-point detection.}
Existing studies can be divided into two groups based on how they use temporal
information. The first group maps each window or observation to a topological
summary and then studies changes in the resulting sequence.
\citet{gideakatz2018landscapes} computed persistence landscapes for sliding
windows of financial time series to describe structural changes around market
crashes. \citet{islambekov2019} used Betti-number sequences in a nonparametric
change-point framework. \citet{zheng2023percept} proposed PERCEPT, which
represents persistence diagrams by persistence histograms and constructs an
online detection statistic. \citet{thomas2025} extracted topological covariates
from image sequences and used Bayesian logistic regression to infer the change
point and regression coefficients jointly. The second group keeps temporal
relations in the graph or filtration. \citet{stolz2017functional} constructed
time-varying functional networks from coupled time series and used persistent
homology to describe changes in synchronization structure.
\citet{myers2019persistent} applied persistent homology to graph representations
of time series. Their ordinal partition network encodes transitions between
ordinal states and helps distinguish dynamical states. \citet{tan2023attractor}
constructed an attractor network from delay-embedded trajectories and used
unusual Markov transitions to form an online change-point score.
\citet{myers2023zigzag} used zigzag persistence to link adjacent temporal-network
snapshots and track changes in network shape. These studies mainly use
Betti-number sequences, persistence diagrams, persistence histograms,
persistence landscapes, or related persistent-homology summaries. The role of
persistent Laplacian spectra in online change-point detection has not been
studied systematically.

\parlabel{Theory and applications of persistent Laplacians.}
Research on persistent Laplacians now covers algebraic constructions,
algorithms, stability, data representation, and machine-learning applications.
\citet{memoli2022} systematically studied their properties, algorithms, and
stability. \citet{liu2024_algebraic_stability} later established an algebraic
stability theorem for persistent Laplacians. \citet{davies2023} evaluated
persistent Laplacian spectra for data embedding, classification, and regression.
\citet{cottrell2024plpca} proposed persistent Laplacian-enhanced principal
component analysis (PLPCA) and applied it to microarray data.
\citet{cottrell2024knn_tpca} proposed k-nearest-neighbors-induced topological PCA
(kNN-tPCA) for dimension reduction and feature selection in single-cell RNA
sequencing data. They evaluated the resulting representations in clustering and
classification tasks. However, the use of PL spectral features in sequential
procedures for online change-point detection and related time-series monitoring
remains largely unexplored.

\parlabel{Statistical foundations for topological summaries.}
Existing studies provide several foundations for incorporating topological
summaries into statistical models. These foundations include stability theory,
statistical inference, and vector and kernel representations. For stability,
\citet{cohensteiner2007} established the bottleneck stability theorem for
persistence diagrams. For statistical inference,
\citet{chazal2014_convergence} studied bootstrap inference for persistence
diagrams and persistence landscapes. \citet{fasy2014} constructed confidence
sets for persistence diagrams. \citet{turner2014} studied Fr\'echet means of
distributions of persistence diagrams, including their computation and
convergence. For representations, \citet{bubenik2015} introduced persistence
landscapes and proved a strong law of large numbers and a central limit theorem.
\citet{patrangenaru2019object} used persistence landscapes to vectorize object
data for statistical analysis. \citet{adams2017persistence_images} proposed
persistence images as a stable vector representation. \citet{carriere2017}
constructed the sliced Wasserstein kernel for persistence diagrams. However,
existing work has not systematically studied sequential detection procedures
based on topological features, including persistent-homology summaries and
persistent Laplacian spectra, under false-alarm constraints and detection-delay
criteria.

\section{Method}
\label{sec:method}

\subsection{Online Monitoring and Point-Cloud Representation}
\label{sec:problem_setup}

Let $(x_t)_{t\ge1}$ be a stream of multivariate observations, with
$x_t\in\mathbb R^d$. Let $w$, $h_{\rm stride}$, and $L$ be positive integers
representing the window length, stride, and delay-embedding dimension,
respectively. Assume that $1\le L<w$. Each delay-coordinate vector contains
$L$ consecutive observations.

For $n\ge1$, when $x_{(n-1)h_{\rm stride}+w}$ arrives, the $n$th monitoring
window is
\[
    W_n
    =
    \bigl(
    x_{(n-1)h_{\rm stride}+1},
    \ldots,
    x_{(n-1)h_{\rm stride}+w}
    \bigr).
\]
Within $W_n$, define the delay-coordinate vector
\[
    z_{n,j}
    =
    \bigl(
    x_{(n-1)h_{\rm stride}+j}^\top,
    \ldots,
    x_{(n-1)h_{\rm stride}+j-L+1}^\top
    \bigr)^\top
    \in\mathbb R^{dL},
    \qquad j=L,\ldots,w.
\]
Set $m=w-L+1$. These vectors form the point cloud for the $n$th monitoring
window:
\[
    Z_n
    =
    (z_{n,L},z_{n,L+1},\ldots,z_{n,w})
    \in(\mathbb R^{dL})^m.
\]
As observations arrive, the monitoring window is updated every
$h_{\rm stride}$ observations. This produces the point-cloud sequence
$(Z_n)_{n\ge1}$, which serves as input to the online monitoring procedure.

Let $Q_0$ and $Q_1$ denote the pre- and post-change marginal distributions of
the window point clouds, with $Q_0\ne Q_1$. We define $\nu$ as the change point
at the window level: $Z_n$ has marginal distribution $Q_0$ for $n<\nu$ and
$Q_1$ for $n\ge\nu$.

For each state $k\in\{0,1\}$, assume that the raw observation process
$(x_t)_{t\ge1}$ is strictly stationary and geometrically $\beta$-mixing:
\[
    \beta_{x,k}(r)\le C_k\exp(-c_kr),\qquad r\ge1,
\]
where $C_k,c_k>0$, and $\beta_{x,k}(r)$ measures the dependence between the
past and the future separated by $r$ time points~\citep{bradley2005}.

Proposition~\ref{prop:overlap_decimation_main} in the Appendix shows that the
point-cloud sequence remains strictly stationary and that its $\beta$-mixing
coefficients satisfy
$\beta_{Z,k}(r)\le\beta_{x,k}((r h_{\rm stride}-w+1)_+)$.
Given independent Phase~I data, the point-cloud features and the PL-CUSUM
one-step increments satisfy the same mixing bound. Thus, overlapping windows
only change the lag at which dependence begins to decay. If the raw
observations are independent, let $g=\lceil w/h_{\rm stride}\rceil$. The point
clouds, feature vectors, and one-step increments are then all
$(g-1)$-dependent sequences.

Control-limit calibration uses a moving-block bootstrap (MBB) to resample
state~0 score paths and preserve local dependence in the calibration
sequence~\citep{kunsch1989}. The block lengths and data partitions are given
in Appendix~\ref{app:protocol_details}. The independent and identically
distributed (i.i.d.) finite-support model in
Section~\ref{sec:finite_support_observation} provides the oracle local
benchmark. The plug-in analysis in Section~\ref{sec:plugin_guarantee} allows
the statewise weak dependence described above.

Let $\mathscr F_n^Z=\sigma(Z_1,\ldots,Z_n)$ be the $\sigma$-algebra generated
by the first $n$ window point clouds. The alarm time $T$ is a stopping time with respect to
$\{\mathscr F_n^Z\}_{n\ge1}$. Let $\mathbb P_\infty$ and $\mathbb E_\infty$
denote probability and expectation under the no-change regime. Let
$\mathbb P_\nu$ and $\mathbb E_\nu$ denote probability and expectation when
the change occurs at the $\nu$th window.

False-alarm risk is measured by the average run length
$\mathrm{ARL}_0(T)=\mathbb E_\infty T$ or the finite-horizon false-alarm
probability $\mathrm{FAR}_N(T)=\mathbb P_\infty(T\le N)$. Post-change delay is
measured by Lorden's worst-case criterion:
\[
    \bar d(T)
    =
    \sup_{\nu\ge1}
    \operatorname*{ess\,sup}
    \mathbb E_\nu
    \left[
        (T-\nu+1)_+
        \,\middle|\,
        \mathscr F_{\nu-1}^Z
    \right].
\]

\subsection{PL Feature Construction for Window Point Clouds}
\label{sec:pl_spectral_features}

This subsection constructs two types of features from each window point cloud:
persistent Betti vectors and positive-spectrum vectors of persistent
combinatorial Laplacians. For background on persistent homology, see
\citet{edelsbrunner2010}. For definitions and properties of persistent
Laplacians, see \citet{wang2020psg,memoli2022}.

For any scale $\varepsilon>0$, define the Vietoris--Rips complex of $Z_n$ on
the vertex set $\{L,\ldots,w\}$ by
\[
    \VR_\varepsilon(Z_n)
    =
    \left\{
        \sigma\subseteq\{L,\ldots,w\}:
        \|z_{n,i}-z_{n,j}\|_2\le\varepsilon
        \text{ for all }i,j\in\sigma
    \right\}.
\]
For a finite point cloud, the Vietoris--Rips complex changes only when
$\varepsilon$ crosses a pairwise distance. Thus, finitely many critical scales
determine the full filtration~\citep{edelsbrunner2010}. We use an increasing
scale grid $0<\varepsilon_1<\cdots<\varepsilon_K$.
For $1\le a\le b\le K$, call $(a,b)$ a scale pair associated with
$(\varepsilon_a,\varepsilon_b)$. Since
$\VR_{\varepsilon_a}(Z_n)\subseteq\VR_{\varepsilon_b}(Z_n)$, the inclusion
induces a linear map on the $q$th homology groups with real
coefficients. For each integer $q\ge0$, define the $q$th persistent Betti
number by
\[
    \beta_q^{a,b}(Z_n)
    =
    \dim\operatorname{im}
    \left\{
        H_q\bigl(\VR_{\varepsilon_a}(Z_n);\mathbb R\bigr)
        \longrightarrow
        H_q\bigl(\VR_{\varepsilon_b}(Z_n);\mathbb R\bigr)
    \right\}.
\]
A Vietoris--Rips complex on $m$ vertices has dimension at most $m-1$.
Therefore, $\beta_q^{a,b}(Z_n)=0$ for $q\ge m$. Choose a homological
truncation level $q_{\max}\in\{0,\ldots,m-1\}$ and a set of scale pairs
$\mathcal S\subseteq\{(a,b):1\le a\le b\le K\}$.
For each $(a,b)\in\mathcal S$, define the persistent Betti vector by
\[
    B^{a,b}(Z_n)
    =
    \bigl[\beta_q^{a,b}(Z_n)\bigr]_{q=0}^{q_{\max}}
    \in\mathbb R^{q_{\max}+1}.
\]
We next construct the positive spectra of persistent combinatorial Laplacians
on the same Vietoris--Rips filtration. For integers $q\ge0$ and
$r=1,\ldots,K$, let
\[
    C_q^r
    =
    C_q\bigl(\VR_{\varepsilon_r}(Z_n);\mathbb R\bigr),
    \qquad
    \partial_q^r:C_q^r\to C_{q-1}^r
\]
denote the real $q$-chain group and boundary operator of
$\VR_{\varepsilon_r}(Z_n)$, respectively. Set $C_{-1}^r=\{0\}$ and
$\partial_0^r=0$. Equip each chain group with the inner product that makes
the oriented simplex basis orthonormal.

For $1\le a\le b\le K$, define
\[
\begin{aligned}
    \mathcal C_q^{a,b}
    &:=
    \{c\in C_q^b:\partial_q^b c\in C_{q-1}^a\},\\
    \delta_q^{a,b}
    &:=
    \left.\partial_q^b\right|_{\mathcal C_q^{a,b}}
    :
    \mathcal C_q^{a,b}\to C_{q-1}^a .
\end{aligned}
\]
The $q$th persistent combinatorial Laplacian of $Z_n$ at the scale pair
$(a,b)$ is the following operator on $C_q^a$:
\[
    \mathcal L_q^{a,b}(Z_n)
    =
    \delta_{q+1}^{a,b}(\delta_{q+1}^{a,b})^*
    +
    (\partial_q^a)^*\partial_q^a.
\]
Here, $*$ denotes the adjoint with respect to the inner products defined
above. If $a=b$, then $\mathcal L_q^{a,a}(Z_n)$ is the standard $q$th
combinatorial Laplacian of $\VR_{\varepsilon_a}(Z_n)$.

If $q=0$, all complexes in the Vietoris--Rips filtration have the same vertex
set. Therefore, $\mathcal L_0^{a,b}(Z_n)$ is the graph Laplacian of
$\VR_{\varepsilon_b}(Z_n)$, and
$\beta_0^{a,b}(Z_n)=\beta_0\bigl(\VR_{\varepsilon_b}(Z_n)\bigr)$.

The persistent combinatorial Laplacian satisfies
$\dim\ker\mathcal L_q^{a,b}(Z_n)=\beta_q^{a,b}(Z_n)$~\citep{memoli2022}.
Thus, the multiplicity of its zero eigenvalue equals the corresponding
persistent Betti number.

Arrange the positive eigenvalues of $\mathcal L_q^{a,b}(Z_n)$ in decreasing
order. Retain the first $J$ eigenvalues and pad the vector with zeros if fewer
than $J$ are available. Denote the resulting vector by
$\lambda_q^{a,b}(Z_n)\in\mathbb R^J$.
If $q\ge m$, then $C_q^a=\{0\}$ and
$\lambda_q^{a,b}(Z_n)=0_J$, where $0_J$ is the zero vector in
$\mathbb R^J$.

For each scale pair $(a,b)$, define the positive-spectrum vector by
\[
    \Lambda^{a,b}(Z_n)
    =
    \bigl[\lambda_q^{a,b}(Z_n)\bigr]_{q=0}^{q_{\max}}
    \in\mathbb R^{J(q_{\max}+1)}.
\]
The corresponding full PL feature is
\[
    \Phi^{a,b}(Z_n)
    =
    \bigl(B^{a,b}(Z_n),\Lambda^{a,b}(Z_n)\bigr)
    \in\mathbb R^{(J+1)(q_{\max}+1)}.
\]
\subsection{Whitened Score and CUSUM Statistic}
\label{sec:score_statistic}

Let $F:(\mathbb R^{dL})^m\to\mathbb R^p$ be a Borel-measurable
finite-dimensional feature map on point clouds. In PL-CUSUM, $F$ can be
$\Phi^{a,b}$, $B^{a,b}$, or $\Lambda^{a,b}$. These
choices correspond to the full PL feature, the persistent Betti vector, and
the PL positive-spectrum vector, respectively. For any such map, set
$X_n=F(Z_n)\in\mathbb R^p$.

With the point-cloud size, scale grid, scale-pair set, homological truncation
level, and spectral truncation level fixed,
Proposition~\ref{prop:finite_image_pl_main} shows that the distribution of
$X_n$ has finite support. The local model in
Section~\ref{sec:theory_interface} describes changes in the probabilities of
these feature values. The raw observations $x_t$ and the window point clouds
$Z_n$ may still take values in continuous spaces.

\noindent\textbf{Ridge-whitened separation.}
\label{def:separation}
Under state $k\in\{0,1\}$, let
\[
    \mu_k=\mathbb E_kX_n,
    \qquad
    \Sigma_k=\operatorname{Var}_k(X_n),
    \qquad
    \Sigma=\frac{\Sigma_0+\Sigma_1}{2}.
\]
To measure the difference between the state means on the pooled covariance
scale, let $\rho>0$ be a ridge parameter and define
\begin{equation}
\label{eq:ridge_whitened_separation}
    W_\rho=(\Sigma+\rho I_p)^{-1/2},
    \qquad
    \gamma_\rho
    =
    \bigl\|W_\rho(\mu_1-\mu_0)\bigr\|_2.
\end{equation}
Here, $I_p$ is the $p$-dimensional identity matrix. If $\rho=0$ and $\Sigma$
is positive definite, then
$\gamma_0^2=(\mu_1-\mu_0)^\top\Sigma^{-1}(\mu_1-\mu_0)$. Thus, $\gamma_0$ is
the Mahalanobis distance based on the pooled covariance
matrix~\citep{mahalanobis1936}.

When $\gamma_\rho>0$, define
$u=W_\rho(\mu_1-\mu_0)/\gamma_\rho$.
When $\gamma_\rho=0$, set $u=0$. The oracle scalar score and its CUSUM
reference value are
\[
    s_n=u^\top W_\rho X_n,
    \qquad
    c
    =
    \frac12u^\top W_\rho(\mu_0+\mu_1)
    =\frac{\mathbb E_0s_n+\mathbb E_1s_n}{2}.
\]
It follows that
\[
    \mathbb E_0(s_n-c)
    =
    -\frac{\gamma_\rho}{2},
    \qquad
    \mathbb E_1(s_n-c)
    =
    \frac{\gamma_\rho}{2}.
\]
Thus, the reference value gives equal drift magnitudes under the two states.
Among all $r\in[\mathbb E_0s_n,\mathbb E_1s_n]$, the midpoint $c$ maximizes
$\min\{r-\mathbb E_0s_n,\mathbb E_1s_n-r\}$.

In Phase~I, let
$X_{k1},\ldots,X_{kN_{\rm tr}}$ be the $N_{\rm tr}$ training feature vectors
from state $k$. Compute
\[
    \widehat\mu_k
    =
    \frac{1}{N_{\rm tr}}\sum_{j=1}^{N_{\rm tr}}X_{kj},
    \qquad
    \widehat\Sigma_k
    =
    \frac{1}{N_{\rm tr}-1}\sum_{j=1}^{N_{\rm tr}}
    (X_{kj}-\widehat\mu_k)(X_{kj}-\widehat\mu_k)^\top.
\]
and set $\widehat\Sigma=(\widehat\Sigma_0+\widehat\Sigma_1)/2$.
Before whitening, pool the training vectors from both states and remove
coordinates whose standard deviation is at most $10^{-10}$. If this rule
would remove every coordinate, retain all coordinates. After this step, $p$
denotes the number of retained coordinates.

For computation, let $\rho_{\rm mult}>0$ be a ridge multiplier and set the
ridge floor to $\rho_{\min}=10^{-10}$. Using the average coordinate variance
as the covariance scale, define the population and sample ridge values by
\[
\begin{aligned}
    \rho
    & =
    \rho_{\rm mult}
    \max\left\{\frac{\operatorname{tr}(\Sigma)}{p},\rho_{\min}\right\},
    &\qquad
    \widehat\rho
    & =
    \rho_{\rm mult}
    \max\left\{\frac{\operatorname{tr}(\widehat\Sigma)}{p},\rho_{\min}\right\},\\
    \widehat W
    & =
    (\widehat\Sigma+\widehat\rho I_p)^{-1/2},
    &\qquad
    \widehat\gamma
    & =
    \bigl\|\widehat W(\widehat\mu_1-\widehat\mu_0)\bigr\|_2.
\end{aligned}
\]
The maximum with $\rho_{\min}$ keeps the ridge value positive when the average
coordinate variance is zero or close to zero.

If $\widehat\gamma>0$, set
$\widehat u=\widehat W(\widehat\mu_1-\widehat\mu_0)/\widehat\gamma$.
Otherwise, set $\widehat u=0$. The Phase~II plug-in score and CUSUM reference
value are
\[
    \widehat s_n
    =
    \widehat u^\top\widehat W X_n,
    \qquad
    \widehat c
    =
    \frac12
    \widehat u^\top\widehat W
    (\widehat\mu_0+\widehat\mu_1).
\]

The online statistic and stopping time are
\[
    S_0=0,
    \qquad
    S_n
    =
    \max\{0,S_{n-1}+\widehat s_n-\widehat c\},
    \qquad
    T_\eta
    =
    \inf\{n\ge1:S_n\ge\eta\}.
\]
We use the convention $\inf\varnothing=\infty$. Control-limit calibration is
described in Section~\ref{sec:training_calibration_eval}. The false-alarm and
delay bounds for finite training samples are given in
Section~\ref{sec:plugin_guarantee}.

\subsection{Training--Calibration--Evaluation Workflow and Algorithms}
\label{sec:training_calibration_eval}

PL-CUSUM uses a two-phase workflow. Training, calibration, selection, and
evaluation use disjoint data partitions. For continuous sequences, guard
intervals separate adjacent partitions.

\noindent\textbf{Phase~I: training, calibration, and selection.}
\emph{Training.}
Each prespecified configuration
consists of a feature input, delay-embedding dimension $L$, spectral truncation
level $J$, ridge multiplier $\rho_{\rm mult}$, and scale-pair set $\mathcal S$.
The feature input can be the full PL feature, the PL positive-spectrum vector,
or the persistent Betti vector. For each configuration, the training segment
estimates $\widehat\gamma^{a,b}$ for every $(a,b)\in\mathcal S$ as described in
Section~\ref{sec:score_statistic}. It selects the scale pair with the largest
$\widehat\gamma^{a,b}$ and computes the corresponding $\widehat W$,
$\widehat u$, and $\widehat c$.

\noindent\emph{Calibration.}
The calibration segment uses only state-0 data. For the scale pair selected
during training, it computes the scalar scores. It then uses the MBB to
generate $M_{\rm cal}$ state-0 score paths of length
$N$~\citep{kunsch1989}. Each path starts from $S_0=0$ and follows the same
CUSUM recursion. For path $m$, let $M^{(m)}=\max_{t\le N}S_t^{(m)}$.
The empirical finite-horizon false-alarm probability at a control limit
$\eta$ is
\[
    \widehat{\mathrm{FAR}}_N(\eta)
    =
    \frac{1}{M_{\rm cal}}
    \sum_{m=1}^{M_{\rm cal}}
    \mathbf 1\{M^{(m)}\ge\eta\}.
\]
The path maxima determine the possible control-limit values. We choose
$\widehat\eta$ as the smallest value that satisfies
$\widehat{\mathrm{FAR}}_N(\widehat\eta)\le\alpha$.
If no control limit satisfies this condition, the configuration is marked as
failing the false-alarm budget.

\noindent\emph{Selection.}
Each sequence in the selection segment consists of a state-0 prefix followed
by a post-change state-1 segment. A configuration is eligible only if it
passes calibration and keeps the pre-change false-alarm proportion within its
budget. Among the eligible configurations, we first select the one with the
largest proportion of sequences that have no pre-change alarm and a
post-change detection. Ties are resolved by the smaller median delay. If a tie
remains, we select the configuration with the larger post-change detection
proportion. The estimation error caused by finite training samples is analysed
in Section~\ref{sec:plugin_guarantee}.

\noindent\textbf{Phase~II: online monitoring.}
Phase~II uses the feature map, projection parameters, and control limit
determined in Phase~I. It performs online monitoring using the recursion in
Section~\ref{sec:score_statistic}. Algorithm~\ref{alg:pl_cusum} presents the
training, calibration, and online recursion for a given feature input and
parameter combination. The selection step is performed outside the algorithm.
The experimental metrics and baseline calibration procedures are described in
Section~5.

\begin{algorithm}[H]
\small
\caption{PL-CUSUM: training, calibration, and online recursion for a given configuration}
\label{alg:pl_cusum}
\begin{algorithmic}[1]
\Require Training point clouds
$\mathcal Z_r^{\rm tr}=\{Z_{rj}^{\rm tr}\}_{j=1}^{N_{\rm tr}}$ for states
$r\in\{0,1\}$, state-0 calibration point clouds $\mathcal Z_{\rm cal}$, and
the Phase~II point-cloud stream $(Z_n)_{n\ge1}$
\Statex Feature family $\{F^{a,b}:(a,b)\in\mathcal S\}$ and parameters
$\rho_{\rm mult}$, $\alpha$, $N$, and $M_{\rm cal}$
\Ensure Selected $F$, $(\widehat W,\widehat u,\widehat c)$,
$\widehat\eta$, and stopping time $T$, or a flag indicating that the
false-alarm budget is not met

\Statex \textbf{Phase~I: training and scale selection}
\ForAll{$(a,b)\in\mathcal S$}
    \State Compute $X_{rj}^{a,b}=F^{a,b}(Z_{rj}^{\rm tr})$ for
    $r\in\{0,1\}$
    \State Remove coordinates whose standard deviation is at most $10^{-10}$
    across the pooled training vectors
    \State Compute $\widehat\mu_r^{a,b}$, $\widehat\Sigma^{a,b}$,
    $\widehat W^{a,b}$, $\widehat u^{a,b}$, $\widehat c^{a,b}$, and
    $\widehat\gamma^{a,b}$
\EndFor
\State $(a^*,b^*)\gets
\arg\max_{(a,b)\in\mathcal S}\widehat\gamma^{a,b}$
\State $F\gets F^{a^*,b^*}$ and
$(\widehat W,\widehat u,\widehat c)\gets
(\widehat W^{a^*,b^*},\widehat u^{a^*,b^*},\widehat c^{a^*,b^*})$

\Statex \textbf{Phase~I: control-limit calibration}
\State Compute
$\widehat s_i^{\rm cal}=\widehat u^\top\widehat W F(Z_i^{\rm cal})$
\For{$m=1,\ldots,M_{\rm cal}$}
    \State Block-resample the calibration scores to obtain a length-$N$ path
    $(\widehat s_1^{(m)},\ldots,\widehat s_N^{(m)})$
    \State Set $S_0^{(m)}\gets0$ and, for $t=1,\ldots,N$, update
    $S_t^{(m)}\gets
    \max\{0,S_{t-1}^{(m)}+\widehat s_t^{(m)}-\widehat c\}$
    \State $M^{(m)}\gets\max_{t\le N}S_t^{(m)}$
\EndFor
\State Compute $\widehat{\mathrm{FAR}}_N(\eta)$ for each control limit
determined by $\{M^{(m)}\}_{m=1}^{M_{\rm cal}}$
\If{no control limit satisfies
$\widehat{\mathrm{FAR}}_N(\eta)\le\alpha$}
    \State Mark the configuration as failing the false-alarm budget and \Return
\EndIf
\State $\widehat\eta\gets
\min\{\eta:\widehat{\mathrm{FAR}}_N(\eta)\le\alpha\}$

\Statex \textbf{Phase~II: online recursion}
\State $n\gets0$ and $S_0\gets0$
\While{\textnormal{true}}
    \State $n\gets n+1$ and read the $n$th window point cloud $Z_n$
    \State $X_n\gets F(Z_n)$ and
    $\widehat s_n\gets\widehat u^\top\widehat W X_n$
    \State $S_n\gets\max\{0,S_{n-1}+\widehat s_n-\widehat c\}$
    \If{$S_n\ge\widehat\eta$}
        \State $T\gets n$, raise an alarm, and \Return $T$
    \EndIf
\EndWhile
\end{algorithmic}
\end{algorithm}

\FloatBarrier

\section{Theory}
\label{sec:theory_interface}

\subsection{Local Model and Separation}
\label{sec:finite_support_observation}

Let $X_t=F(Z_t)\in\{x_1,\ldots,x_M\}\subset\mathbb R^p$ be the feature
vector at monitoring time $t$. If the change occurs at $\nu$,
then $X_t$ is independently distributed according to $P_0$ for $t<\nu$ and
according to $P_\theta$ for $t\ge\nu$. For $i=1,\ldots,M$, let $P_0(i)$ be
the probability mass assigned to the support point $x_i$, and assume that
$P_0(i)>0$. Define $\mathscr F_t^X=\sigma(X_1,\ldots,X_t)$. All alarm times below are stopping times with respect to
$\{\mathscr F_t^X\}_{t\ge1}$. The ARL and Lorden delay are defined as in
Section~\ref{sec:problem_setup}.

We represent the post-change distribution by an exponential tilting family
with baseline distribution $P_0$. For $\theta\in\mathbb R^p$, define
\[
    \psi_x(\theta)
    =
    \log\left\{
        \sum_{i=1}^M P_0(i)\exp(\theta^\top x_i)
    \right\},
    \qquad
    P_\theta(i)
    =
    P_0(i)\exp\{\theta^\top x_i-\psi_x(\theta)\},
    \qquad i=1,\ldots,M.
\]
By construction, $P_\theta$ is absolutely continuous with respect to $P_0$.
Its mean and covariance matrix are
\[
    \mu_\theta^x
    =
    \sum_{i=1}^M P_\theta(i)x_i,
    \qquad
    \Sigma_\theta^x
    =
    \sum_{i=1}^M
    P_\theta(i)
    (x_i-\mu_\theta^x)(x_i-\mu_\theta^x)^\top.
\]
Following equation~\eqref{eq:ridge_whitened_separation}, define the
ridge-whitened separation between $P_0$ and $P_\theta$ by
\[
    \gamma_\rho(\theta)
    =
    \left\|
        \left(
            \frac{\Sigma_0^x+\Sigma_\theta^x}{2}
            +\rho I_p
        \right)^{-1/2}
        (\mu_\theta^x-\mu_0^x)
    \right\|_2.
\]
Set $G_\rho=\Sigma_0^x(\Sigma_0^x+\rho I_p)^{-1}\Sigma_0^x$.
Let $h$ be a local direction satisfying $h^\top G_\rho h>0$, and set
$\theta_n=h/\sqrt n$. The corresponding local quantities are
\[
    \gamma_\rho^2(\theta_n)
    =
    n^{-1}h^\top G_\rho h+o(n^{-1}),
    \qquad
    I_{h,n}^x
    =
    \KL(P_{\theta_n}\|P_0),
    \qquad
    \kappa_\rho(h)
    =
    \frac{h^\top\Sigma_0^x h}
    {h^\top G_\rho h}.
\]
When $F=\Phi$, denote the corresponding quantities by
$\psi_\Phi$, $\mu_\theta^\Phi$, $\Sigma_\theta^\Phi$,
$\gamma_{\Phi,\rho}$, $I_{h,n}^\Phi$, and
$\kappa_{\Phi,\rho}(h)$. To express the log-likelihood ratio in the next
subsection, let $U_t\in\{1,\ldots,M\}$ be the support-point index of $X_t$.
Thus, $X_t=x_{U_t}$.

\subsection{Oracle Detection Bounds}
\label{sec:local_bounds_main}

Define the log-likelihood-ratio (LLR) increment by
$\ell_{h,n}(i)=\log\{P_{\theta_n}(i)/P_0(i)\}$ for $i=1,\ldots,M$.
The corresponding oracle LLR-CUSUM recursion and stopping time are
\[
    R_0=0,
    \qquad
    R_m
    =
    \max\{0,R_{m-1}+\ell_{h,n}(U_m)\},
    \qquad
    \tau_\eta
    =
    \inf\{m:R_m\ge\eta\}.
\]

Let $\mathcal C_{\mathcal A}$ be the class of
$\{\mathscr F_t^X\}_{t\ge1}$-stopping times that satisfy
$\mathbb E_\infty\tau\ge\mathcal A$. Define Lorden's delay and the minimax
risk by
\[
\begin{aligned}
    \bar d_{\theta_n}(\tau)
    &=
    \sup_{\nu\ge1}
    \operatorname*{ess\,sup}
    \mathbb E_{\nu,\theta_n}
    \left[
        (\tau-\nu+1)_+
        \,\middle|\,
        \mathscr F_{\nu-1}^X
    \right],\\
    R_{n,\mathcal A}(h)
    &=
    \inf_{\tau\in\mathcal C_{\mathcal A}}
    \bar d_{\theta_n}(\tau).
\end{aligned}
\]

The theorem below uses the full PL feature $F=\Phi$. Let $\rho>0$, and let
$h$ be a local direction satisfying $h^\top G_\rho h>0$. Set
$\theta_{n_k}=h/\sqrt{n_k}$ and write
\[
    \gamma_k
    =
    \gamma_{\Phi,\rho}(\theta_{n_k}),
    \qquad
    C_{\Phi,\rho}
    =
    1+\frac{\rho}{\lambda_{\min}^+(\Sigma_0^\Phi)},
    \qquad
    \eta_k
    =
    \log\mathcal A_k.
\]
Here, $\lambda_{\min}^+(\Sigma_0^\Phi)$ is the smallest positive eigenvalue
of $\Sigma_0^\Phi$.

\begin{theorem}
\label{thm:ph_pl_local_minimax_interface}
Suppose that $\mathcal A_k\to\infty$, $n_k\to\infty$, and
\[
    \frac{\log\mathcal A_k}{I_{h,n_k}^{\Phi}}
    =o(\mathcal A_k),
    \qquad
    \log n_k=o(\log\mathcal A_k).
\]
Then the oracle LLR-CUSUM stopping time $\tau_{\eta_k}$ and the local
minimax risk satisfy
\[
\begin{aligned}
    \mathrm{ARL}_0(\tau_{\eta_k})
    &\ge
    \mathcal A_k,\\
    \bar d_{\theta_{n_k}}(\tau_{\eta_k})
    &\le
    \frac{2\log\mathcal A_k}{\gamma_k^2}\{1+o(1)\},\\
    R_{n_k,\mathcal A_k}(h)
    &\ge
    \frac{2}{\kappa_{\Phi,\rho}(h)}
    \frac{\log\mathcal A_k}{\gamma_k^2}\{1+o(1)\}\\
    &\ge
    \frac{2}{C_{\Phi,\rho}}
    \frac{\log\mathcal A_k}{\gamma_k^2}\{1+o(1)\}.
\end{aligned}
\]
\end{theorem}

The first condition ensures that the time segment used in the lower-bound
proof is $o(\mathcal A_k)$. It is therefore negligible relative to the ARL
constraint. The second condition implies
$n_k=\mathcal A_k^{o(1)}$. Since
$I_{h,n_k}^{\Phi}\asymp n_k^{-1}$, the time segment used in the proof grows
only at rate $\mathcal A_k^{o(1)}$. Hence, the remainder in the
change-of-measure argument tends to zero. By
Corollary~\ref{cor:key_local_relations},
\[
    I_{h,n_k}^{\Phi}
    =
    \frac{\kappa_{\Phi,\rho}(h)}{2}
    \gamma_k^2\{1+o(1)\}.
\]
The theorem follows from
Propositions~\ref{prop:pl_cusum_upper_lan_main} and
\ref{prop:local_minimax_lower_main} in the Appendix, together with
Corollary~\ref{cor:key_local_relations}.

\begin{corollary}
\label{cor:detectability_boundary}
Under the conditions of Theorem~\ref{thm:ph_pl_local_minimax_interface},
suppose that $\gamma_k^2\mathcal A_k/\log\mathcal A_k\to\infty$.
Then
\[
    R_{n_k,\mathcal A_k}(h)
    \asymp
    \frac{\log\mathcal A_k}{\gamma_k^2}
    =
    o(\mathcal A_k).
\]
The oracle LLR-CUSUM stopping time $\tau_{\eta_k}$ attains the same delay
order.
\end{corollary}
The proof is given in Appendix~\ref{app:main_result_proofs}. If
$\gamma_k^2\mathcal A_k=O(\log\mathcal A_k)$,
then $\log\mathcal A_k/\gamma_k^2$ is at least of order $\mathcal A_k$.
Thus, the detectability condition above does not hold.
Proposition~\ref{prop:below_boundary_main} in the Appendix further shows that,
for every
$\zeta\in(0,1)$,
\[
    R_{n_k,\mathcal A_k}(h)
    \ge
    \mathcal A_k^{1-\zeta}\{1+o(1)\}.
\]
Equivalently, for every $c\in(0,1)$, the local minimax delay is no smaller
than order $\mathcal A_k^c$.

\subsection{Betti Invariance and Spectral Separation}
\label{sec:composite_least_favorable}

Let $K_\bullet$ and $L_\bullet$ be two finite filtrations, and let
$\phi_r:K_r\to L_r$ be a simplicial map at each level $r$. For $a\le b$, let
$\iota_{K,*}:H_q(K_a)\to H_q(K_b)$ and
$\iota_{L,*}:H_q(L_a)\to H_q(L_b)$ denote the homology maps induced by the
inclusions in $K_\bullet$ and $L_\bullet$, respectively.

\begin{lemma}
\label{lem:natural_homology_isomorphism_persistent_betti}
Fix a homological dimension $q\ge0$. Suppose that every induced map
$(\phi_r)_*:H_q(K_r)\to H_q(L_r)$ is an isomorphism. Suppose also that these isomorphisms commute with the
filtration inclusions. Thus, for every $a\le b$, the following diagram
commutes:
\[
\begin{tikzcd}
H_q(K_a) \arrow[r,"\iota_{K,*}"] \arrow[d,"(\phi_a)_*"']
&
H_q(K_b) \arrow[d,"(\phi_b)_*"]
\\
H_q(L_a) \arrow[r,"\iota_{L,*}"']
&
H_q(L_b).
\end{tikzcd}
\]
Then, for every $a\le b$,
\[
    \beta_q^{a,b}(K_\bullet)
    =
    \beta_q^{a,b}(L_\bullet).
\]
\end{lemma}

The proof is given in Appendix~\ref{app:main_result_proofs}.

\begin{remark}
\label{rem:dominated_vertex_filtered_collapse}
Suppose that a map between two finite filtrations is obtained by a sequence of
levelwise dominated-vertex deletions. If these deletions are compatible with
the filtration inclusions, then the conditions of
Lemma~\ref{lem:natural_homology_isomorphism_persistent_betti} hold. For the
homotopy interpretation of dominated vertices and strong collapses, see
\citet{barmak2012}.
\end{remark}

We next use a family of trees to construct point clouds with identical
persistent Betti vectors but different positive PL spectra. For a tree $T$,
let $V(T)$ and $E(T)$ denote its vertex and edge sets. Let $\deg_T(v)$ be the
degree of $v\in V(T)$, let $d_T(u,v)$ be the edge-count distance between
$u,v\in V(T)$, and let $L_T$ be the graph Laplacian of $T$.

\begin{proposition}
\label{prop:tree_family}
For any integer $m\geq4$, let $\mathcal T_m$ be the set of all unweighted
trees on $m$ vertices, with unit edge lengths. The following statements hold.

\textup{(1)} Every $T\in\mathcal T_m$ admits a point-cloud representation
\[
    Y_T=\{y_v:v\in V(T)\}\subset\mathbb R^{m-1}
\]
whose squared Euclidean distances reproduce the tree metric:
\[
    \|y_u-y_v\|_2^2=d_T(u,v),
    \qquad u,v\in V(T).
\]

\textup{(2)} For any $T,T'\in\mathcal T_m$ and every scale pair $(a,b)$,
\[
    B^{a,b}(Y_T)=B^{a,b}(Y_{T'}).
\]
Moreover, suppose that $1\leq\varepsilon_a\leq\varepsilon_b<\sqrt{2}$ and
\[
    \sum_{v\in V(T)}\deg_T(v)^2
    \neq
    \sum_{v\in V(T')}\deg_{T'}(v)^2,
\]
then $\mathcal L_0^{a,b}(Y_T)$ and $\mathcal L_0^{a,b}(Y_{T'})$ have
different positive spectra.
\end{proposition}
The proof is given in Appendix~\ref{app:main_result_proofs}.

The path tree $P_m$ and the star tree $S_m$ satisfy the condition in
Proposition~\ref{prop:tree_family}, since
\[
    \sum_{v\in V(P_m)}\deg_{P_m}(v)^2=4m-6,
    \qquad
    \sum_{v\in V(S_m)}\deg_{S_m}(v)^2=m(m-1),
\]
and these quantities differ for $m\geq4$. Their largest graph-Laplacian
eigenvalues also satisfy
\[
    \lambda_{\max}(L_{P_m})
    =2+2\cos(\pi/m)<4\leq m
    =\lambda_{\max}(L_{S_m}).
\]
A numerical study of the spectral separation between the path and star trees
as the tree order varies is given in Appendix~D.1.

To compare two feature representations of the same window point cloud, let
$F_1$ and $F_2$ be finite-range feature maps applied to $Z_n$. Suppose that a
map $g$ satisfies $F_1=g\circ F_2$. Then $F_1(Z_n)$ is fully determined by
$F_2(Z_n)$. For $j\in\{1,2\}$ and $k\in\{0,1\}$, let $P_k^{F_j}$ denote the
distribution of $F_j(Z_n)$ under $Q_k$.

\begin{lemma}
\label{lem:feature_compression_kl}
If $P_1^{F_2}\ll P_0^{F_2}$, then
\[
\begin{aligned}
    \KL(P_1^{F_2}\|P_0^{F_2})
    &=
    \KL(P_1^{F_1}\|P_0^{F_1}) \\
    &\quad+
    \sum_y P_1^{F_1}(y)
    \KL\!\left(
        P_1^{F_2\mid F_1=y}
        \middle\|
        P_0^{F_2\mid F_1=y}
    \right).
\end{aligned}
\]
Here, the sum is over the range of $F_1$. The distribution
$P_k^{F_2\mid F_1=y}$ is the conditional distribution of $F_2(Z_n)$ given
$F_1(Z_n)=y$ under $Q_k$. Hence,
\[
    \KL(P_1^{F_1}\|P_0^{F_1})
    \le
    \KL(P_1^{F_2}\|P_0^{F_2}).
\]
Equality holds if and only if
\[
    P_1^{F_2\mid F_1=y}
    =
    P_0^{F_2\mid F_1=y}
\]
for every $y$ such that $P_1^{F_1}(y)>0$.
\end{lemma}
The proof is given in Appendix~\ref{app:main_result_proofs}.

This lemma has two applications in PL-CUSUM.

\textup{(1)} Set $F_2=\Phi=(B,\Lambda)$, $F_1=\Lambda$, and
$g(B,\Lambda)=\Lambda$. Lemma~\ref{lem:feature_compression_kl} gives
\[
    \KL(P_1^\Lambda\|P_0^\Lambda)
    \le
    \KL(P_1^\Phi\|P_0^\Phi).
\]
Equality holds if and only if the conditional distribution of $B$ given
$\Lambda$ is the same in both states. In the tree-family model of
Proposition~\ref{prop:tree_family}, all support points have the same persistent
Betti vector. Denote this vector by $c_\star$. Then equality holds, and
\[
    \KL(P_1^B\|P_0^B)=0,
    \qquad
    \KL(P_1^\Phi\|P_0^\Phi)
    =
    \KL(P_1^\Lambda\|P_0^\Lambda).
\]
Applying Equation~\eqref{eq:ridge_whitened_separation} to $\Phi$ and $\Lambda$
defines $\gamma_{\Phi,\rho}$ and $\gamma_{\Lambda,\rho}$. With the same ridge
value $\rho$,
\[
    \gamma_{\Phi,\rho}=\gamma_{\Lambda,\rho}.
\]
Thus, in this finite-support model, the full PL feature and the
positive-spectrum vector have the same KL divergence between states and the
same ridge-whitened separation.

\textup{(2)} To compare spectral truncation levels, let $\Lambda_J$ denote the
spectral vector formed from the first $J$ positive eigenvalues. If
$J_1\le J_2$, set $F_2=\Lambda_{J_2}$ and $F_1=\Lambda_{J_1}$. Let $g$ remove
the coordinates of $\Lambda_{J_2}$ beyond the first $J_1$ positive
eigenvalues. Lemma~\ref{lem:feature_compression_kl} gives
\[
    \KL(P_1^{\Lambda_{J_1}}\|P_0^{\Lambda_{J_1}})
    \le
    \KL(P_1^{\Lambda_{J_2}}\|P_0^{\Lambda_{J_2}}).
\]
Equality holds if and only if the conditional distribution of the removed
spectral coordinates given $\Lambda_{J_1}$ is the same in both states. This
inequality compares the population KL divergence under the same pair of state
distributions. With a finite Phase~I sample, increasing $J$ also increases the
dimension used to estimate the projection direction and covariance matrix.
Thus, plug-in detection performance need not improve monotonically with $J$.
Figure~\ref{fig:tree_sensitivity_raw} and Appendix
Figure~\ref{fig:appendix_dimension_stress} examine calibrated parameter
sensitivity and high-dimensional finite-sample performance, respectively.

\subsection{Finite-Horizon Guarantee for the Plug-In Whitened Score}
\label{sec:plugin_guarantee}

Unlike the oracle analysis in Section~\ref{sec:local_bounds_main}, this section
studies parameter-estimation error when the Phase~I sample is finite. Given
$\widehat W$, $\widehat u$, and $\widehat c$ obtained from independent Phase~I
data, let $Y_n=\widehat s_n-\widehat c$ denote the one-step PL-CUSUM increment.
Under the statewise stationarity and weak-dependence conditions in
Section~\ref{sec:problem_setup}, $(Y_n)$ is strictly stationary and
geometrically $\beta$-mixing.

Under state 0,
\[
    \widehat c-\mathbb E_0\widehat s_n
    =
    \frac{1}{2}\widehat u^\top\widehat W
    (\widehat\mu_1-\widehat\mu_0)
    +
    \widehat u^\top\widehat W(\widehat\mu_0-\mu_0),
\]
whereas under state 1,
\[
    \mathbb E_1\widehat s_n-\widehat c
    =
    \frac{1}{2}\widehat u^\top\widehat W
    (\widehat\mu_1-\widehat\mu_0)
    +
    \widehat u^\top\widehat W(\mu_1-\widehat\mu_1).
\]
The first expression is the magnitude of the mean downward drift of the
plug-in increment under state 0. The second is the mean upward drift under
state 1. Their deviations from the oracle quantities
$c-\mathbb E_0s_n$ and $\mathbb E_1s_n-c$ arise from estimation of the means
and covariance matrix. The errors $\widehat\mu_r-\mu_r$, for
$r\in\{0,1\}$, affect the estimated mean difference.
The error $\widehat\Sigma-\Sigma$ affects the whitening matrix.
Both errors therefore affect $\widehat u$, $\widehat s_n$, and $\widehat c$.

These errors can be controlled through the stability of the whitening map and
direction normalisation. The ridge rule gives
$|\widehat\rho-\rho|\leq
\rho_{\rm mult}\|\widehat\Sigma-\Sigma\|_{\mathrm{op}}$.
When $r$ is bounded away from zero, the map
$(M,r)\mapsto(M+rI_p)^{-1/2}$ satisfies a Lipschitz bound on the positive
semidefinite cone. Since $\gamma_\rho>0$, the vector
$W_\rho(\mu_1-\mu_0)$ is nonzero. The normalisation map is locally Lipschitz at
this vector. Therefore, there is a constant $C$ that depends only on $\rho$,
$\rho_{\rm mult}$, $\gamma_\rho$, and a uniform bound on the feature vectors.

Let $e$ be an upper bound satisfying
\[
    C\left(
    \|\widehat\mu_1-\mu_1\|_2
    +
    \|\widehat\mu_0-\mu_0\|_2
    +
    \|\widehat\Sigma-\Sigma\|_{\mathrm{op}}
    \right)
    \leq e.
\]
It follows that
\[
\max\left\{
\left|(\widehat c-\mathbb E_0\widehat s_n)
      -(c-\mathbb E_0s_n)\right|,
\left|(\mathbb E_1\widehat s_n-\widehat c)
      -(\mathbb E_1s_n-c)\right|
\right\}
\leq e.
\]
Since
$c-\mathbb E_0s_n=\mathbb E_1s_n-c=\gamma_\rho/2$, this error bound gives
\begin{equation}
\label{eq:plugin_drift_margin}
    \min\left\{
    \widehat c-\mathbb E_0\widehat s_n,\,
    \mathbb E_1\widehat s_n-\widehat c
    \right\}
    \geq \frac{\gamma_\rho}{2}-e.
\end{equation}
Because the feature distribution has finite support, the window features are
bounded. Suppose that the training windows are mutually independent and that
each state provides $N_{\rm tr}$ windows. Appendix
Proposition~\ref{prop:training_error_event} shows that, for any
$\delta\in(0,1)$, if $N_{\rm tr}\geq C_0\{p+\log(1/\delta)\}$, then, with
probability at least $1-\delta$,
\[
    e=O\!\left\{
    \sqrt{\frac{p+\log(1/\delta)}{N_{\rm tr}}}
    \right\}.
\]
Thus, when $e<\gamma_\rho/2$, equation~\eqref{eq:plugin_drift_margin}
ensures a negative mean increment under state 0 and a positive mean increment
under state 1.

In what follows, $\mathbb E_r\widehat s_n$ denotes the conditional expectation
of the evaluation score under state $r$, given the Phase~I estimates. Let
$\mathcal X_F$ be the finite range of the feature map, and set
$R_F=\max_{x\in\mathcal X_F}\|x\|_2$ and
$\sigma^2=R_F^2/\widehat\rho$.

Assume that there is a constant $C_{\rm dep}\geq1$ such that, for every state
$r\in\{0,1\}$, starting index $a\geq1$, interval length $m\geq1$, and
$\lambda\in\mathbb R$,
\begin{equation}
\label{eq:dependent_subgaussian_sum}
\begin{aligned}
&\mathbb E_r\!\left[
\exp\!\left\{
\lambda\sum_{j=a}^{a+m-1}
\bigl(\widehat s_j-\mathbb E_r\widehat s_j\bigr)
\right\}
\,\middle|\,
\widehat W,\widehat u,\widehat c
\right] \\
&\qquad\leq
\exp\!\left\{
\frac{C_{\rm dep}m\sigma^2\lambda^2}{2}
\right\}.
\end{aligned}
\end{equation}
The constant $C_{\rm dep}$ measures the effect of serial dependence on
partial-sum fluctuations. One may take $C_{\rm dep}=1$ when the evaluation
scores are independent. If the raw observations are independent and the
monitoring windows overlap, Appendix
Proposition~\ref{prop:overlap_decimation_main} gives $(g-1)$-dependence, and
one may take $C_{\rm dep}=g$.

\begin{proposition}
\label{prop:main_plugin_order_main}
Conditional on the Phase~I estimates, let
$\gamma_{\rm eff}:=\gamma_\rho-2e>0$, and suppose that
equations~\eqref{eq:plugin_drift_margin} and
\eqref{eq:dependent_subgaussian_sum} hold. If the control limit satisfies
\[
    \eta
    \geq
    \frac{C_{\rm dep}\sigma^2}{\gamma_{\rm eff}}
    \log\frac{N(N+1)}{2\alpha},
\]
then
\[
    \mathbb P_\infty\!\left(
    T_\eta\leq N
    \,\middle|\,
    \widehat W,\widehat u,\widehat c
    \right)
    \leq\alpha,
\]
and
\[
    \mathbb E_\nu\!\left[
    (T_\eta-\nu+1)_+
    \,\middle|\,
    \widehat W,\widehat u,\widehat c
    \right]
    \leq
    \frac{4\eta}{\gamma_{\rm eff}}
    +1+
    \frac{32C_{\rm dep}\sigma^2}{\gamma_{\rm eff}^2}.
\]
In particular, if
$\eta\asymp
C_{\rm dep}\sigma^2\log(N/\alpha)/\gamma_{\rm eff}$ and
$\log(N/\alpha)\geq1$, then
\[
    \mathbb E_\nu\!\left[
    (T_\eta-\nu+1)_+
    \,\middle|\,
    \widehat W,\widehat u,\widehat c
    \right]
    =
    O\!\left\{
    \frac{C_{\rm dep}\sigma^2\log(N/\alpha)}
    {(\gamma_\rho-2e)^2}
    \right\}.
\]
\end{proposition}
The proof is given in Appendix
Proposition~\ref{prop:main_plugin_order}.

Proposition~\ref{prop:main_plugin_order_main} shows that Phase~I estimation
error changes the effective separation to $\gamma_\rho-2e$, while temporal
dependence enters through $C_{\rm dep}$ and changes the scale of partial-sum
fluctuations. Independent, non-overlapping windows correspond to
$C_{\rm dep}=1$. If the raw observations are independent and the monitoring
windows overlap, one may take $C_{\rm dep}=g$. At the control-limit order
required for finite-horizon false-alarm control, the delay remains of order
$\log(N/\alpha)/(\gamma_\rho-2e)^2$.

\begin{corollary}
\label{cor:plugin_unconditional_main}
Assume the remaining conditions of
Proposition~\ref{prop:main_plugin_order_main}, and let
$0<\delta<\alpha<1$. Suppose that
$e<\gamma_\rho/2$ and equation~\eqref{eq:plugin_drift_margin} hold with
probability at least $1-\delta$ over the Phase~I training sample. If the
threshold is chosen as in Proposition~\ref{prop:main_plugin_order_main},
with $\alpha-\delta$ in place of $\alpha$, then
\[
    \mathbb P_\infty(T_\eta\leq N)\leq\alpha.
\]
\end{corollary}
The proof is given in Appendix
Corollary~\ref{cor:plugin_unconditional}.

\section{Experiments}
\label{sec:experiments}
\label{sec:experiment_protocol}

This section uses synthetic experiments to examine the theoretical results and
real data to compare PL-CUSUM with existing online methods. The synthetic
experiments study the separation provided by the positive PL spectrum, the
oracle delay scaling, and Phase~I estimation error. They also examine the
effects of spectral truncation $J$ and the ridge parameter. The real-data
experiments use SWaT and Electric Motor Vibrations. We compare the methods in
terms of $\mathrm{ARL}_0$, detection probability, and detection delay.

\subsection{Synthetic Experiments}
\label{sec:synthetic_experiments}

The random-tree experiments use $m=116$. The choice of tree order is reported
in Appendix Figure~\ref{fig:appendix_tree_order}. Each random support contains
64 trees. Six are the path tree, the star tree, the complete binary tree, and
three broom trees. The remaining 58 are labelled trees generated from
independently sampled uniform Pr{\"u}fer codes. By
Proposition~\ref{prop:tree_family}, their point-cloud representations have the
same persistent Betti values at the scale pair used in the experiment. The
spectral input consists of the ten largest positive eigenvalues.

On each support, $P_0$ assigns equal probability to the 64 trees. Let
$\mu_\Lambda=\mathbb E_0\Lambda(T)$, and let $\omega$ be a unit leading
eigenvector of $\operatorname{Cov}_0\{\Lambda(T)\}$. For a tilt magnitude
$\Delta>0$, set
$\tau=\Delta/\max_T|\omega^\top\{\Lambda(T)-\mu_\Lambda\}|$ and define
$P_1(T)\propto
P_0(T)\exp[\tau\omega^\top\{\Lambda(T)-\mu_\Lambda\}]$.
This construction is a direct instance of the finite-support
exponential-tilting model in Section~\ref{sec:finite_support_observation}.

For each experimental condition, $M_{\rm cal}=2000$ state-0 paths of length
500 are used to calibrate the control limit. Another 3000 state-0 paths are
used to estimate the pre-change alarm probability. For each support, 3000
change-point paths with post-change distribution $P_1$ are used to estimate
the detection probability and delay. We generate 50 independent tree
supports. The oracle scaling experiment considers 23 tilt magnitudes and
$\mathcal A\in\{100,200,400,800,1600,3200,6400\}$.
The plug-in experiment considers 13 Phase~I sample sizes, ranging from 16 to
1024 observations per state. Each sample size is repeated 100 times.

\begin{figure}[!htbp]
\centering
\includegraphics[width=\linewidth]{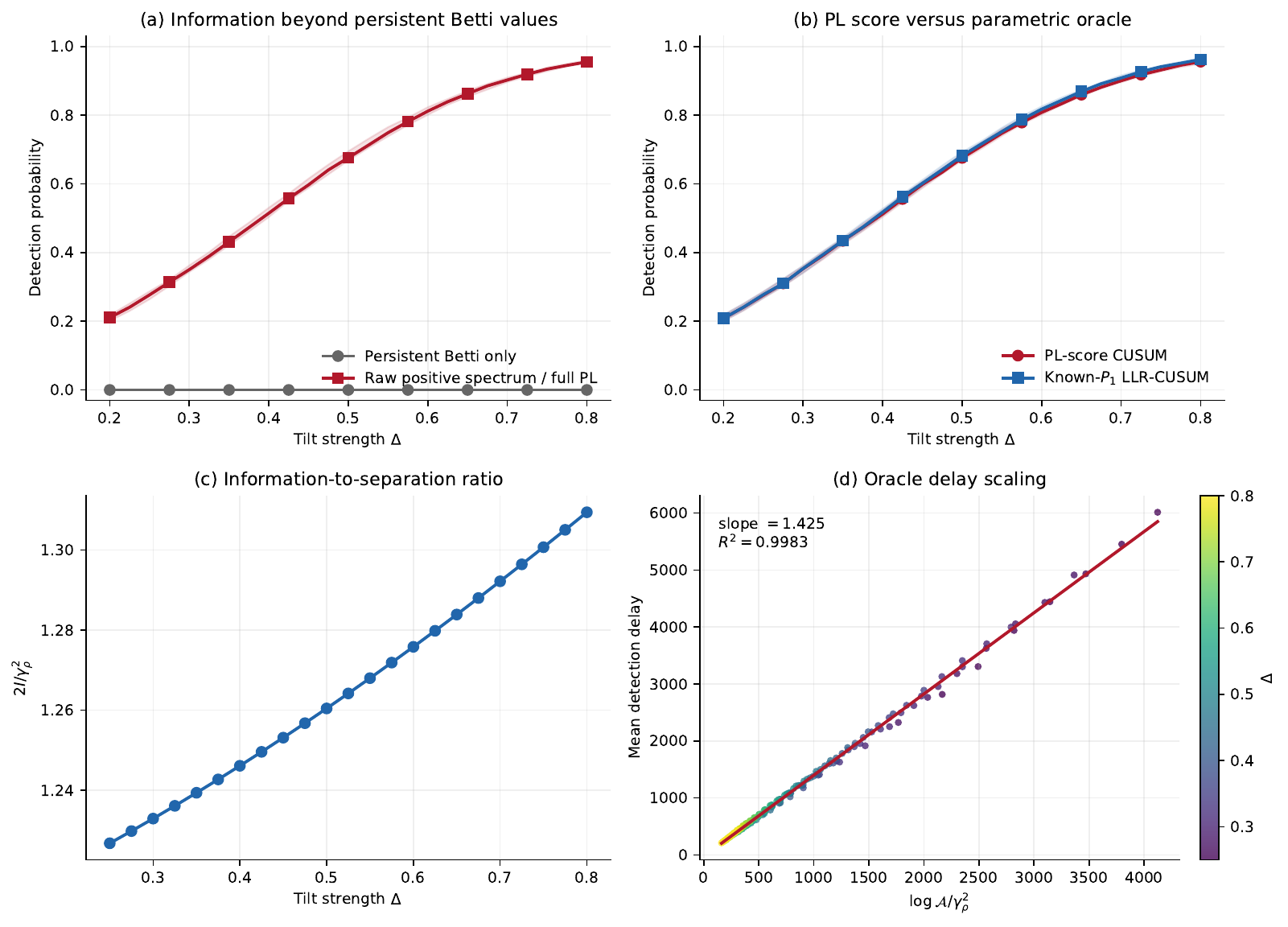}
\caption{Results from the tree-family experiments.
(a) Detection probabilities for persistent Betti features alone and for
positive PL spectral features.
(b) PL-score CUSUM and LLR-CUSUM when $P_1$ is known.
(c) The information-to-separation ratio $2I/\gamma_\rho^2$.
(d) Oracle mean delay versus $\log\mathcal A/\gamma_\rho^2$.
The shaded bands in panels (a)--(b) show interquartile ranges across
independent tree supports.}
\label{fig:tree_theory_raw}
\end{figure}

In Figure~\ref{fig:tree_theory_raw}(a), the detector based only on persistent
Betti features has detection probability zero at every tilt magnitude because
these features are constant by construction. The curves for the
positive-spectrum and full PL features coincide. The detection probability
increases from 0.211 at $\Delta=0.20$ to 0.955 at $\Delta=0.80$. The median
pre-change alarm probability ranges from 0.0480 to 0.0497. This contrast shows
that the positive PL spectrum provides the detection information in this
model.

Figure~\ref{fig:tree_theory_raw}(b) uses paired supports and paths to compare
the PL-score recursion with the parametric log-likelihood-ratio recursion. The
latter uses the known distribution $P_1$. The maximum difference between their
median detection probabilities is 0.0102. For each support and each $\Delta$,
we compute the correlation between $s(T)$ and
$\log\{P_1(T)/P_0(T)\}$ over the 64 trees in that support. At every $\Delta$,
the median correlation across the 50 supports exceeds 0.9999999998.

Figure~\ref{fig:tree_theory_raw}(c) shows that, across the 161
$(\Delta,\mathcal A)$ combinations formed from 23 tilt magnitudes,
$2I/\gamma_\rho^2$ ranges from 1.227 to 1.309 and increases slowly with
$\Delta$. This ratio remains within a stable constant range. Thus, the KL
information and ridge-whitened separation have the same local quadratic order.
Its departure from one reflects constant differences caused by finite tilt
magnitudes and ridge regularisation. For the same combinations,
Figure~\ref{fig:tree_theory_raw}(d) gives
\[
    \operatorname{Delay}
    =1.4255\frac{\log\mathcal A}{\gamma_\rho^2}-27.02,
    \qquad R^2=0.9983.
\]
This provides numerical evidence for the oracle order in
Theorem~\ref{thm:ph_pl_local_minimax_interface}. It also shows that the
proportionality constant need not equal one for finite signals.

Figure~\ref{fig:plugin_oracle_raw} next examines how the plug-in recursion
approaches the oracle recursion as the Phase~I sample size increases.

\begin{figure}[!htbp]
\centering
\includegraphics[width=\linewidth]{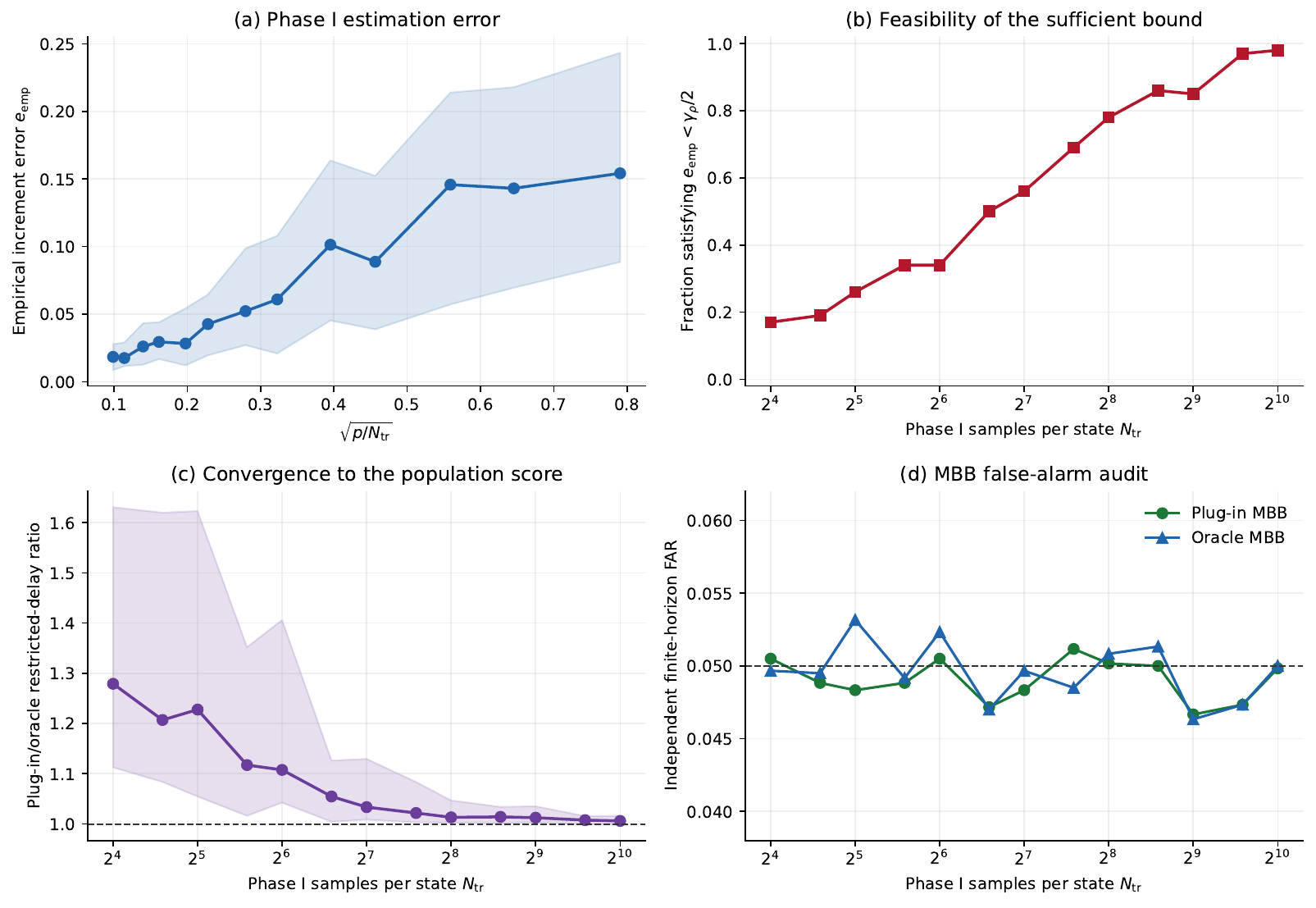}
\caption{Plug-in PL-CUSUM as the Phase~I sample size increases.
Panels (a)--(d) show the estimation error, the fraction satisfying the
finite-horizon sufficient condition, the paired plug-in-to-oracle
restricted-delay ratio, and the pre-change alarm probability under
MBB-calibrated control limits, respectively. The curves and shaded bands show
the medians and interquartile ranges across 100 Phase~I samples.}
\label{fig:plugin_oracle_raw}
\end{figure}

The median plug-in-to-oracle restricted-delay ratio decreases from 1.279 at
$N_{\rm tr}=16$ to 1.006 at $N_{\rm tr}=1024$. Its Spearman correlation with
the training sample size is approximately $-0.989$. Over the same range, the
fraction satisfying the sufficient condition $e<\gamma_\rho/2$ increases from
0.17 to 0.98. With MBB-calibrated control limits, the median pre-change alarm
probability ranges from 0.0467 to 0.0512 for the plug-in procedure and from
0.0463 to 0.0532 for the oracle procedure. The paired results show that the
delay gap narrows as the Phase~I estimation error decreases.

Finally, each sensitivity replication constructs $(P_0,P_1)$ from the full
positive spectrum. All parameter combinations within a replication are
evaluated on the same tree support and simulated paths. The spectral
truncation level is $J\in\{2,5,10,20,40,80,\text{full}\}$, and the ridge
multiplier is chosen from $\{0.125,0.5,2,8,32\}$.
The experiment covers two tilt magnitudes and 100 supports, giving 7000 paired
conditions. Figure~\ref{fig:tree_sensitivity_raw} summarizes the detection
probability and relative whitened separation.

\begin{figure}[!htbp]
\centering
\includegraphics[width=\linewidth]{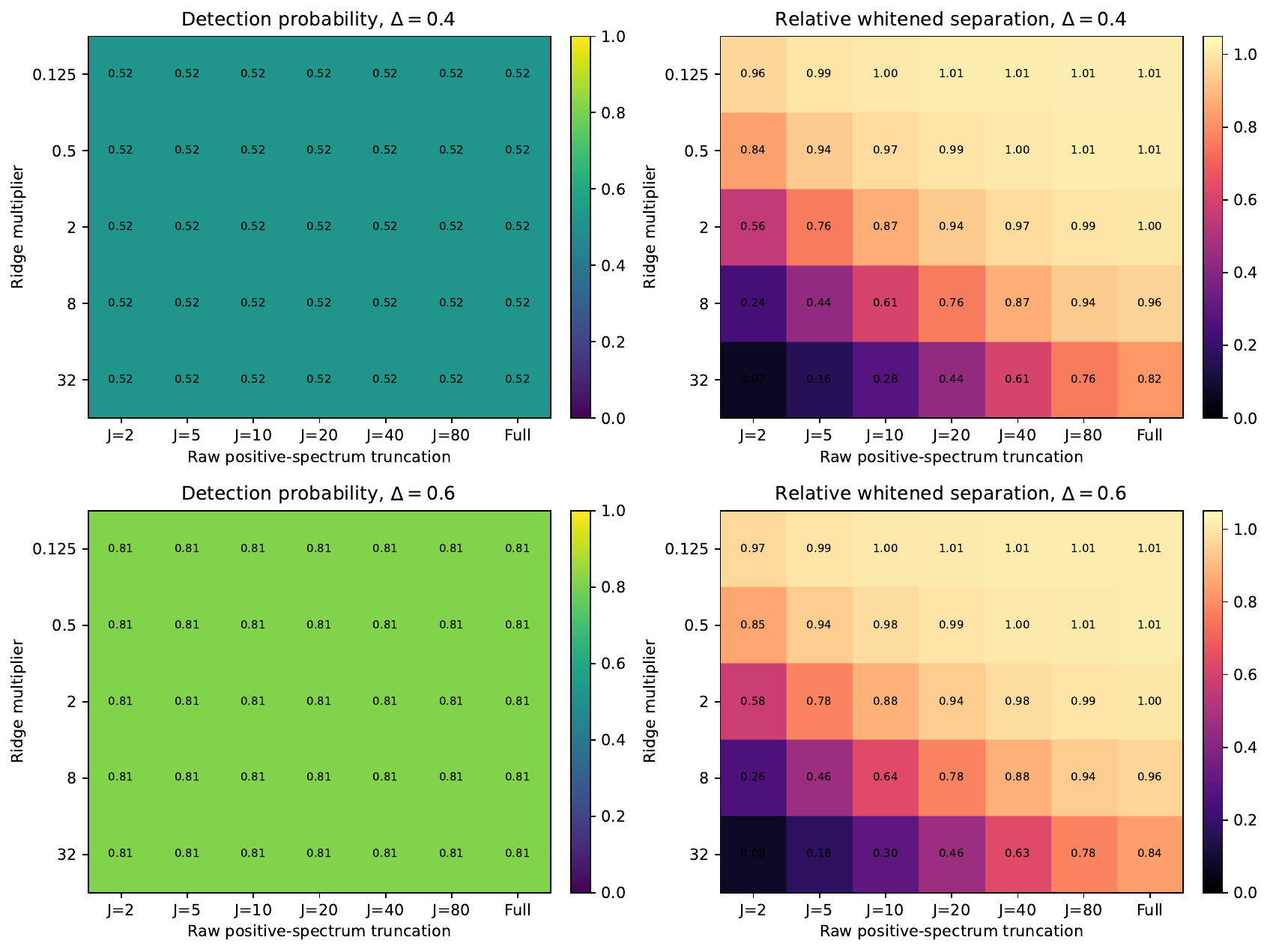}
\caption{Sensitivity to the positive-spectrum truncation level $J$ and the
relative ridge multiplier. Each row corresponds to one tilt magnitude. The
left column shows the calibrated detection probability. The right column
shows $\gamma_\rho^2/\gamma_{\rho,\mathrm{ref}}^2$, where the reference uses
the full positive spectrum and a ridge multiplier of 2. Each cell summarizes
100 paired supports.}
\label{fig:tree_sensitivity_raw}
\end{figure}

Figure~\ref{fig:tree_sensitivity_raw} shows that stronger spectral truncation
combined with a larger ridge multiplier substantially reduces the relative
whitened separation. The smallest values of
$\gamma_\rho^2/\gamma_{\rho,\mathrm{ref}}^2$ range from 0.071 to 0.079. The
corresponding control limits decrease from approximately 28--35 to 8--9.

Each parameter combination is calibrated separately. The lower feature scale
is therefore accompanied by a lower control limit at these two moderate tilt
magnitudes, and the detection probability changes little. The median detection
probability is approximately 0.52 for all combinations at $\Delta=0.4$ and
approximately 0.81 at $\Delta=0.6$. The corresponding median conditional
delays are 306 and 236--237, respectively. Across the full grid, the median
pre-change alarm probability ranges from 0.0490 to 0.0502. Thus, at the signal
strengths considered here, the detection results are stable across $J$ and
the ridge multiplier. Full calibration diagnostics and finite-sample
behaviour at higher spectral dimensions are reported in
Appendix~\ref{app:synthetic_diagnostics}.

\FloatBarrier
\subsection{Real-Data Experiments}
\label{sec:real_experiments}

The real-data experiments examine high-dimensional process monitoring and
nonlinear vibration dynamics. SWaT uses $32\times51$ windows with a stride of
16 and contains 51 synchronized process variables. Electric Motor uses
$64\times3$ windows with a stride of 64. We detect a 100~$\Omega$ electrical
fault under four operating conditions: no load, load, background vibration,
and load with background vibration. These conditions assess detection as the
load and background vibration change.

All methods are calibrated for the three target ARLs
$\mathcal A\in\{500,1000,2000\}$ using $M_{\rm cal}=2000$ state-0 paths.
After the control limits are
determined, a separate set of 2000 state-0 paths of length 200 is used to
estimate $\mathrm{ARL}_0$. Detection is evaluated using 1000 change-point
paths, with the change occurring at $\nu=100$. The data partitions and method
parameters are reported in Appendix
Tables~\ref{tab:appendix_data_protocol} and
\ref{tab:appendix_method_parameters}. If none of the 2000 state-0 paths raises
an alarm within 200 steps, the one-sided 95\% lower confidence bound is
\[
    \mathrm{ARL}_0
    >
    -\frac{2000\times200}{\log(0.05)}
    =
    133{,}523.
\]

For each change-point path, FA indicates an alarm before $\nu$. Success
indicates that the first alarm occurs within $[\nu,N]$, and Failure indicates
that no alarm occurs by the end of monitoring. EDD is the mean of $T-\nu$
conditional on Success. Thus, EDD equal to zero means that the alarm occurs at
the first post-change window.

The main comparison uses the target $\mathcal A=1000$. Complete
FA/Success/Failure results and the results for the other two ARL targets are
reported in Appendix
Tables~\ref{tab:appendix_real_arl_matrix}--%
\ref{tab:appendix_real_counts_matrix}.

\begin{table}[!htbp]
\centering
\caption{Results for SWaT and the four Electric Motor operating conditions at
$\mathcal A=1000$.
(a) Estimates of $\mathrm{ARL}_0$ from 2000 state-0 paths, where
$>133{,}523$ indicates that no alarm was observed.
(b) Success with conditional EDD in parentheses. EDD is not reported when
Success is below 0.5. Bold values indicate the highest Success within each
task.}
\label{tab:real_arl1000}

\textit{(a) Estimated $\mathrm{ARL}_0$}\par\smallskip
\resizebox{\linewidth}{!}{%
\begin{tabular}{lrrrrr}
\toprule
Method & SWaT & Motor: no load & Motor: load & Motor: background & Motor: load + background \\
\midrule
PL-CUSUM & $>133523$ & $>133523$ & $>133523$ & 18081.6 & $>133523$ \\
OK-CUSUM & 1076.8 & 1046.2 & 1017.2 & 1042.9 & 954.7 \\
Scan B & 957.5 & 963.2 & 992.5 & 977.7 & 1052.9 \\
KCUSUM & 1014.0 & 1116.5 & 848.5 & 1135.5 & 980.6 \\
NEWMA & 1135.5 & 935.4 & 1042.9 & 938.1 & 974.7 \\
PCA-CUSUM & 1073.4 & 1101.8 & 1083.9 & 1020.3 & 862.3 \\
Hotelling $T^2$ & 1039.7 & 968.9 & 1020.3 & 1039.7 & 935.4 \\
RAW-CUSUM & 1135.5 & 966.1 & 995.6 & 1020.3 & 946.3 \\
GraphScan-kNN & 1004.7 & 949.1 & 1151.0 & 998.6 & 1004.7 \\
PCA-BOCPD & 6679.2 & 1116.5 & 1020.3 & 980.6 & 995.6 \\
E-GaussianBet & 1026.7 & 989.5 & 1017.2 & 951.9 & 977.7 \\
\bottomrule
\end{tabular}
}

\medskip
\textit{(b) Success (EDD)}\par\smallskip
\resizebox{\linewidth}{!}{%
\begin{tabular}{lrrrrr}
\toprule
Method & SWaT & Motor: no load & Motor: load & Motor: background & Motor: load + background \\
\midrule
PL-CUSUM & \textbf{1.000 (0.00)} & \textbf{1.000 (0.01)} & \textbf{1.000 (0.00)} & \textbf{0.998 (0.09)} & \textbf{1.000 (0.00)} \\
OK-CUSUM & 0.914 (18.89) & 0.904 (17.33) & 0.915 (13.33) & 0.916 (21.52) & 0.909 (13.65) \\
Scan B & 0.910 (17.20) & 0.919 (23.25) & 0.933 (21.22) & 0.916 (28.24) & 0.934 (21.31) \\
KCUSUM & 0.768 (61.95) & 0.544 (50.53) & 0.634 (48.41) & 0.379 (--) & 0.644 (49.61) \\
NEWMA & 0.000 (--) & 0.000 (--) & 0.000 (--) & 0.000 (--) & 0.000 (--) \\
PCA-CUSUM & 0.000 (--) & 0.912 (0.91) & 0.904 (3.08) & 0.925 (0.01) & 0.047 (--) \\
Hotelling $T^2$ & 0.000 (--) & 0.927 (14.32) & 0.935 (5.51) & 0.934 (26.54) & 0.929 (5.94) \\
RAW-CUSUM & 0.000 (--) & 0.928 (0.44) & 0.975 (0.84) & 0.919 (0.58) & 0.926 (0.18) \\
GraphScan-kNN & 0.005 (--) & 0.233 (--) & 0.594 (13.77) & 0.357 (--) & 0.673 (12.18) \\
PCA-BOCPD & 0.000 (--) & 0.008 (--) & 0.020 (--) & 0.114 (--) & 0.031 (--) \\
E-GaussianBet & 0.000 (--) & 0.254 (--) & 0.414 (--) & 0.930 (48.11) & 0.410 (--) \\
\bottomrule
\end{tabular}
}
\end{table}

Table~\ref{tab:real_arl1000}(a) shows that, except for PCA-BOCPD on
SWaT, the estimated $\mathrm{ARL}_0$ values of the baseline methods range from
848 to 1151, close to the target of 1000. No state-0 alarms are observed for
PL-CUSUM in four of the five tasks. In the Motor background-vibration
condition, its estimated $\mathrm{ARL}_0$ is 18,082.

In Table~\ref{tab:real_arl1000}(b), PL-CUSUM achieves Success values between
0.998 and 1.000 across the five tasks, with EDD values between 0 and 0.09. On
SWaT, OK-CUSUM and Scan~B both achieve Success of approximately 0.91, with EDD
above 17. Across the Motor conditions, RAW-CUSUM achieves Success between
0.919 and 0.975, whereas PCA-CUSUM and other methods vary more substantially
across operating conditions. These results show that PL-CUSUM maintains stable
detection performance across the five tasks.

Results for the three target ARLs are shown in
Figure~\ref{fig:real_arl_tradeoff}.

\begin{figure}[!htbp]
\centering
\includegraphics[width=\linewidth]{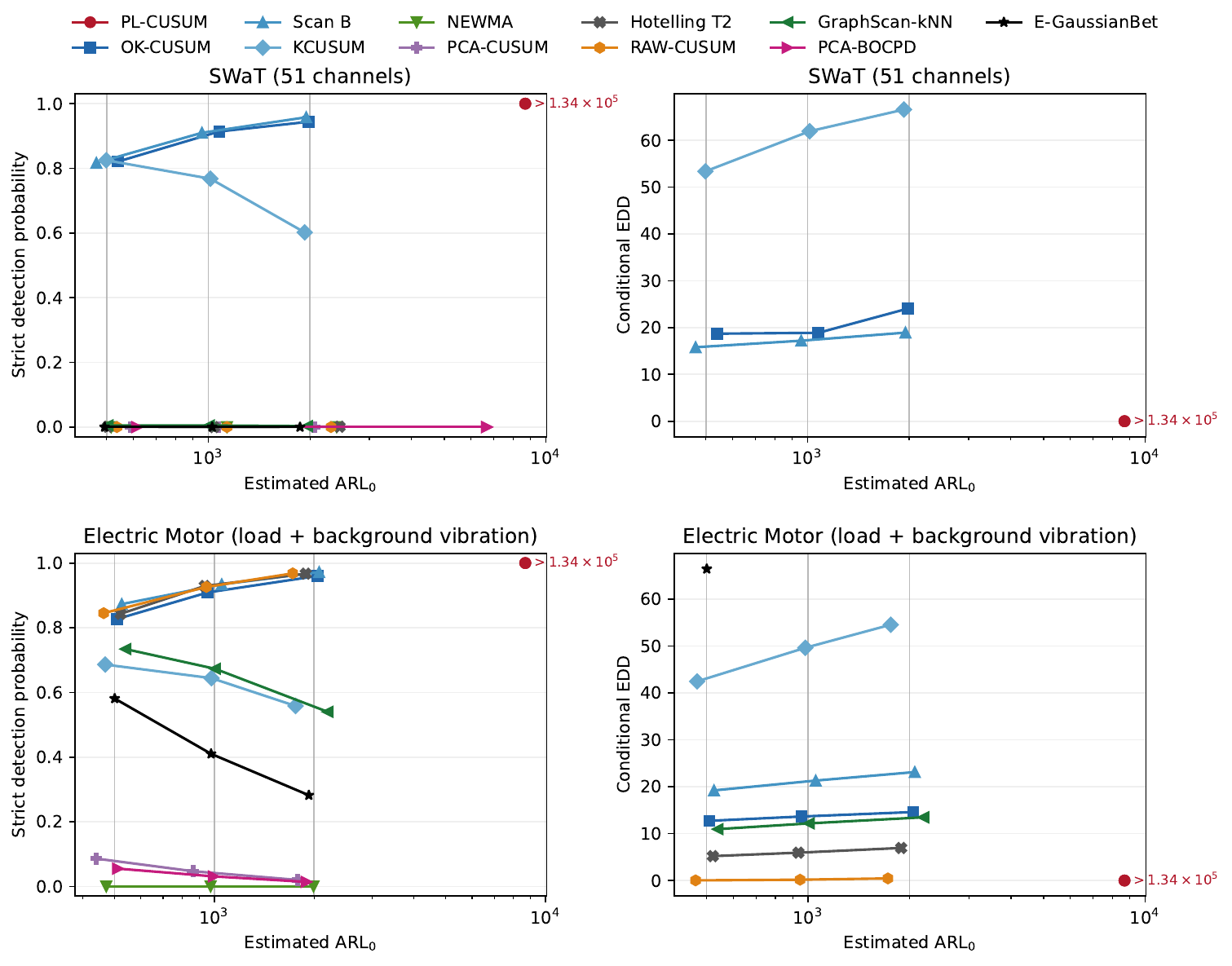}
\caption{Detection results for SWaT and Electric Motor under combined load and
background vibration at target ARLs of 500, 1000, and 2000. The horizontal
axis shows $\mathrm{ARL}_0$ estimated from separate state-0 paths. The left
column shows Success, and the right column shows conditional EDD when Success
is at least 0.5. The vertical lines indicate the target ARLs. Markers at the
right edge indicate the one-sided 95\% lower ARL bound when no state-0 alarm
is observed.}
\label{fig:real_arl_tradeoff}
\end{figure}

PL-CUSUM produces no state-0 alarms at any of the three targets for either
task. It also raises an alarm at the first post-change window. SWaT selects
$L=1$, so its point clouds describe high-dimensional cross-channel geometry.
The Motor condition selects $L=3$, which also represents delayed vibration
dynamics. Method trajectories and PL one-step increment distributions for
SWaT and all four Motor conditions are reported in
Appendix~\ref{app:real_diagnostics}.

\FloatBarrier
\section{Discussion}
\label{sec:conclusion}

PL-CUSUM converts persistent Betti features and positive persistent Laplacian
spectra into a recursively updated ridge-whitened score. This construction
allows topological information to be used for online monitoring under explicit
false-alarm and delay criteria. The oracle delay
upper bound and the local minimax lower bound are both of order
$\log\mathcal A/\gamma_\rho^2$. For finite samples, the effective separation
in the plug-in bound is $\gamma_\rho-2e$, and temporal dependence enters the
false-alarm and delay bounds through $C_{\rm dep}$. More generally, the
ridge-whitened projection and the corresponding sequential theory apply to any
Borel-measurable feature map $F$ with finite range. Persistent Betti vectors,
positive PL spectra, and their joint representation are the concrete instances
used here. These results connect the separation provided by a topological
representation with the information criteria used in classical sequential
detection \citep{lorden1971,lai1998}. The tree-family
result also shows that the positive spectrum can identify changes that
persistent Betti features cannot distinguish. The real-data experiments show
that the resulting representation can support online monitoring of
high-dimensional processes and nonlinear vibration signals.

Several questions remain open. First, the local minimax theory uses a
finite-support exponential-tilting model. For continuous weights or continuous
spectral feature vectors, the finite-support assumption could be replaced by a
local moment condition, such as finiteness of the moment-generating function
under $P_0$ in a neighbourhood of the origin. Under such conditions, upper
bounds may be developed through exponential-martingale and renewal arguments
\citep{lorden1971,tartakovsky2014}, while lower bounds may be compared with the
general information bound of \citet{lai1998}. Second, spectral truncation $J$
serves both as feature selection and as regularisation when the spectral
dimension is large relative to the Phase~I sample size. The corresponding
finite-sample results are reported in
Appendix~\ref{app:synthetic_diagnostics}.
Finally, the current PL features retain the order of observations within each
delay vector, but they are unchanged when the delay vectors within a window are
reordered. If two states have the same set of delay vectors and differ only in
their order, persistent Betti vectors and positive PL spectra cannot distinguish
them. The local model in Section~\ref{sec:finite_support_observation} describes
changes in the marginal distribution of $X_t$. It does not cover a change in
the transition law when the marginal distribution is unchanged. Ordinal
partition networks and attractor networks explicitly encode such transitions
\citep{myers2019persistent,tan2023attractor}. A future extension could construct
weighted or directed state graphs from delay vectors and monitor their spectra
after suitable PL operators have been defined. Zigzag persistence can track
topological evolution between adjacent graph snapshots
\citep{myers2023zigzag}. The first direction adds transition information within
a window. The second describes how topology persists across windows.

\section*{Data Availability}
\phantomsection
\addcontentsline{toc}{section}{Data Availability}

The SWaT and Electric Motor Vibrations data used in this study were obtained
from their public repositories. The experimental code, real-data task
configurations, tree-family simulation scripts, parameter settings, and
validation files are available in the GitHub repository \url{https://github.com/
Mousaee/pl-cusum-experiments}. Appendix~D provides extended diagnostics,
complete method parameters, and the real-data protocol. The repository result
files report the estimated ARL$_0$, FA, Success, Failure, EDD, SD, and run
status for every task, method, and target ARL. For data sets subject to
licensing or authorization restrictions, the repository lists the source and
access conditions.

\section*{Conflict of Interest}
\phantomsection
\addcontentsline{toc}{section}{Conflict of Interest}

None declared.

\section*{Acknowledgements}
\phantomsection
\addcontentsline{toc}{section}{Acknowledgements}

This work was supported by the 2024 Talent Development Fund, Tianchi Talent
Program (Second Batch), Young Doctoral Scholar Project (Bing Cai Xing [2024]
No.~97, CZ001314); the 2025 Science and Technology Development Natural Science
General Project, ``Mining and Classification of Important Nodes in Complex
Networks Based on Local Directed Homology of Directed Simplicial Sets'' (Bing
Cai Jiao [2025] No.~101, KC218501); and the 2025 Talent Development Fund,
Tianchi Talent Program (Second Batch), Young Doctoral Scholar Project (Bing Cai
Xing [2025] No.~106, CZ001327).

\bibliographystyle{abbrvnat}
\bibliography{references}

\begin{appendices}
\section{Auxiliary Theoretical Results}
\label{app:theory_proofs}

This appendix uses the finite-support exponential-tilting setup and notation
from Section~\ref{sec:finite_support_observation}. For a matrix-valued
remainder, the notation $R(\theta)=O(a_\theta)$ means that
$\|R(\theta)\|_{\mathrm{op}}=O(a_\theta)$.

\begin{lemma}
\label{lem:finite_support_exponential_tilt_main}
As $\norm{\theta}_2\to0$, the following expansions hold:
\begin{enumerate}
    \item[\textup{(a)}]
    \[
        \mu_\theta^x
        =
        \mu_0^x+\Sigma_0^x\theta+O(\norm{\theta}_2^2).
    \]
    \item[\textup{(b)}]
    \[
        \Sigma_\theta^x
        =
        \Sigma_0^x+O(\norm{\theta}_2).
    \]
    \item[\textup{(c)}]
    \[
        \KL(P_\theta\|P_0)
        =
        \frac12\,\theta^\top\Sigma_0^x\theta
        +
        O(\norm{\theta}_2^3).
    \]
    \item[\textup{(d)}]
    For each $\rho>0$,
    \[
        \gamma_\rho(\theta)
        =
        \norm{(\Sigma_0^x+\rho I_p)^{-1/2}\Sigma_0^x\theta}_2
        +
        O(\norm{\theta}_2^2).
    \]
\end{enumerate}
\end{lemma}

\begin{proof}
Since $\sum_{i=1}^M P_0(i)=1$, we have $\psi_x(0)=0$. The finite-support
assumption ensures that $\psi_x$ is analytic in a neighbourhood of the origin
and that its third derivatives are bounded there. Differentiating with respect
to $\theta$ gives
\[
    \nabla\psi_x(\theta)=\mu_\theta^x,
    \qquad
    \nabla^2\psi_x(\theta)=\Sigma_\theta^x.
\]
These identities are standard properties of finite-dimensional exponential
families; see \citet{barndorff1978} and \citet{vandervaart1998}. At
$\theta=0$,
\[
    \nabla\psi_x(0)=\mu_0^x,
    \qquad
    \nabla^2\psi_x(0)=\Sigma_0^x.
\]
Taylor expansion therefore gives
\[
\begin{aligned}
    \psi_x(\theta)
    &=
    \theta^\top\mu_0^x
    +
    \frac12\,\theta^\top\Sigma_0^x\theta
    +O(\norm{\theta}_2^3),\\
    \nabla\psi_x(\theta)
    &=
    \mu_0^x+\Sigma_0^x\theta+O(\norm{\theta}_2^2),\\
    \nabla^2\psi_x(\theta)
    &=
    \Sigma_0^x+O(\norm{\theta}_2).
\end{aligned}
\]
The last two expansions establish parts \textup{(a)} and \textup{(b)}.

By the definition of the exponential-tilting family,
\[
\begin{aligned}
    \KL(P_\theta\|P_0)
    &=
    \sum_{i=1}^M P_\theta(i)
    \log\frac{P_\theta(i)}{P_0(i)}\\
    &=
    \theta^\top\mu_\theta^x-\psi_x(\theta).
\end{aligned}
\]
Substituting the expansions of $\psi_x(\theta)$ and $\mu_\theta^x$ gives
\[
    \KL(P_\theta\|P_0)
    =
    \frac12\,\theta^\top\Sigma_0^x\theta
    +O(\norm{\theta}_2^3),
\]
which proves part \textup{(c)}.

Finally, part \textup{(b)} and the local Lipschitz property of the matrix
function $M\mapsto(M+\rho I_p)^{-1/2}$ \citep{bhatia1997} imply that
\[
    \left(
    \frac{\Sigma_0^x+\Sigma_\theta^x}{2}+\rho I_p
    \right)^{-1/2}
    =
    (\Sigma_0^x+\rho I_p)^{-1/2}
    +O(\norm{\theta}_2).
\]
Combining this expansion with part \textup{(a)} yields
\[
\begin{aligned}
    &\left(
    \frac{\Sigma_0^x+\Sigma_\theta^x}{2}+\rho I_p
    \right)^{-1/2}
    (\mu_\theta^x-\mu_0^x)\\
    &\qquad=
    (\Sigma_0^x+\rho I_p)^{-1/2}\Sigma_0^x\theta
    +O(\norm{\theta}_2^2).
\end{aligned}
\]
Taking Euclidean norms on both sides and applying the reverse triangle
inequality proves part \textup{(d)}.
\end{proof}

\begin{corollary}
\label{cor:key_local_relations}
Suppose $h^\top G_\rho h>0$. Then, as $n\to\infty$,
\[
    I_{h,n}^x
    =
    \frac{1}{2n}h^\top\Sigma_0^x h+O(n^{-3/2}),
    \qquad
    \gamma_\rho^2(\theta_n)
    =
    \frac{1}{n}h^\top G_\rho h+O(n^{-3/2}).
\]
Consequently,
\[
    I_{h,n}^x
    =
    \frac{\kappa_\rho(h)}2
    \gamma_\rho^2(\theta_n)\{1+o(1)\},
    \qquad
    \kappa_\rho(h)
    =
    \frac{h^\top\Sigma_0^x h}{h^\top G_\rho h}.
\]
Moreover,
\[
    1
    \leq
    \kappa_\rho(h)
    \leq
    1+\frac{\rho}{\lambda_{\min}^+(\Sigma_0^x)},
\]
where $\lambda_{\min}^+(\Sigma_0^x)$ is the smallest positive eigenvalue of
$\Sigma_0^x$.
\end{corollary}

\begin{proof}
Apply Lemma~\ref{lem:finite_support_exponential_tilt_main} with
$\theta_n=h/\sqrt n$. The KL expansion gives
\[
    I_{h,n}^x
    =
    \KL(P_{\theta_n}\|P_0)
    =
    \frac{1}{2n}h^\top\Sigma_0^x h+O(n^{-3/2}).
\]
Set
\[
    A_\rho=(\Sigma_0^x+\rho I_p)^{-1/2}\Sigma_0^x .
\]
By Lemma~\ref{lem:finite_support_exponential_tilt_main}\textup{(d)},
\[
    \gamma_\rho(\theta_n)
    =
    \norm{A_\rho\theta_n}_2+O(\norm{\theta_n}_2^2)
    =
    n^{-1/2}\norm{A_\rho h}_2+O(n^{-1}).
\]
Since $h^\top G_\rho h>0$, we have $\norm{A_\rho h}_2>0$. Squaring both
sides yields
\[
    \gamma_\rho^2(\theta_n)
    =
    \frac1n\norm{A_\rho h}_2^2+O(n^{-3/2}).
\]
Moreover,
\[
    \norm{A_\rho h}_2^2
    =
    h^\top\Sigma_0^x(\Sigma_0^x+\rho I_p)^{-1}\Sigma_0^x h
    =
    h^\top G_\rho h .
\]
This proves the expansion for the ridge-whitened separation. Since
$\Sigma_0^x$ is positive semidefinite, $h^\top G_\rho h>0$ also implies
$h^\top\Sigma_0^x h>0$. Dividing the two expansions gives
\[
    \frac{I_{h,n}^x}{\gamma_\rho^2(\theta_n)}
    =
    \frac{h^\top\Sigma_0^x h}{2h^\top G_\rho h}
    \{1+O(n^{-1/2})\}
    =
    \frac{\kappa_\rho(h)}2\{1+o(1)\}.
\]

Finally, for every positive eigenvalue $\lambda$ of $\Sigma_0^x$,
\[
    \frac{\lambda^2}{\lambda+\rho}
    \leq
    \lambda
    \leq
    \left\{1+\frac{\rho}{\lambda_{\min}^+(\Sigma_0^x)}\right\}
    \frac{\lambda^2}{\lambda+\rho}.
\]
The spectral decomposition of $\Sigma_0^x$ therefore gives
\[
    h^\top G_\rho h
    \leq
    h^\top\Sigma_0^x h
    \leq
    \left\{1+\frac{\rho}{\lambda_{\min}^+(\Sigma_0^x)}\right\}
    h^\top G_\rho h.
\]
Dividing by $h^\top G_\rho h>0$ proves the stated bounds on
$\kappa_\rho(h)$.
\end{proof}

The following lemma gives a lower bound on the no-change average run length of
the oracle CUSUM.

\begin{lemma}
\label{lem:cusum_arl}
Suppose that $\{\ell_t\}_{t\geq1}$ are independent and identically distributed
under $\mathbb P_\infty$ and satisfy
$\mathbb E_\infty e^{\ell_t}=1$. For any $\eta>0$, let $\tau_\eta$ denote
Page's CUSUM stopping time with increments $\ell_t$ and threshold $\eta$. Then
\[
    \mathbb E_\infty\tau_\eta\geq e^\eta.
\]
\end{lemma}

\begin{proof}
Let $\mathscr F_t^\ell=\sigma(\ell_1,\ldots,\ell_t)$ and define
\[
    V_0=0,\qquad
    V_t=(1+V_{t-1})e^{\ell_t}
    =
    \sum_{k=1}^t
    \exp\left\{\sum_{j=k}^t\ell_j\right\}.
\]
By independence of the increments and the identity
$\mathbb E_\infty e^{\ell_t}=1$,
\[
    \mathbb E_\infty(V_t\mid\mathscr F_{t-1}^\ell)
    =
    1+V_{t-1}.
\]
Hence $\{V_t-t\}_{t\geq0}$ is a martingale.

Page's CUSUM statistic can also be written as
\[
    R_t
    =
    \max\left\{
    0,\,
    \max_{1\leq k\leq t}\sum_{j=k}^t\ell_j
    \right\}.
\]
Therefore, on the event $\{\tau_\eta<\infty\}$,
\[
    V_{\tau_\eta}
    \geq e^{R_{\tau_\eta}}
    \geq e^\eta.
\]
For any positive integer $m$, applying the optional stopping theorem to the
bounded stopping time $\tau_\eta\wedge m$ gives
\[
    \mathbb E_\infty(\tau_\eta\wedge m)
    =
    \mathbb E_\infty V_{\tau_\eta\wedge m}
    \geq
    e^\eta\mathbb P_\infty(\tau_\eta\leq m).
\]
If $\mathbb P_\infty(\tau_\eta<\infty)=1$, letting $m\to\infty$ yields
$\mathbb E_\infty\tau_\eta\geq e^\eta$. If
$\mathbb P_\infty(\tau_\eta<\infty)<1$, then
$\mathbb E_\infty\tau_\eta=\infty$, and the conclusion again follows.
\end{proof}

When both $P_0$ and $P_{\theta_n}$ are known, the LLR-CUSUM serves as the
oracle benchmark used below. CUSUM was introduced by \citet{page1954}.
\citet{lorden1971} formulated the worst-case delay criterion, and
\citet{moustakides1986} proved the optimality of LLR-CUSUM under this
criterion.

\begin{proposition}
\label{prop:pl_cusum_upper_lan_main}
For any $\eta>0$, let $\tau_\eta$ be the oracle LLR-CUSUM stopping time for
known $P_0$ and $P_{\theta_n}$. Then
\[
    \operatorname{ARL}_0(\tau_\eta)\geq e^\eta,
    \qquad
    \bar d_{\theta_n}(\tau_\eta)
    \leq
    \frac{\eta+O(n^{-1/2})}{I_{h,n}^x}.
\]
In particular, if $\eta=\log\mathcal A_n$ and $\mathcal A_n\to\infty$, then
\[
    \bar d_{\theta_n}(\tau_{\log\mathcal A_n})
    \leq
    \frac{2\log\mathcal A_n}{\gamma_\rho^2(\theta_n)}
    \{1+o(1)\}.
\]
\end{proposition}

\begin{proof}
Write
\[
    \ell_{h,n}(i)=\log\{P_{\theta_n}(i)/P_0(i)\},
    \qquad R_0=0,
\]
and define
\[
    R_t=\max\{0,R_{t-1}+\ell_{h,n}(U_t)\},
    \qquad
    \tau_\eta=\inf\{t\ge1:R_t\ge\eta\}.
\]

\pfstep{ARL bound}
Under the no-change distribution, $U_s$ has distribution $P_0$, and hence
\[
    \mathbb E_\infty e^{\ell_{h,n}(U_s)}
    =
    \sum_i P_0(i)\frac{P_{\theta_n}(i)}{P_0(i)}
    =1.
\]
Lemma~\ref{lem:cusum_arl} therefore gives
$\mathbb E_\infty\tau_\eta\geq e^\eta$.

\pfstep{First-passage bound}
By the definition of the exponential tilting family,
\[
    \ell_{h,n}(i)
    =
    \theta_n^\top x_i-\psi_x(\theta_n).
\]
Since the support is finite and $\theta_n=h/\sqrt n$, there exists a constant
$C<\infty$, independent of $n$, such that
\[
    \max_i|\ell_{h,n}(i)|\leq Cn^{-1/2}.
\]
Under the post-change distribution, let $Y_t=\ell_{h,n}(U_t)$. The sequence
$\{Y_t\}$ is independent and identically distributed, with
\[
    \mathbb E_{\theta_n}Y_t
    =
    \sum_iP_{\theta_n}(i)
    \log\frac{P_{\theta_n}(i)}{P_0(i)}
    =
    I_{h,n}^x>0.
\]
Define the first-passage time
\[
    \sigma_\eta
    =
    \inf\left\{r\ge1:\sum_{j=1}^r Y_j\ge\eta\right\}.
\]
Because $Y_t$ is bounded and has positive mean,
$\mathbb E_{\theta_n}\sigma_\eta<\infty$, so Wald's identity applies
\citep{wald1947}. By the definition of $\sigma_\eta$,
\[
    0\leq
    \sum_{j=1}^{\sigma_\eta}Y_j-\eta
    \leq Cn^{-1/2}.
\]
It follows that
\[
    I_{h,n}^x\mathbb E_{\theta_n}\sigma_\eta
    =
    \mathbb E_{\theta_n}\sum_{j=1}^{\sigma_\eta}Y_j
    \leq \eta+Cn^{-1/2},
\]
and therefore
\[
    \mathbb E_{\theta_n}\sigma_\eta
    \leq
    \frac{\eta+O(n^{-1/2})}{I_{h,n}^x}.
\]

\pfstep{Lorden delay bound}
For any change point $\nu$, define
\[
    \sigma_\eta^{(\nu)}
    =
    \inf\left\{r\ge1:
    \sum_{j=0}^{r-1}\ell_{h,n}(U_{\nu+j})\ge\eta
    \right\}.
\]
The CUSUM recursion implies that, for every $r\geq1$,
\[
    R_{\nu+r-1}
    \geq
    \sum_{j=0}^{r-1}\ell_{h,n}(U_{\nu+j}).
\]
Consequently,
\[
    (\tau_\eta-\nu+1)_+\leq\sigma_\eta^{(\nu)}.
\]
Under $\mathbb P_{\nu,\theta_n}$, conditional on $\mathscr F_{\nu-1}^X$,
$\sigma_\eta^{(\nu)}$ has the same distribution as $\sigma_\eta$. Hence,
almost surely,
\[
    \mathbb E_{\nu,\theta_n}
    \left\{(\tau_\eta-\nu+1)_+\mid\mathscr F_{\nu-1}^X\right\}
    \leq
    \mathbb E_{\theta_n}\sigma_\eta.
\]
Taking the supremum over $\nu$ and the essential supremum over
$\mathscr F_{\nu-1}^X$ gives
\[
    \bar d_{\theta_n}(\tau_\eta)
    \leq
    \frac{\eta+O(n^{-1/2})}{I_{h,n}^x}.
\]

\pfstep{Separation form}
If $\eta=\log\mathcal A_n$ and $\mathcal A_n\to\infty$, then
$O(n^{-1/2})=o(\log\mathcal A_n)$. Corollary~\ref{cor:key_local_relations}
and $\kappa_\rho(h)\geq1$ therefore yield
\[
\begin{aligned}
    \bar d_{\theta_n}(\tau_{\log\mathcal A_n})
    &\leq
    \frac{2\log\mathcal A_n}
    {\kappa_\rho(h)\gamma_\rho^2(\theta_n)}
    \{1+o(1)\}\\
    &\leq
    \frac{2\log\mathcal A_n}{\gamma_\rho^2(\theta_n)}
    \{1+o(1)\}.
\end{aligned}
\]
\end{proof}

\begin{lemma}
\label{lem:low_alarm_block_main}
Let $\mathcal A>0$, let $(\mathscr F_t)_{t\geq0}$ be a filtration, and let
$\tau$ be an $(\mathscr F_t)$-stopping time satisfying
$\mathbb E_\infty\tau\geq\mathcal A$.

Given $L\in\mathbb N$ and $c_0>1$, suppose that
$c_0L/\mathcal A<1$. Then there exists $j\geq0$ such that
\[
\begin{gathered}
    H_j=\{\tau>jL\}\in\mathscr F_{jL},
    \qquad
    \mathbb P_\infty(H_j)>0,\\
    \mathbb P_\infty
    \bigl(\tau\leq(j+1)L\mid H_j\bigr)
    \leq
    c_0\frac{L}{\mathcal A}.
\end{gathered}
\]
\end{lemma}

\begin{proof}
Suppose that the conclusion fails. Then, for every $j$ satisfying
$\mathbb P_\infty(H_j)>0$,
\[
    \mathbb P_\infty
    \bigl(\tau\leq(j+1)L\mid H_j\bigr)
    >
    c_0\frac{L}{\mathcal A}.
\]
Since $H_{j+1}\subset H_j$, this is equivalent to
\[
    \mathbb P_\infty(H_{j+1}\mid H_j)
    <
    1-c_0\frac{L}{\mathcal A}.
\]
If $\mathbb P_\infty(H_j)=0$, then
$\mathbb P_\infty(H_{j+1})=0$. Since
$\mathbb P_\infty(H_0)\leq1$, induction gives
\[
    \mathbb P_\infty(H_j)
    \leq
    \left(1-c_0\frac{L}{\mathcal A}\right)^j,
    \qquad j\geq0.
\]
Using the tail-sum formula for nonnegative integer-valued random variables,
\[
\begin{aligned}
    \mathbb E_\infty\tau
    &=
    \sum_{t\geq0}\mathbb P_\infty(\tau>t)\\
    &\leq
    L\sum_{j\geq0}\mathbb P_\infty(\tau>jL)\\
    &=
    L\sum_{j\geq0}\mathbb P_\infty(H_j)\\
    &\leq
    L\sum_{j\geq0}
    \left(1-c_0\frac{L}{\mathcal A}\right)^j\\
    &=
    \frac{\mathcal A}{c_0}
    <
    \mathcal A.
\end{aligned}
\]
This contradicts $\mathbb E_\infty\tau\geq\mathcal A$.
\end{proof}

\begin{proposition}
\label{prop:local_minimax_lower_main}
Suppose that $n_k\to\infty$, $\mathcal A_k\to\infty$, and
\[
    \frac{\log\mathcal A_k}{I_{h,n_k}^x}=o(\mathcal A_k),
    \qquad
    \log n_k=o(\log\mathcal A_k).
\]
Then
\[
    R_{n_k,\mathcal A_k}(h)
    \geq
    \frac{\log\mathcal A_k}{I_{h,n_k}^x}\{1+o(1)\}
    =
    \frac{2}{\kappa_\rho(h)}
    \frac{\log\mathcal A_k}{\gamma_\rho^2(\theta_{n_k})}
    \{1+o(1)\}.
\]
\end{proposition}

\begin{proof}
Write $I_k:=I_{h,n_k}^x$, $\theta_k:=h/\sqrt{n_k}$, and
$g_k(\cdot):=\ell_{h,n_k}(\cdot)$. Let
$\tau_k\in\mathcal C_{\mathcal A_k}$ be arbitrary, and write it as $\tau$
below.

\pfstep{Local asymptotic order}
Since $\Sigma_0^x$ is positive semidefinite, $h^\top G_\rho h>0$ implies
$h\notin\ker\Sigma_0^x$. Hence $h^\top\Sigma_0^x h>0$. By
Corollary~\ref{cor:key_local_relations},
\[
    I_k
    =
    \frac{1}{2n_k}h^\top\Sigma_0^x h+O(n_k^{-3/2}),
    \qquad
    I_k\asymp n_k^{-1}.
\]
Moreover, $g_k(i)=\theta_k^\top x_i-\psi_x(\theta_k)$ and
$\|\theta_k\|_2=\|h\|_2/\sqrt{n_k}$. Therefore,
\[
    \max_i|g_k(i)|=O(n_k^{-1/2}),
\]
and
\[
    \operatorname{Var}_{\theta_k}\{g_k(U)\}
    =
    \theta_k^\top\Sigma_{\theta_k}^x\theta_k
    =
    O(n_k^{-1}).
\]

\pfstep{Block length}
Let $\epsilon\in(0,1)$ and set
\[
    \ell_k
    =
    \left\lfloor
    (1-\epsilon)\frac{\log\mathcal A_k}{I_k}
    \right\rfloor .
\]
Then
\[
    \ell_k I_k
    =
    (1-\epsilon)\log\mathcal A_k+O(I_k)
    =
    (1-\epsilon)\log\mathcal A_k\{1+o(1)\},
\]
and
\[
    \ell_k\asymp n_k\log\mathcal A_k\to\infty,
    \qquad
    \ell_k=o(\mathcal A_k).
\]
The last relation follows from
$\log\mathcal A_k/I_k=o(\mathcal A_k)$.

\pfstep{Low-alarm-probability block}
Choose $c_0>1$. Since $\ell_k=o(\mathcal A_k)$,
$c_0\ell_k/\mathcal A_k<1$ for all sufficiently large $k$. By
Lemma~\ref{lem:low_alarm_block_main}, there exists $j_k\geq0$ such that
$s_k=j_k\ell_k$ and
\[
    H_k=\{\tau>s_k\}\in\mathscr F_{s_k}^X
\]
satisfy $\mathbb P_\infty(H_k)>0$ and
\[
    \mathbb P_\infty(\tau\leq s_k+\ell_k\mid H_k)
    \leq
    c_0\frac{\ell_k}{\mathcal A_k}.
\]
Set the change point to $\nu_k=s_k+1$. On $H_k$, no alarm occurs before the
change.

\pfstep{Likelihood-ratio concentration}
Let
\[
    \Lambda_k
    =
    \sum_{t=\nu_k}^{\nu_k+\ell_k-1}g_k(U_t).
\]
Since $H_k\in\mathscr F_{s_k}^X$, conditional on $H_k$, the block variables are
independent and identically distributed according to $P_{\theta_k}$ under
$\mathbb P_{\nu_k,\theta_k}$. Therefore,
\[
    \mathbb E_{\nu_k,\theta_k}(\Lambda_k\mid H_k)
    =
    \ell_k I_k,
\]
and
\[
    \operatorname{Var}_{\nu_k,\theta_k}(\Lambda_k\mid H_k)
    =
    \ell_k O(n_k^{-1})
    =
    O(\log\mathcal A_k).
\]
Together with
$\ell_kI_k=(1-\epsilon)\log\mathcal A_k\{1+o(1)\}$, this gives
\[
    \frac{\operatorname{Var}_{\nu_k,\theta_k}(\Lambda_k\mid H_k)}
    {(\ell_k I_k)^2}
    =
    O\!\left(\frac{1}{\log\mathcal A_k}\right)
    \to0.
\]
For each $\delta>0$, define
\[
    E_k:=\{\Lambda_k\leq(1+\delta)\ell_k I_k\}.
\]
Chebyshev's inequality gives
\[
    \mathbb P_{\nu_k,\theta_k}(E_k^c\mid H_k)\to0.
\]

\pfstep{Conditional change of measure}
Define the block alarm event
\[
    \mathcal D_k
    =
    \{s_k<\tau\leq s_k+\ell_k\}
    \in
    \mathscr F_{s_k+\ell_k}^X.
\]
Since $H_k=\{\tau>s_k\}$,
\[
    \mathbb P_\infty(\mathcal D_k\mid H_k)
    =
    \mathbb P_\infty(\tau\leq s_k+\ell_k\mid H_k)
    \leq
    c_0\frac{\ell_k}{\mathcal A_k}.
\]
Moreover, $H_k$ depends only on pre-change observations. Hence
\[
    \mathbb P_\infty(H_k)
    =
    \mathbb P_{\nu_k,\theta_k}(H_k).
\]
On $\mathscr F_{s_k+\ell_k}^X$, $e^{\Lambda_k}$ is the likelihood ratio of the
post-change measure with respect to the no-change measure. Therefore,
\[
\begin{aligned}
    \mathbb P_\infty(\mathcal D_k\mid H_k)
    &=
    \mathbb E_{\nu_k,\theta_k}
    \left[e^{-\Lambda_k}\mathbf 1_{\mathcal D_k}\mid H_k\right]\\
    &\geq
    e^{-(1+\delta)\ell_k I_k}
    \left\{
    \mathbb P_{\nu_k,\theta_k}(\mathcal D_k\mid H_k)
    -
    \mathbb P_{\nu_k,\theta_k}(E_k^c\mid H_k)
    \right\}.
\end{aligned}
\]
Consequently,
\[
\begin{aligned}
    \mathbb P_{\nu_k,\theta_k}(\mathcal D_k\mid H_k)
    &\leq
    \mathbb P_{\nu_k,\theta_k}(E_k^c\mid H_k)\\
    &\quad+
    c_0\frac{\ell_k}{\mathcal A_k}
    \exp\{(1+\delta)\ell_k I_k\}.
\end{aligned}
\]

\pfstep{No-alarm probability under the ARL constraint}
Choose $\delta>0$ such that
$\beta:=(1-\epsilon)(1+\delta)<1$. The block-length relation gives
\[
    \exp\{(1+\delta)\ell_k I_k\}
    =
    \mathcal A_k^{\beta+o(1)}.
\]
Since $\ell_k=O(n_k\log\mathcal A_k)$,
\[
    \log\ell_k
    =
    O\{\log n_k+\log\log\mathcal A_k\}
    =
    o(\log\mathcal A_k).
\]
Hence $\ell_k=\mathcal A_k^{o(1)}$, and
\[
    \frac{\ell_k}{\mathcal A_k}
    \mathcal A_k^{\beta+o(1)}
    =
    \mathcal A_k^{\beta-1+o(1)}
    \to0.
\]
Combining this result with the likelihood-ratio concentration bound yields
\[
    \mathbb P_{\nu_k,\theta_k}(\mathcal D_k\mid H_k)\to0.
\]
Equivalently,
\[
    \mathbb P_{\nu_k,\theta_k}
    (\tau>s_k+\ell_k\mid H_k)
    \to1.
\]

\pfstep{Delay lower bound}
Since $H_k\in\mathscr F_{\nu_k-1}^X$ and
$\mathbb P_{\nu_k,\theta_k}(H_k)>0$, Lorden's criterion gives
\[
\begin{aligned}
    \bar d_{\theta_k}(\tau)
    &\geq
    \mathbb E_{\nu_k,\theta_k}
    \bigl[(\tau-\nu_k+1)_+\mid H_k\bigr]\\
    &\geq
    \ell_k\,\mathbb P_{\nu_k,\theta_k}
    (\tau>s_k+\ell_k\mid H_k)\\
    &\geq
    \ell_k\{1+o(1)\}
    =
    (1-\epsilon)\frac{\log\mathcal A_k}{I_k}\{1+o(1)\}.
\end{aligned}
\]
The preceding remainder terms do not depend on the chosen stopping-time
sequence. Therefore, take the infimum over
$\tau_k\in\mathcal C_{\mathcal A_k}$. Letting $k\to\infty$ and then
$\epsilon\downarrow0$ gives
\[
    R_{n_k,\mathcal A_k}(h)
    \geq
    \frac{\log\mathcal A_k}{I_k}\{1+o(1)\}.
\]
Finally, Corollary~\ref{cor:key_local_relations} gives
\[
    I_k
    =
    \frac{\kappa_\rho(h)}{2}
    \gamma_\rho^2(\theta_{n_k})\{1+o(1)\}.
\]
Hence
\[
    R_{n_k,\mathcal A_k}(h)
    \geq
    \frac{2}{\kappa_\rho(h)}
    \frac{\log\mathcal A_k}
    {\gamma_\rho^2(\theta_{n_k})}
    \{1+o(1)\}.
\]
\end{proof}

\begin{proposition}
\label{prop:finite_image_pl_main}
Given the scale grid $0<\varepsilon_1<\cdots<\varepsilon_K$, homological
truncation level $q_{\max}$, finite scale-pair set $\mathcal S$, and spectral
truncation level $J$ from Section~\ref{sec:pl_spectral_features}, the maps
$B^{a,b}$, $\Lambda^{a,b}$, and $\Phi^{a,b}$ are Borel measurable and have
finite ranges for every $(a,b)\in\mathcal S$.
\end{proposition}

\begin{proof}
For a point cloud $Z=(y_1,\ldots,y_m)$, define the finite indicator vector
\[
    C(Z)
    =
    \left(
    \mathbf 1\{\|y_r-y_s\|_2\leq\varepsilon_k\}
    \right)_{1\leq r<s\leq m,\,1\leq k\leq K}
    \in\{0,1\}^{K\binom m2}.
\]
For each $k$, $C(Z)$ specifies every edge of
$\VR_{\varepsilon_k}(Z)$. A vertex set is a Vietoris--Rips simplex if and only
if every pair of its vertices is joined by an edge. Thus, $C(Z)$ determines all
complexes on the scale grid. It also determines the corresponding boundary
matrices, persistent combinatorial Laplacians, persistent Betti vectors, and
positive-spectrum vectors. Since $C(Z)$ takes only finitely many values, the
ranges of $B^{a,b}$, $\Lambda^{a,b}$, and $\Phi^{a,b}$ are finite.

For every $r<s$ and $k$, the distance map
$Z\mapsto\|y_r-y_s\|_2$ is continuous. Hence
$\{Z:\|y_r-y_s\|_2\leq\varepsilon_k\}$ is closed, and $C$ is Borel
measurable. Each of the three feature maps is the composition of $C$ with a
map on a finite set and is therefore Borel measurable.
\end{proof}

Consequently, a feature map $F$ selected by Phase~I from a finite configuration
set has finite range. The distribution of $F(Z)$ therefore satisfies the
finite-support assumption in Section~\ref{sec:finite_support_observation}.

For state $k\in\{0,1\}$, suppose that the raw observation process
$(x_t)_{t\geq1}$ is strictly stationary under $\mathbb P_k$. Define its
$\beta$-mixing coefficients by
\[
    \beta_{x,k}(r)
    =
    \sup_{t\geq1}
    \beta\!\left(
        \sigma(x_1,\ldots,x_t),
        \sigma(x_{t+r},x_{t+r+1},\ldots)
    \right),
    \qquad r\geq1,
\]
where $\beta(\mathcal A,\mathcal B)$ is the $\beta$-mixing coefficient between
two $\sigma$-fields. Beta-mixing is also called absolute regularity
\citep{bradley2005}, and we set $\beta_{x,k}(0)=1$. Given the Phase~I
estimates, $Z_n$, $X_n=F(Z_n)$, and
$Y_n=\widehat u^\top\widehat W X_n-\widehat c$ are Borel-measurable functions
of the $n$th observation window of length $w$.

\begin{proposition}
\label{prop:overlap_decimation_main}
In state $k\in\{0,1\}$, suppose that $(x_t)_{t\geq1}$ is strictly stationary.
For any Borel-measurable function $H$, let
\[
    V_n
    =
    H\!\left(
        x_{(n-1)h_{\rm stride}+1},
        \ldots,
        x_{(n-1)h_{\rm stride}+w}
    \right),
\]
and let $\beta_{V,k}$ denote its $\beta$-mixing coefficients.

\textup{(1)} The sequence $(V_n)$ is strictly stationary and, for every
$r\geq1$,
\[
    \beta_{V,k}(r)
    \leq
    \beta_{x,k}\!\left((r h_{\rm stride}-w+1)_+\right).
\]
Hence, if the raw process is geometrically $\beta$-mixing, then $(V_n)$ is
also geometrically $\beta$-mixing.

\textup{(2)} If the raw process is $q$-dependent, then $(V_n)$ is
\[
    \left(
        \left\lceil\frac{w+q}{h_{\rm stride}}\right\rceil-1
    \right)\text{-dependent}.
\]
In particular, if the raw observations are independent, then $(V_n)$ is
$(g-1)$-dependent, where $g=\lceil w/h_{\rm stride}\rceil$.
\end{proposition}

\begin{proof}
$V_n$ is a measurable function of the $n$th observation window of length $w$.
Thus, strict stationarity of the raw process implies strict stationarity of
$(V_n)$.

For any $t,r\geq1$, the first $t$ windows use raw observations only up to index
$(t-1)h_{\rm stride}+w$. The windows from $t+r$ onward start at index
$(t+r-1)h_{\rm stride}+1$. Monotonicity of the $\beta$-mixing coefficient
under restriction to smaller $\sigma$-fields gives
\[
    \beta_{V,k}(r)
    \leq
    \beta_{x,k}\!\left((r h_{\rm stride}-w+1)_+\right).
\]
The geometric $\beta$-mixing result follows directly.

If the raw process is $q$-dependent and
$r\geq\lceil(w+q)/h_{\rm stride}\rceil$, then
$r h_{\rm stride}-w+1\geq q+1$. The corresponding window $\sigma$-fields are
therefore independent. Hence, $(V_n)$ is
$\bigl(\lceil(w+q)/h_{\rm stride}\rceil-1\bigr)$-dependent. Setting $q=0$
gives the result for independent raw observations.
\end{proof}

Taking $H$ to be the point-cloud construction, its composition with $F$, or
its composition with $F$ and the whitened projection,
Proposition~\ref{prop:overlap_decimation_main} applies to $(Z_n)$, $(X_n)$,
and $(Y_n)$, respectively.

The next proposition retains the local-tilting model introduced before
Theorem~\ref{thm:ph_pl_local_minimax_interface}. It also retains the notation
$\theta_{n_k}$, $\gamma_k$, and $R_{n_k,\mathcal A_k}(h)$. It gives a delay
lower bound when $\gamma_k^2\mathcal A_k$ is of order at most
$\log\mathcal A_k$.
\begin{proposition}
\label{prop:below_boundary_main}
If $n_k\to\infty$, $\mathcal A_k\to\infty$, and
\[
    \gamma_k^2\mathcal A_k=O(\log\mathcal A_k),
\]
then, for every $\zeta\in(0,1)$,
\[
    R_{n_k,\mathcal A_k}(h)
    \geq
    \mathcal A_k^{1-\zeta}\{1+o(1)\}.
\]
\end{proposition}

\begin{proof}
Let $\tau_k\in\mathcal C_{\mathcal A_k}$ be arbitrary. For notational
simplicity, write $\tau=\tau_k$, and set
\[
    m_k=\left\lceil\mathcal A_k^{1-\zeta}\right\rceil .
\]

\pfstep{Local asymptotic order}
Write $I_k=I_{h,n_k}^{\Phi}$. By
Corollary~\ref{cor:key_local_relations},
\[
    I_k
    =
    \frac{\kappa_{\Phi,\rho}(h)}{2}
    \gamma_k^2\{1+o(1)\},
    \qquad
    I_k\asymp n_k^{-1}.
\]
The assumption $\gamma_k^2\mathcal A_k=O(\log\mathcal A_k)$ gives
\[
\begin{aligned}
    m_kI_k
    &=
    O\{\mathcal A_k^{1-\zeta}\gamma_k^2\}\\
    &=
    O(\mathcal A_k^{-\zeta}\log\mathcal A_k)
    \to0.
\end{aligned}
\]

\pfstep{Low-alarm-probability block}
Let $c_0>1$. Since $m_k=o(\mathcal A_k)$,
$c_0m_k/\mathcal A_k<1$ for all sufficiently large $k$. By
Lemma~\ref{lem:low_alarm_block_main}, there exists $j_k\geq0$ such that
$s_k=j_km_k$ and
\[
    H_k=\{\tau>s_k\}\in\mathscr F_{s_k}^X,
    \qquad
    \mathbb P_\infty(H_k)>0,
\]
with
\[
    \mathbb P_\infty(\tau\leq s_k+m_k\mid H_k)
    \leq
    c_0\frac{m_k}{\mathcal A_k}
    \to0.
\]

\pfstep{Block log-likelihood ratio}
Set the change point to $\nu_k=s_k+1$, and write
\[
    \Lambda_k
    =
    \sum_{t=\nu_k}^{\nu_k+m_k-1}
    \ell_{h,n_k}(U_t).
\]
Since $H_k\in\mathscr F_{\nu_k-1}^X$ and the pre-change law agrees with
$\mathbb P_\infty$,
\[
    \mathbb P_{\nu_k,\theta_{n_k}}(H_k)
    =
    \mathbb P_\infty(H_k)>0.
\]
Conditional on $H_k$, the support-point indices in the block are independent
and identically distributed according to $P_{\theta_{n_k}}$. Therefore,
\[
    \mathbb E_{\nu_k,\theta_{n_k}}(\Lambda_k\mid H_k)
    =
    m_kI_k
    \to0.
\]
Moreover,
\[
\begin{aligned}
    \operatorname{Var}_{\theta_{n_k}}
    \{\ell_{h,n_k}(U)\}
    &=
    \theta_{n_k}^{\top}
    \Sigma_{\theta_{n_k}}^\Phi
    \theta_{n_k}\\
    &=
    O(n_k^{-1})
    =
    O(I_k).
\end{aligned}
\]
Thus,
\[
    \operatorname{Var}_{\nu_k,\theta_{n_k}}(\Lambda_k\mid H_k)
    =
    O(m_kI_k)
    \to0.
\]
Chebyshev's inequality implies that $\Lambda_k\to0$ in probability under
$\mathbb P_{\nu_k,\theta_{n_k}}(\cdot\mid H_k)$.

\pfstep{Block alarm probability}
Let
\[
    \mathcal D_k
    =
    \{s_k<\tau\leq s_k+m_k\}.
\]
On $H_k$, the event $\mathcal D_k$ coincides with
$\{\tau\leq s_k+m_k\}$. Hence
\[
    \mathbb P_\infty(\mathcal D_k\mid H_k)
    \leq
    c_0\frac{m_k}{\mathcal A_k}
    \to0.
\]
Since $\Lambda_k\to0$ in probability, there exists a sequence
$\epsilon_k\to0$ such that
\[
    \mathbb P_{\nu_k,\theta_{n_k}}(\Lambda_k>\epsilon_k\mid H_k)\to0.
\]
The blockwise change-of-measure identity gives
\[
\begin{aligned}
    \mathbb P_{\nu_k,\theta_{n_k}}(\mathcal D_k\mid H_k)
    &=
    \mathbb E_\infty
    \left[
        e^{\Lambda_k}\mathbf 1_{\mathcal D_k}
        \mid H_k
    \right]\\
    &\leq
    \mathbb P_{\nu_k,\theta_{n_k}}(\Lambda_k>\epsilon_k\mid H_k)
    +
    e^{\epsilon_k}
    \mathbb P_\infty(\mathcal D_k\mid H_k)\\
    &\to0.
\end{aligned}
\]

\pfstep{Delay lower bound}
Since $H_k\in\mathscr F_{\nu_k-1}^X$, Lorden's criterion gives
\[
\begin{aligned}
    \bar d_{\theta_{n_k}}(\tau)
    &\ge
    \mathbb E_{\nu_k,\theta_{n_k}}
    \bigl[(\tau-\nu_k+1)_+\mid H_k\bigr] \\
    &\ge
    m_k
    \mathbb P_{\nu_k,\theta_{n_k}}
    (\tau>s_k+m_k\mid H_k) \\
    &=
    m_k\{1+o(1)\}\\
    &=
    \mathcal A_k^{1-\zeta}\{1+o(1)\}.
\end{aligned}
\]
The preceding probability bounds depend only on $m_kI_k\to0$ and
$m_k/\mathcal A_k\to0$. Their remainder terms do not depend on the chosen
stopping-time sequence. Taking the infimum over
$\tau_k\in\mathcal C_{\mathcal A_k}$ gives
\[
    R_{n_k,\mathcal A_k}(h)
    \geq
    \mathcal A_k^{1-\zeta}\{1+o(1)\}.
\]
\end{proof}

\section{Proofs of the Main Results}
\label{app:main_result_proofs}

\begin{proof}[Proof of Corollary~\ref{cor:detectability_boundary}]
By Theorem~\ref{thm:ph_pl_local_minimax_interface},
\[
    R_{n_k,\mathcal A_k}(h)
    \geq
    \frac{2}{C_{\Phi,\rho}}
    \frac{\log\mathcal A_k}{\gamma_k^2}
    \{1+o(1)\}.
\]
The oracle stopping time $\tau_{\eta_k}$ satisfies
$\operatorname{ARL}_0(\tau_{\eta_k})\geq\mathcal A_k$. Hence
$\tau_{\eta_k}\in\mathcal C_{\mathcal A_k}$. Combining this fact with the
delay upper bound in the theorem gives
\[
    R_{n_k,\mathcal A_k}(h)
    \leq
    \bar d_{\theta_{n_k}}(\tau_{\eta_k})
    \leq
    \frac{2\log\mathcal A_k}{\gamma_k^2}
    \{1+o(1)\}.
\]
The lower and upper bounds imply
\[
    R_{n_k,\mathcal A_k}(h)
    \asymp
    \frac{\log\mathcal A_k}{\gamma_k^2},
    \qquad
    \bar d_{\theta_{n_k}}(\tau_{\eta_k})
    \asymp
    \frac{\log\mathcal A_k}{\gamma_k^2}.
\]
If $\gamma_k^2\mathcal A_k/\log\mathcal A_k\to\infty$, then
\[
    \frac{\log\mathcal A_k}{\gamma_k^2}
    =
    o(\mathcal A_k).
\]
Thus, the local minimax delay is of smaller order than the ARL constraint,
and the oracle LLR-CUSUM attains the same delay order.
\end{proof}

\begin{proof}[Proof of Lemma~\ref{lem:natural_homology_isomorphism_persistent_betti}]
Commutativity of the diagram gives
\[
    (\phi_b)_*\circ\iota_{K,*}
    =
    \iota_{L,*}\circ(\phi_a)_*.
\]
Since $(\phi_a)_*$ is surjective,
\[
\begin{aligned}
    (\phi_b)_*\bigl(\operatorname{im}\iota_{K,*}\bigr)
    &=
    \operatorname{im}\bigl((\phi_b)_*\circ\iota_{K,*}\bigr) \\
    &=
    \operatorname{im}\bigl(\iota_{L,*}\circ(\phi_a)_*\bigr) \\
    &=
    \operatorname{im}\iota_{L,*}.
\end{aligned}
\]
Moreover, $(\phi_b)_*$ is an isomorphism, so its restriction defines an
isomorphism from $\operatorname{im}\iota_{K,*}$ onto
$\operatorname{im}\iota_{L,*}$. Hence
\[
    \dim\operatorname{im}\iota_{K,*}
    =
    \dim\operatorname{im}\iota_{L,*},
\]
which is precisely
\[
    \beta_q^{a,b}(K_\bullet)
    =
    \beta_q^{a,b}(L_\bullet).
\]
\end{proof}

\begin{proof}[Proof of Proposition~\ref{prop:tree_family}]
\pfstep{Metric realization}
Since $T$ has $m-1$ edges, choose a root $o\in V(T)$ and an orthonormal basis
$\{\xi_e:e\in E(T)\}$ of $\R^{m-1}$. For each $v\in V(T)$, define
\[
    y_v
    =
    \sum_{e\in\operatorname{path}_T(o,v)}\xi_e,
    \qquad
    Y_T=\{y_v:v\in V(T)\},
\]
where $\operatorname{path}_T(o,v)$ denotes the edge set of the unique path
from $o$ to $v$. For any $u,v\in V(T)$, cancellation of the common edges in
their two root paths
leaves exactly one signed basis vector for each edge on the path from $u$ to
$v$.
Consequently,
\[
    \norm{y_u-y_v}_2^2=d_T(u,v).
\]
Thus, $y_u$ and $y_v$ are joined by an edge in $\VR_\varepsilon(Y_T)$ if and
only if
\[
    d_T(u,v)\leq\varepsilon^2.
\]
Because $d_T(u,v)$ is integer valued, the $1$-skeleton of
$\VR_\varepsilon(Y_T)$ is $T^{\lfloor\varepsilon^2\rfloor}$, where $T^k$ is
obtained by joining distinct vertices whose tree distance is at most $k$.

\pfstep{Persistent Betti vectors}
If $\varepsilon<1$, then $\VR_\varepsilon(Y_T)$ consists of $m$ isolated
vertices. Now suppose that $\varepsilon\geq1$ and set
$k=\lfloor\varepsilon^2\rfloor\geq1$. If $v$ is a leaf of $T$ and $u$ is its
unique neighbour, then
\[
    N_{T^k}[v]
    \subseteq
    N_{T^k}[u].
\]
Indeed, if $x=v$, then $d_T(x,u)=1\leq k$. If $x\neq v$ and
$d_T(x,v)\leq k$, the path from $v$ to $x$ passes through $u$, and hence
\[
    d_T(x,u)=d_T(x,v)-1\leq k.
\]

Deleting $v$ leaves all distances among the remaining vertices unchanged.
The graph remaining at each filtration level is therefore the corresponding
tree power of the pruned tree. Repeating this argument gives a common
leaf-deletion order at every level with $\varepsilon\geq1$, and each deletion
is compatible with the filtration inclusions. By
Remark~\ref{rem:dominated_vertex_filtered_collapse} and
Lemma~\ref{lem:natural_homology_isomorphism_persistent_betti}, the subfiltration
for $\varepsilon\geq1$ strongly collapses, compatibly across levels, to the
one-point filtration.

It follows that, for every scale pair $a\leq b$,
\[
    \beta_q^{a,b}(Y_T)
    =
    \begin{cases}
        m, & q=0,\ \varepsilon_b<1,\\
        1, & q=0,\ \varepsilon_b\ge1,\\
        0, & q\ge1.
    \end{cases}
\]
This expression does not depend on the structure of the tree. Hence, for any
$T,T'\in\mathcal T_m$ and every scale pair $(a,b)$,
\[
    B^{a,b}(Y_T)=B^{a,b}(Y_{T'}).
\]

\pfstep{Positive PL spectrum}
When $1\le\varepsilon_a\le\varepsilon_b<\sqrt2$,
\[
    \lfloor\varepsilon_a^2\rfloor
    =
    \lfloor\varepsilon_b^2\rfloor
    =1.
\]
Therefore,
\[
    \VR_{\varepsilon_a}(Y_T)
    =
    \VR_{\varepsilon_b}(Y_T),
\]
and both complexes have $T$ as their $1$-skeleton. By the definition of the
persistent combinatorial Laplacian, $\mathcal L_0^{a,b}(Y_T)$ is the graph
Laplacian $L_T$ of $T$.

Write $L_T=D_T-A_T$, where $D_T$ and $A_T$ are the degree and adjacency
matrices, respectively. Since $A_T$ has zero diagonal and
$\operatorname{tr}(A_T^2)=2|E(T)|$,
\[
\begin{aligned}
    \operatorname{tr}(L_T^2)
    &=
    \operatorname{tr}(D_T^2)
    +
    \operatorname{tr}(A_T^2) \\
    &=
    \sum_{v\in V(T)}\deg_T(v)^2
    +
    2|E(T)|.
\end{aligned}
\]
Every tree on $m$ vertices satisfies $|E(T)|=m-1$. Hence, if
\[
    \sum_{v\in V(T)}\deg_T(v)^2
    \ne
    \sum_{v\in V(T')}\deg_{T'}(v)^2,
\]
then
\[
    \operatorname{tr}(L_T^2)
    \ne
    \operatorname{tr}(L_{T'}^2).
\]
Thus, $L_T$ and $L_{T'}$ have different spectra. Since both trees are
connected, each graph Laplacian has exactly one zero eigenvalue, and their
positive spectra must therefore differ.
\end{proof}

\begin{proof}[Proof of Lemma~\ref{lem:feature_compression_kl}]
Because $F_1=g\circ F_2$, the pair $(F_1,F_2)$ contains exactly the same
information as $F_2$. They therefore have the same KL divergence. Applying the
chain rule for KL divergence to the pair \citep{coverthomas2006} gives
\[
\begin{aligned}
    \KL(P_1^{F_2}\|P_0^{F_2})
    &=
    \KL(P_1^{F_1}\|P_0^{F_1}) \\
    &\quad+
    \sum_y P_1^{F_1}(y)
    \KL\!\left(
        P_1^{F_2\mid F_1=y}
        \middle\|
        P_0^{F_2\mid F_1=y}
    \right).
\end{aligned}
\]
Conditional KL divergences are nonnegative, which proves the inequality.
Equality holds if and only if the conditional KL divergence is zero for every
$y$ such that $P_1^{F_1}(y)>0$. This is equivalent to equality of the two
conditional distributions.
\end{proof}

\section{Finite-Horizon Bounds for the Plug-In Procedure}
\label{app:plugin_bound}

\begin{proposition}
\label{prop:whitened_projection_subgaussian}
For a selected feature map $F$, let $\mathcal X_F$ denote its finite range and
set
\[
    R_F
    =
    \max_{x\in\mathcal X_F}\|x\|_2 .
\]
For any state $k=0,1$, any $r>0$, and any deterministic vector
$a\in\mathbb R^p$ with $\|a\|_2\leq r^{-1/2}$, the following bound holds for
every $\lambda\in\mathbb R$:
\[
    \mathbb E_k
    \exp\!\left[\lambda a^\top(X_n-\mu_k)\right]
    \le
    \exp\left\{\frac{\lambda^2 R_F^2}{2r}\right\}.
\]
\end{proposition}

\begin{proof}
For any two possible values $x$ and $x'$ of $X_n$,
\[
    |a^\top x-a^\top x'|
    \le
    \|a\|_2\|x-x'\|_2
    \le
    2R_F\|a\|_2.
\]
Thus, the range of $a^\top X_n$ has length at most
$2R_F\|a\|_2$. Hoeffding's lemma \citep{hoeffding1963} therefore gives
\[
\begin{aligned}
    \mathbb E_k
    \exp\!\left[\lambda a^\top(X_n-\mu_k)\right]
    &\le
    \exp\left\{
        \frac{\lambda^2R_F^2\|a\|_2^2}{2}
    \right\} \\
    &\le
    \exp\left\{
        \frac{\lambda^2R_F^2}{2r}
    \right\}.
\end{aligned}
\]
\end{proof}

For each $r=0,1$, suppose that the state-$r$ training segment contains
independent window features $X_{r1},\ldots,X_{rN_{\rm tr}}$, and that the
state-$r$ evaluation feature $X$ is independent of the training segment. Write
\[
    R=R_F,
    \qquad
    \|X\|_2\leq R.
\]
In what follows, $\mathbb E_r\widehat s_n$ denotes the conditional expectation
of the score evaluated on a state-$r$ feature, given the training sample.

\begin{proposition}
\label{prop:training_error_event}
If $\rho>0$ and $\gamma_\rho>0$, then there exist constants $c_0$ and $C_0$
depending only on $(R,\rho,\rho_{\rm mult},\gamma_\rho)$ such that, for any
$\delta\in(0,1)$, whenever
\[
    N_{\rm tr}
    \ge
    C_0\{p+\log(1/\delta)\},
\]
the following bounds hold with probability at least $1-\delta$:
\[
    |\widehat\rho-\rho|\leq\frac{\rho}{2},
    \qquad
    \left\|\widehat W(\widehat\mu_1-\widehat\mu_0)\right\|_2
    \ge
    \frac{\gamma_\rho}{2},
\]
and
\[
    \max_{r=0,1}
    \left|\mathbb E_r\widehat s_n-\mathbb E_rs_n\right|
    +
    \left|\widehat c-c\right|
    \le
    c_0\sqrt{\frac{p+\log(1/\delta)}{N_{\rm tr}}}.
\]
\end{proposition}

\begin{proof}
Set
\[
    \Delta_{\rm tr}
    =
    \sqrt{\frac{p+\log(1/\delta)}{N_{\rm tr}}}.
\]

\pfstep{Concentration of training estimates}
By boundedness of the features, vector Hoeffding concentration
\citep{hoeffding1963}, and the matrix Bernstein inequality
\citep{tropp2012}, there exists a constant $C_R$ depending only on $R$ such
that, with probability at least $1-\delta$,
\[
    \max_{r=0,1}\|\widehat\mu_r-\mu_r\|_2
    +
    \|\widehat\Sigma-\Sigma\|_{\mathrm{op}}
    \le
    C_R\Delta_{\rm tr}.
\]
All remaining bounds are established on this event.

\pfstep{Whitened-direction error}
The ridge rule and the inequality
$|\operatorname{tr}(A)|\leq p\|A\|_{\mathrm{op}}$ give
\[
    |\widehat\rho-\rho|
    \le
    \rho_{\rm mult}\|\widehat\Sigma-\Sigma\|_{\mathrm{op}}
    \le
    C_{R,\rho_{\rm mult}}\Delta_{\rm tr}.
\]
Choose $C_0$ large enough that this bound is at most $\rho/2$. The matrix map
$A\mapsto A^{-1/2}$ is Lipschitz on matrices whose eigenvalues are bounded
below by $\rho/2$ \citep{bhatia1997}. Therefore,
\[
    \|\widehat W-W_\rho\|_{\mathrm{op}}
    \le
    C_\rho
    \left(
        \|\widehat\Sigma-\Sigma\|_{\mathrm{op}}
        +|\widehat\rho-\rho|
    \right)
    \le
    C_{R,\rho,\rho_{\rm mult}}\Delta_{\rm tr}.
\]
Moreover,
\[
    \|\widehat W\|_{\mathrm{op}}\leq(\widehat\rho)^{-1/2}
    \leq\sqrt{2/\rho},
    \qquad
    \|\mu_1-\mu_0\|_2\leq2R.
\]
Define
\[
    a=W_\rho(\mu_1-\mu_0),
    \qquad
    \widehat a=\widehat W(\widehat\mu_1-\widehat\mu_0).
\]
The identity
\[
\begin{aligned}
    \widehat a-a
    &=
    (\widehat W-W_\rho)(\mu_1-\mu_0) \\
    &\quad+
    \widehat W
    \{(\widehat\mu_1-\widehat\mu_0)-(\mu_1-\mu_0)\}
\end{aligned}
\]
gives
\[
    \|\widehat a-a\|_2
    \le
    C_{R,\rho,\rho_{\rm mult}}\Delta_{\rm tr}.
\]
Since $\|a\|_2=\gamma_\rho$, $C_0$ can be chosen so that
$N_{\rm tr}\geq C_0\{p+\log(1/\delta)\}$ implies
\[
    C_{R,\rho,\rho_{\rm mult}}\Delta_{\rm tr}
    \leq
    \frac{\gamma_\rho}{2}.
\]
It follows that
\[
    \|\widehat a\|_2
    \geq
    \frac{\gamma_\rho}{2},
\]
and the normalisation inequality yields
\[
    \|\widehat u-u\|_2
    \leq
    \frac{2}{\gamma_\rho}\|\widehat a-a\|_2
    \le
    C_{R,\rho,\gamma_\rho}\Delta_{\rm tr}.
\]

\pfstep{Score and reference-value errors}
For $r=0,1$,
\[
\begin{aligned}
    |\mathbb E_r\widehat s_n-\mathbb E_rs_n|
    &=
    |(\widehat u^\top\widehat W-u^\top W_\rho)\mu_r| \\
    &\leq
    \|\mu_r\|_2
    \left\{
        \|\widehat W\|_{\mathrm{op}}\|\widehat u-u\|_2
        +
        \|\widehat W-W_\rho\|_{\mathrm{op}}
    \right\} \\
    &\leq
    C_{R,\rho,\gamma_\rho}\Delta_{\rm tr}.
\end{aligned}
\]
From the definitions of $c$ and $\widehat c$,
\[
\begin{aligned}
    \widehat c-c
    &=
    \frac12
    (\widehat u^\top\widehat W-u^\top W_\rho)
    (\mu_0+\mu_1) \\
    &\quad+
    \frac12
    \widehat u^\top\widehat W
    \{(\widehat\mu_0-\mu_0)+(\widehat\mu_1-\mu_1)\}.
\end{aligned}
\]
The preceding bounds for the means, whitening matrices, and directions imply
\[
    |\widehat c-c|
    \leq
    C_{R,\rho,\gamma_\rho}\Delta_{\rm tr}.
\]
Absorbing the constants into $c_0$ proves the proposition.
\end{proof}

Define
\[
    \mathcal E_e
    =
    \left\{
    |\widehat\rho-\rho|\leq\frac{\rho}{2},
    \quad
    \left\|\widehat W(\widehat\mu_1-\widehat\mu_0)\right\|_2
    \geq
    \frac{\gamma_\rho}{2},
    \quad
    \max_{r=0,1}
    \left|\mathbb E_r\widehat s_n-\mathbb E_r s_n\right|
    +
    \left|\widehat c-c\right|
    \leq e
    \right\},
\]
where
\[
    e=c_0\sqrt{\frac{p+\log(1/\delta)}{N_{\rm tr}}}.
\]
Conditional on $\widehat W$, $\widehat u$, and $\widehat c$, the score
$\widehat s_n=\widehat u^\top\widehat W X_n$ satisfies
\[
    \widehat s_n=(\widehat W\widehat u)^\top X_n.
\]
Since $\|\widehat W\|_{\mathrm{op}}\leq\widehat\rho^{-1/2}$ and
$\|\widehat u\|_2=1$, we have
$\|\widehat W\widehat u\|_2\leq\widehat\rho^{-1/2}$.
Proposition~\ref{prop:whitened_projection_subgaussian} therefore shows that
the centred one-step score is $\sigma^2$-sub-Gaussian in both states, where
$\sigma^2=R_F^2/\widehat\rho$.

\begin{proposition}
\label{prop:mdependent_partial_sum}
Conditional on the Phase~I estimates, suppose that the evaluation-score
sequence $(\widehat s_n)$ is $q$-dependent in state $r\in\{0,1\}$. Then
equation~\eqref{eq:dependent_subgaussian_sum} holds with
$C_{\rm dep}=q+1$. In particular, one may take $C_{\rm dep}=1$ for independent
evaluation scores. If the raw observations are independent and the monitoring
windows overlap, one may take $q=g-1$ and $C_{\rm dep}=g$.
\end{proposition}

\begin{proof}
Fix state $r$, and let
$\xi_j=\widehat s_j-\mathbb E_r\widehat s_j$.
Partition $\{a,\ldots,a+m-1\}$ into $I_0,\ldots,I_q$ according to the index
modulo $q+1$. Any two indices in the same $I_\ell$ differ by at least $q+1$,
so the corresponding random variables are independent.

By the generalized H\"older inequality and the $\sigma^2$-sub-Gaussian
property of each score,
\[
\begin{aligned}
\mathbb E_r\exp\!\left\{
\lambda\sum_{j=a}^{a+m-1}\xi_j
\right\}
&\leq
\prod_{\ell=0}^{q}
\left[
\mathbb E_r\exp\!\left\{
(q+1)\lambda\sum_{j\in I_\ell}\xi_j
\right\}
\right]^{1/(q+1)} \\
&\leq
\exp\!\left\{
\frac{(q+1)m\sigma^2\lambda^2}{2}
\right\}.
\end{aligned}
\]
This is equation~\eqref{eq:dependent_subgaussian_sum} with
$C_{\rm dep}=q+1$. Finally,
Proposition~\ref{prop:overlap_decimation_main} shows that overlapping-window
scores are $(g-1)$-dependent when the raw observations are independent.
Hence, one may take $C_{\rm dep}=g$.
\end{proof}

\begin{proposition}
\label{prop:main_plugin_order}
On $\mathcal E_e$, let $\gamma_{\rm eff}:=\gamma_\rho-2e>0$, and suppose that
equation~\eqref{eq:dependent_subgaussian_sum} holds. If
\[
    \eta
    \geq
    \frac{C_{\rm dep}\sigma^2}{\gamma_{\rm eff}}
    \log\frac{N(N+1)}{2\alpha},
\]
then
\[
    \mathbb P_\infty\!\left(
    T_\eta\leq N
    \,\middle|\,
    \widehat W,\widehat u,\widehat c
    \right)
    \leq\alpha,
\]
and
\[
    \mathbb E_\nu\!\left[
    (T_\eta-\nu+1)_+
    \,\middle|\,
    \widehat W,\widehat u,\widehat c
    \right]
    \leq
    \frac{4\eta}{\gamma_{\rm eff}}
    +1+
    \frac{32C_{\rm dep}\sigma^2}{\gamma_{\rm eff}^2}.
\]
In particular, if
$\eta\asymp
C_{\rm dep}\sigma^2\log(N/\alpha)/\gamma_{\rm eff}$ and
$\log(N/\alpha)\ge1$, then the delay bound is
\[
    O\!\left\{
    \frac{C_{\rm dep}\sigma^2\log(N/\alpha)}
    {\gamma_{\rm eff}^2}
    \right\}.
\]
\end{proposition}

\begin{proof}
By the definition of $\mathcal E_e$ and the identities
$c-\mathbb E_0s_n=\mathbb E_1s_n-c=\gamma_\rho/2$,
\[
    \widehat c-\mathbb E_0\widehat s_n
    \geq
    \frac{\gamma_{\rm eff}}{2},
    \qquad
    \mathbb E_1\widehat s_n-\widehat c
    \geq
    \frac{\gamma_{\rm eff}}{2}.
\]
Set $Y_n=\widehat s_n-\widehat c$.

\pfstep{False-alarm bound}
The interval-sum representation of Page's CUSUM gives
\[
    \{T_\eta\leq N\}
    \subseteq
    \bigcup_{1\leq a\leq b\leq N}
    \left\{
        \sum_{j=a}^{b}Y_j\geq\eta
    \right\}.
\]
Fix $1\leq a\leq b\leq N$, and let $m=b-a+1$. Under state $0$,
$\mathbb E_0Y_j\leq-\gamma_{\rm eff}/2$. Hence,
\[
\begin{aligned}
&\mathbb P_\infty\!\left(
    \sum_{j=a}^{b}Y_j\geq\eta
    \,\middle|\,
    \widehat W,\widehat u,\widehat c
\right) \\
&\quad\leq
\mathbb P_\infty\!\left(
    \sum_{j=a}^{b}(Y_j-\mathbb E_0Y_j)
    \geq
    \eta+\frac{m\gamma_{\rm eff}}{2}
    \,\middle|\,
    \widehat W,\widehat u,\widehat c
\right) \\
&\quad\leq
\exp\!\left\{
    -\frac{(\eta+m\gamma_{\rm eff}/2)^2}
    {2C_{\rm dep}m\sigma^2}
\right\}.
\end{aligned}
\]
Since $(u+v)^2\geq4uv$,
\[
    \frac{(\eta+m\gamma_{\rm eff}/2)^2}
    {2C_{\rm dep}m\sigma^2}
    \geq
    \frac{\eta\gamma_{\rm eff}}
    {C_{\rm dep}\sigma^2}.
\]
There are $N(N+1)/2$ intervals in the monitoring horizon. Therefore,
\[
\begin{aligned}
&\mathbb P_\infty\!\left(
    T_\eta\leq N
    \,\middle|\,
    \widehat W,\widehat u,\widehat c
\right) \\
&\quad\leq
\frac{N(N+1)}{2}
\exp\!\left\{
    -\frac{\eta\gamma_{\rm eff}}
    {C_{\rm dep}\sigma^2}
\right\}
\leq\alpha.
\end{aligned}
\]

\pfstep{Delay bound}
Let $\mu_Y=\mathbb E_1Y_n$. Then
$\mu_Y\geq\gamma_{\rm eff}/2$. Define
\[
    \tau_A
    =
    \inf\left\{
        m\geq1:
        \sum_{j=\nu}^{\nu+m-1}Y_j\geq\eta
    \right\}.
\]
The CUSUM statistic is nonnegative at the change point. Hence,
$(T_\eta-\nu+1)_+\leq\tau_A$.

For every integer $n\geq2\eta/\mu_Y$, the event $\{\tau_A>n\}$ implies
$\sum_{j=\nu}^{\nu+n-1}Y_j<\eta$. Therefore,
\[
\begin{aligned}
&\mathbb P_\nu\!\left(
    \tau_A>n
    \,\middle|\,
    \widehat W,\widehat u,\widehat c
\right) \\
&\quad\leq
\mathbb P_\nu\!\left(
    \sum_{j=\nu}^{\nu+n-1}(Y_j-\mu_Y)
    \leq-\frac{n\mu_Y}{2}
    \,\middle|\,
    \widehat W,\widehat u,\widehat c
\right) \\
&\quad\leq
\exp\!\left\{
    -\frac{n\mu_Y^2}{8C_{\rm dep}\sigma^2}
\right\}.
\end{aligned}
\]
Let $n_0=\lceil2\eta/\mu_Y\rceil$. The tail-sum formula gives
\[
\begin{aligned}
\mathbb E_\nu\!\left[
    \tau_A
    \,\middle|\,
    \widehat W,\widehat u,\widehat c
\right]
&\leq
n_0+
\sum_{n=n_0}^{\infty}
\exp\!\left\{
    -\frac{n\mu_Y^2}{8C_{\rm dep}\sigma^2}
\right\} \\
&\leq
\frac{2\eta}{\mu_Y}
+1+
\frac{8C_{\rm dep}\sigma^2}{\mu_Y^2}.
\end{aligned}
\]
Since $\mu_Y\geq\gamma_{\rm eff}/2$,
\[
\mathbb E_\nu\!\left[
    (T_\eta-\nu+1)_+
    \,\middle|\,
    \widehat W,\widehat u,\widehat c
\right]
\leq
\frac{4\eta}{\gamma_{\rm eff}}
+1+
\frac{32C_{\rm dep}\sigma^2}{\gamma_{\rm eff}^2}.
\]

Finally, $\mathbb E_1Y_n\leq2R_F/\sqrt{\widehat\rho}=2\sigma$ and
$\mathbb E_1Y_n\geq\gamma_{\rm eff}/2$, so
$\gamma_{\rm eff}\leq4\sigma$. When $\log(N/\alpha)\geq1$, the constant term
and the last term can both be absorbed into
$C_{\rm dep}\sigma^2\log(N/\alpha)/\gamma_{\rm eff}^2$, which proves the
final claim.
\end{proof}

\begin{corollary}
\label{cor:plugin_unconditional}
Under the conditions of Proposition~\ref{prop:main_plugin_order}, let
$0<\delta<\alpha<1$. If
$\mathbb P(\mathcal E_e)\geq1-\delta$ over the Phase~I training sample and
the control limit satisfies
\[
    \eta
    \geq
    \frac{C_{\rm dep}\sigma^2}{\gamma_{\rm eff}}
    \log\frac{N(N+1)}{2(\alpha-\delta)},
\]
then
\[
    \mathbb P_\infty(T_\eta\leq N)\leq\alpha.
\]
\end{corollary}

\begin{proof}
Because the training and evaluation segments are independent, condition on the
Phase~I training sample. On $\mathcal E_e$, the threshold satisfies the
condition of Proposition~\ref{prop:main_plugin_order} with
$\alpha-\delta$ as the conditional false-alarm budget. Hence, the conditional
false-alarm probability is at most $\alpha-\delta$ on $\mathcal E_e$ and at
most $1$ on $\mathcal E_e^c$. Therefore,
\[
\begin{aligned}
    \mathbb P_\infty(T_\eta\le N)
    \le
    (\alpha-\delta)\mathbb P(\mathcal E_e)
    +\mathbb P(\mathcal E_e^c)
    \le
    \alpha-\delta+\delta
    =
    \alpha.
\end{aligned}
\]
\end{proof}

\clearpage
\section{Additional Experiments and Diagnostics}
\label{app:additional_experiments}

\setcounter{figure}{0}
\renewcommand{\thefigure}{D.\arabic{figure}}
\setcounter{table}{0}
\renewcommand{\thetable}{D.\arabic{table}}

\subsection{Synthetic Experiments and Dimensional Scaling}
\label{app:synthetic_diagnostics}

\noindent\textit{Tree-order selection.}
Proposition~\ref{prop:tree_family} shows that the point clouds corresponding
to the path and star trees have identical persistent Betti vectors but
different positive PL spectra. To choose the tree order for the random
tree-family experiments, define the relative separation based on the first
$J$ positive eigenvalues by
\[
    D_J(m)
    =
    \frac{\norm{\Lambda_J(P_m)-\Lambda_J(S_m)}_2}
    {\left\{
    \norm{\Lambda_J(P_m)}_2^2/2+
    \norm{\Lambda_J(S_m)}_2^2/2
    \right\}^{1/2}}.
\]
Figure~\ref{fig:appendix_tree_order} evaluates
$m=4,\ldots,9999$ and compares the full positive spectrum with the three
truncation levels $J=5,10,20$.

\begin{figure}[!htbp]
\centering
\includegraphics[width=0.88\linewidth]{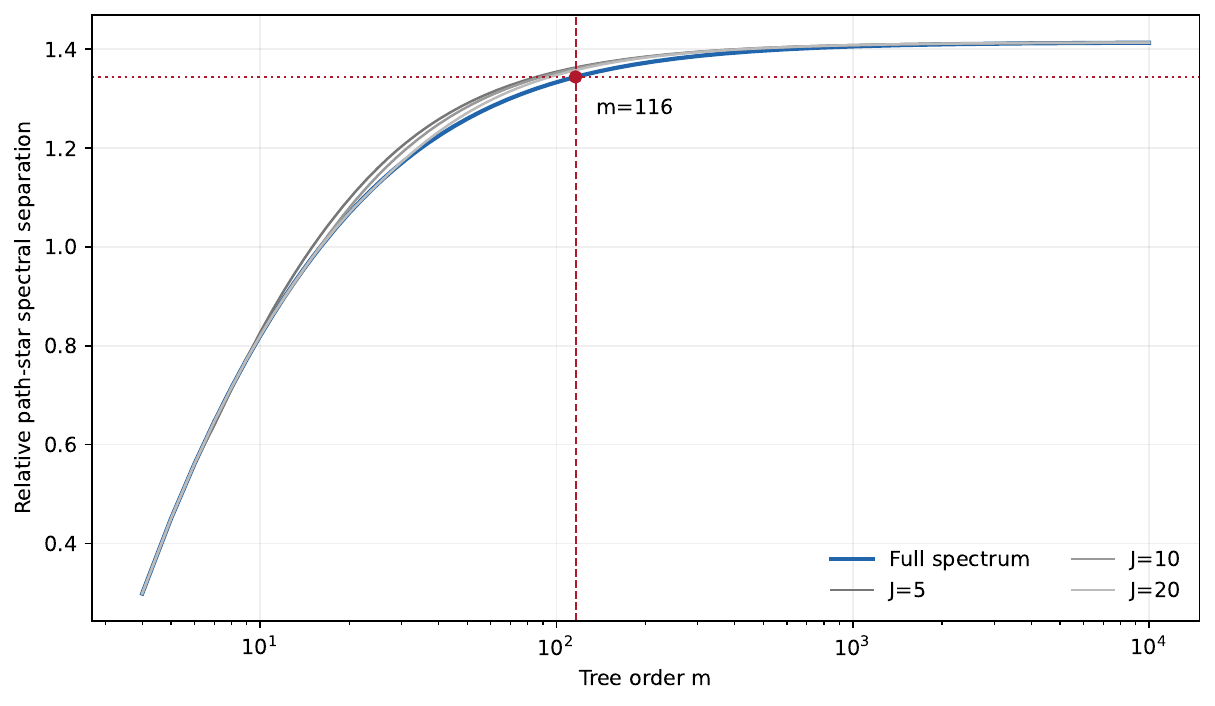}
\caption{Relative positive-spectrum separation between the path and star trees
as the tree order $m$ increases. The red dashed line marks $m=116$, the first
tree order at which the full-spectrum separation reaches 95\% of its
asymptotic value. The grey curves correspond to $J=5,10,20$. By
Proposition~\ref{prop:tree_family}, the persistent Betti vectors remain
identical across the tree structures.}
\label{fig:appendix_tree_order}
\end{figure}
\FloatBarrier

For the full-spectrum curve, $m=116$ is the smallest tree order at which
$D_{\mathrm{full}}(m)$ reaches $0.95\sqrt{2}$. At this point,
$D_{\mathrm{full}}=1.344$, whereas $D_5$, $D_{10}$, and $D_{20}$ are
1.363, 1.361, and 1.358, respectively.

\noindent\textit{Synthetic control-limit calibration.}
To compare the theoretical sufficient control limit with the numerically
calibrated control limit, we consider the nine combinations of
\[
    \Delta\in\{0.30,0.50,0.75\},
    \qquad
    \mathcal A\in\{200,800,3200\}.
\]
For each combination, the control limit is first calibrated by bisection using
2000 state-0 paths. Another 2000 state-0 paths are used to estimate the ARL,
and 1000 paths under the alternative are used to estimate the delay.

\begin{figure}[!htbp]
\centering
\includegraphics[width=\linewidth]{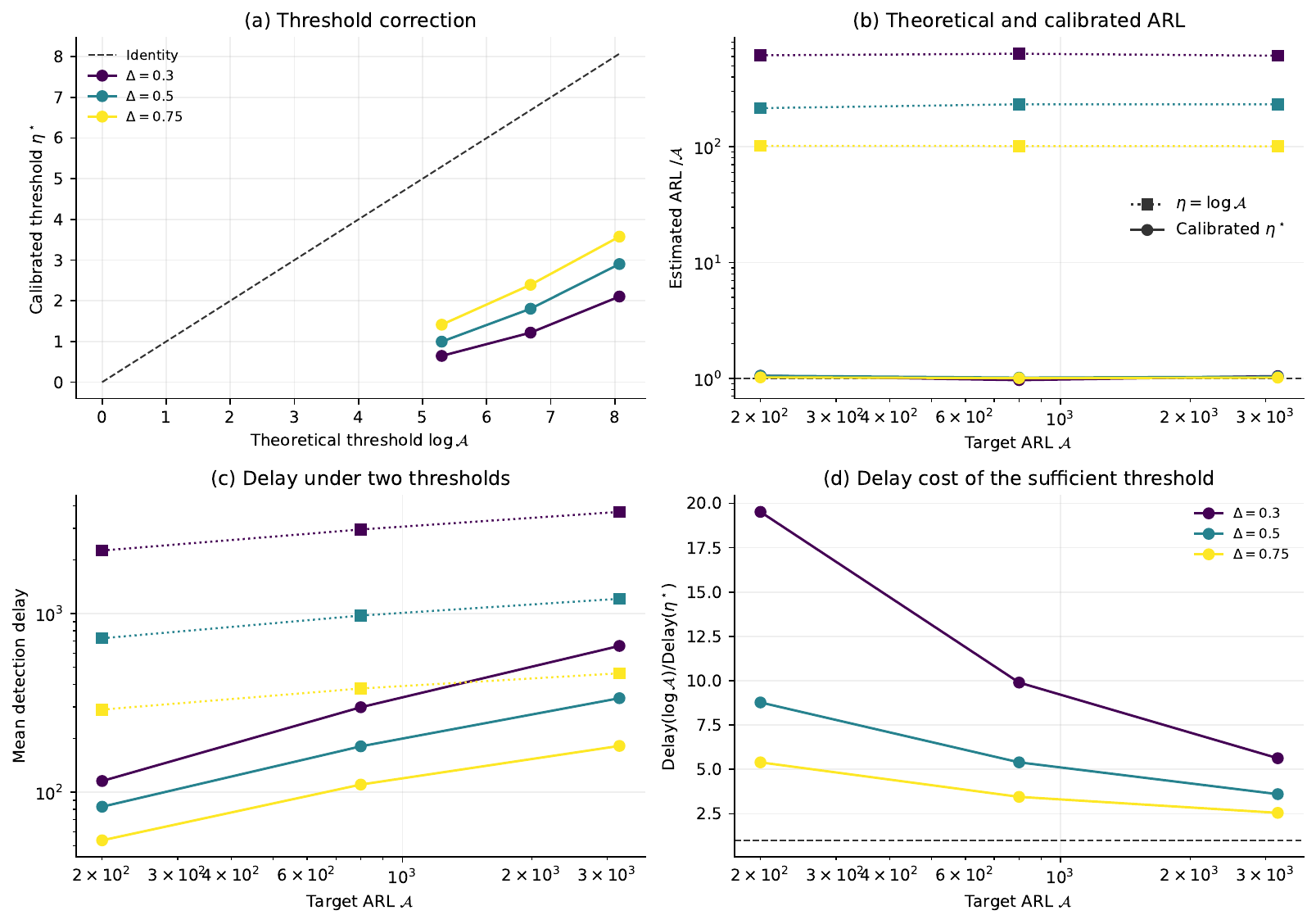}
\caption{Control-limit diagnostics for the synthetic tree family.
(a) The theoretical and numerically calibrated control limits.
(b) The ratio of estimated ARL to target ARL under the two control limits.
(c) Mean detection delay under the two control limits.
(d) The delay penalty from using the theoretical sufficient control limit
instead of the calibrated control limit.}
\label{fig:appendix_synthetic_calibration}
\end{figure}
\FloatBarrier

Figure~\ref{fig:appendix_synthetic_calibration} shows that the theoretical
control limit $\log\mathcal A$ is substantially larger than the calibrated
control limit. The ARL estimates obtained after numerical calibration closely
match their targets, whereas the theoretical control limit produces longer
detection delays. The control limits and ARL estimates for all nine
combinations are reported in
Table~\ref{tab:appendix_synthetic_calibration}.

\begin{table}[!htbp]
\centering
\caption{Control-limit calibration results for the synthetic tree family.
\textit{Ratio} denotes Estimated $\mathrm{ARL}_0/\mathcal A$. The arrows
indicate whether the estimate is above or below its target.}
\label{tab:appendix_synthetic_calibration}
\resizebox{0.88\linewidth}{!}{%
\begin{tabular}{rrrrrrc}
\toprule
$\Delta$ & $\mathcal A$ & $\log\mathcal A$ & $\eta^\star$ & Estimated ARL$_0$ & Ratio & Direction \\
\midrule
0.30 & 200 & 5.298 & 0.645 & 210.6 & 1.053 & $\uparrow$ \\
0.30 & 800 & 6.685 & 1.217 & 773.2 & 0.967 & $\downarrow$ \\
0.30 & 3200 & 8.071 & 2.104 & 3350.5 & 1.047 & $\uparrow$ \\
0.50 & 200 & 5.298 & 0.996 & 210.8 & 1.054 & $\uparrow$ \\
0.50 & 800 & 6.685 & 1.805 & 811.6 & 1.014 & $\uparrow$ \\
0.50 & 3200 & 8.071 & 2.902 & 3309.7 & 1.034 & $\uparrow$ \\
0.75 & 200 & 5.298 & 1.413 & 204.0 & 1.020 & $\uparrow$ \\
0.75 & 800 & 6.685 & 2.392 & 804.1 & 1.005 & $\uparrow$ \\
0.75 & 3200 & 8.071 & 3.576 & 3249.9 & 1.016 & $\uparrow$ \\
\bottomrule
\end{tabular}
}
\end{table}
\FloatBarrier

The ratio of estimated ARL to target ARL ranges from 0.967 to 1.054. The
deviations are small across all nine combinations.

\noindent\textit{Dimensional stress and spectral truncation.}
For a tree of order $m$, the full positive spectrum contains $p=m-1$ positive
eigenvalues. Phase~I uses 80 samples from each state. For the full spectrum,
the population ridge-whitened separation is set to $\gamma_\star=0.5$, and the
ridge multiplier is 2. Let $\widehat v$ be the Fisher direction estimated in
Phase~I. Define its relative population separation by
\[
    R_{\rm sep}
    =
    \frac{\widehat v^\top(\mu_1-\mu_0)}
    {\left\{
    \widehat v^\top(\Sigma+\rho I_p)\widehat v
    \right\}^{1/2}\gamma_\rho}.
\]
Here, $\gamma_\rho$ is the oracle ridge-whitened separation under the same
spectral representation. The value $R_{\rm sep}=1$ means that the estimated
direction attains the oracle separation.

The experiment uses paired replications. Within each replication, only the
truncation level changes. The same Phase~I noise, calibration paths, state-0
paths, and change-point paths are used across truncation levels. We compare
$J=5,10,20$ with the full positive spectrum and use 30 paired replications for
each condition. Figure~\ref{fig:appendix_dimension_stress} evaluates $p=40$
and $p=100,200,\ldots,10000$. The horizontal axis is
$p/(n_0+n_1)$. The dashed line marks $p=n_0+n_1$, where the spectral dimension
equals the total Phase~I sample size.

\begin{figure}[!htbp]
\centering
\includegraphics[width=\linewidth]{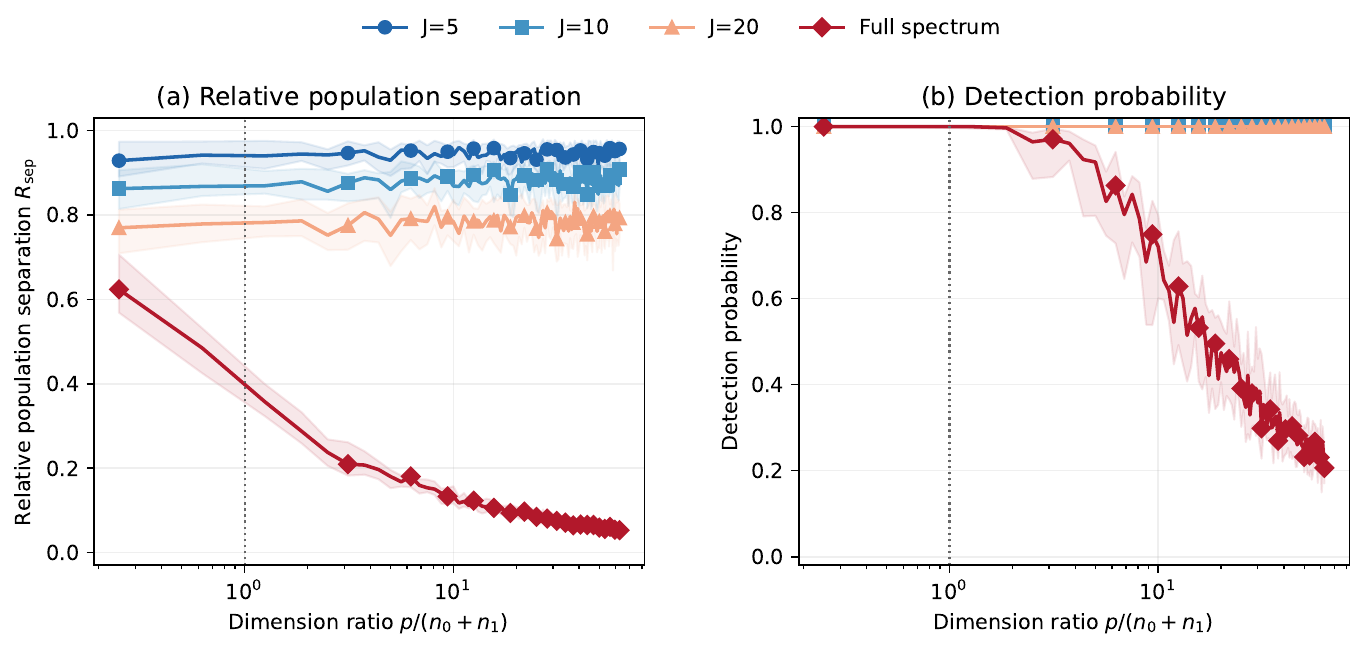}
\caption{Finite-sample performance as the spectral dimension increases
relative to the Phase~I sample size. The experiment uses isotropic Phase~I
noise, $n_0=n_1=80$, $\gamma_\star=0.5$, and a ridge multiplier of 2.
(a) Relative population separation $R_{\rm sep}$.
(b) Detection probability within a monitoring horizon of length 500.
The curves show medians across 30 paired replications, and the shaded bands
show interquartile ranges.}
\label{fig:appendix_dimension_stress}
\end{figure}
\FloatBarrier

As $p$ increases from 40 to 10,000, $R_{\rm sep}$ for the full positive
spectrum decreases from 0.624 to 0.053, and its detection probability
decreases from 1.000 to 0.207. For $J=5,10,20$, the corresponding ranges of
$R_{\rm sep}$ are 0.919--0.961, 0.847--0.909, and 0.743--0.830. Their
detection probabilities remain 1.000 at every dimension.

Table~\ref{tab:appendix_dimension_anchor} gives the full results for two noise
structures, two Phase~I sample sizes, and five spectral dimensions.

\begin{table}[!htbp]
\centering
\caption{Anchor results for the dimensional stress experiment. Each cell under
a truncation level reports $R_{\rm sep}$/detection probability.
Here, $n_0=n_1$ is the Phase~I sample size for each state, and $p$ is the
available full-spectrum dimension.}
\label{tab:appendix_dimension_anchor}
\resizebox{\linewidth}{!}{%
\begin{tabular}{lrrcccc}
\toprule
Noise & $n_0=n_1$ & $p$ & $J=5$ & $J=10$ & $J=20$ & Full \\
\midrule
Isotropic & 40 & 40 & 0.882/1.000 & 0.769/1.000 & 0.622/1.000 & 0.486/1.000 \\
Isotropic & 40 & 400 & 0.877/1.000 & 0.778/1.000 & 0.651/1.000 & 0.182/0.877 \\
Isotropic & 40 & 1000 & 0.861/1.000 & 0.770/1.000 & 0.628/1.000 & 0.116/0.543 \\
Isotropic & 40 & 4000 & 0.916/1.000 & 0.805/1.000 & 0.671/1.000 & 0.062/0.261 \\
Isotropic & 40 & 10000 & 0.923/1.000 & 0.764/1.000 & 0.637/1.000 & 0.040/0.139 \\
Isotropic & 80 & 40 & 0.929/1.000 & 0.863/1.000 & 0.770/1.000 & 0.624/1.000 \\
Isotropic & 80 & 400 & 0.956/1.000 & 0.868/1.000 & 0.750/1.000 & 0.249/0.967 \\
Isotropic & 80 & 1000 & 0.939/1.000 & 0.868/1.000 & 0.773/1.000 & 0.158/0.835 \\
Isotropic & 80 & 4000 & 0.956/1.000 & 0.898/1.000 & 0.790/1.000 & 0.084/0.390 \\
Isotropic & 80 & 10000 & 0.948/1.000 & 0.876/1.000 & 0.765/1.000 & 0.057/0.239 \\
Tail-inflated & 40 & 40 & 0.916/1.000 & 0.856/1.000 & 0.760/1.000 & 0.541/1.000 \\
Tail-inflated & 40 & 400 & 0.942/1.000 & 0.875/1.000 & 0.762/1.000 & 0.155/0.730 \\
Tail-inflated & 40 & 1000 & 0.945/1.000 & 0.879/1.000 & 0.780/1.000 & 0.104/0.434 \\
Tail-inflated & 40 & 4000 & 0.956/1.000 & 0.882/1.000 & 0.786/1.000 & 0.050/0.211 \\
Tail-inflated & 40 & 10000 & 0.950/1.000 & 0.887/1.000 & 0.799/1.000 & 0.033/0.121 \\
Tail-inflated & 80 & 40 & 0.962/1.000 & 0.923/1.000 & 0.834/1.000 & 0.670/1.000 \\
Tail-inflated & 80 & 400 & 0.969/1.000 & 0.925/1.000 & 0.848/1.000 & 0.225/0.967 \\
Tail-inflated & 80 & 1000 & 0.969/1.000 & 0.925/1.000 & 0.878/1.000 & 0.148/0.705 \\
Tail-inflated & 80 & 4000 & 0.977/1.000 & 0.934/1.000 & 0.867/1.000 & 0.072/0.298 \\
Tail-inflated & 80 & 10000 & 0.969/1.000 & 0.941/1.000 & 0.874/1.000 & 0.044/0.164 \\
\bottomrule
\end{tabular}
}
\end{table}
\FloatBarrier

At $p=10{,}000$, the detection probability of the full positive spectrum
ranges from 0.121 to 0.239. The detection probability is 1.000 for
$J=5,10,20$ under every anchor condition. On a separate set of state-0 paths,
the median pre-change alarm probabilities range from 0.047 to 0.054, close to
the target of 0.05. In this experiment, retaining the first 5--20 positive
eigenvalues reduces the dimension of the projection direction estimated in
Phase~I and maintains detection probability when $p\gg n_0+n_1$.

\subsection{Complete Real-Data Results and Detection Mechanisms}
\label{app:real_diagnostics}

\noindent\textit{Complete numerical results.}
Tables~\ref{tab:appendix_real_arl_matrix}--%
\ref{tab:appendix_real_counts_matrix} report the complete results for 11
methods on SWaT and the four Electric Motor conditions. Each task uses 2000
calibration paths, 2000 state-0 evaluation paths, and 1000 change-point paths.

\begin{table}[!htbp]
\centering
\caption{Estimated $\mathrm{ARL}_0$ for SWaT and the four Electric Motor
conditions. The column headings give the target ARLs. The notation
$>133{,}523$ indicates that none of the 2000 state-0 paths of length 200
raised an alarm.}
\label{tab:appendix_real_arl_matrix}
\setlength{\tabcolsep}{2.4pt}
\renewcommand{\arraystretch}{0.92}
\resizebox{\linewidth}{!}{%
\begin{tabular}{l*{15}{r}}
\toprule
Method & \multicolumn{3}{c}{SWaT} & \multicolumn{3}{c}{Motor: no load} & \multicolumn{3}{c}{Motor: load} & \multicolumn{3}{c}{Motor: background} & \multicolumn{3}{c}{Motor: load + background} \\
\cmidrule(lr){2-4}\cmidrule(lr){5-7}\cmidrule(lr){8-10}\cmidrule(lr){11-13}\cmidrule(lr){14-16}
Target ARL & 500 & 1000 & 2000 & 500 & 1000 & 2000 & 500 & 1000 & 2000 & 500 & 1000 & 2000 & 500 & 1000 & 2000 \\
\midrule
PL-CUSUM & $>133523$ & $>133523$ & $>133523$ & $>133523$ & $>133523$ & $>133523$ & $>133523$ & $>133523$ & $>133523$ & 18082 & 18082 & 18082 & $>133523$ & $>133523$ & $>133523$ \\
OK-CUSUM & 540 & 1077 & 1982 & 540 & 1046 & 2211 & 542 & 1017 & 2133 & 534 & 1043 & 2084 & 509 & 955 & 2049 \\
Scan B & 466 & 958 & 1950 & 491 & 963 & 1888 & 474 & 993 & 2171 & 512 & 978 & 1849 & 526 & 1053 & 2072 \\
KCUSUM & 498 & 1014 & 1929 & 507 & 1117 & 2004 & 476 & 848 & 1741 & 549 & 1135 & 2084 & 468 & 981 & 1759 \\
NEWMA & 516 & 1135 & 2399 & 448 & 935 & 1939 & 492 & 1043 & 2308 & 524 & 938 & 1637 & 472 & 975 & 1993 \\
PCA-CUSUM & 588 & 1073 & 2061 & 507 & 1102 & 2146 & 512 & 1084 & 1918 & 467 & 1020 & 2294 & 440 & 862 & 1785 \\
Hotelling $T^2$ & 504 & 1040 & 2446 & 482 & 969 & 2015 & 551 & 1020 & 1939 & 477 & 1040 & 1929 & 522 & 935 & 1888 \\
RAW-CUSUM & 536 & 1135 & 2308 & 503 & 966 & 2108 & 509 & 996 & 1918 & 511 & 1020 & 1960 & 463 & 946 & 1725 \\
GraphScan-kNN & 503 & 1005 & 1960 & 494 & 949 & 2004 & 503 & 1151 & 2197 & 511 & 999 & 1918 & 539 & 1005 & 2197 \\
PCA-BOCPD & 615 & 6679 & 6679 & 556 & 1117 & 2197 & 507 & 1020 & 2353 & 500 & 981 & 1950 & 512 & 996 & 1898 \\
E-GaussianBet & 493 & 1027 & 1869 & 470 & 990 & 2171 & 501 & 1017 & 1929 & 493 & 952 & 1812 & 500 & 978 & 1929 \\
\bottomrule
\end{tabular}
}
\end{table}
\FloatBarrier

\begin{table}[!htbp]
\centering
\caption{Conditional expected detection delay (EDD) for SWaT and the four
Electric Motor conditions. EDD is reported only when Success is at least 0.5.
Bold values indicate the smallest EDD within each task and target ARL.}
\label{tab:appendix_real_edd_matrix}
\setlength{\tabcolsep}{2.4pt}
\renewcommand{\arraystretch}{0.92}
\resizebox{\linewidth}{!}{%
\begin{tabular}{l*{15}{r}}
\toprule
Method & \multicolumn{3}{c}{SWaT} & \multicolumn{3}{c}{Motor: no load} & \multicolumn{3}{c}{Motor: load} & \multicolumn{3}{c}{Motor: background} & \multicolumn{3}{c}{Motor: load + background} \\
\cmidrule(lr){2-4}\cmidrule(lr){5-7}\cmidrule(lr){8-10}\cmidrule(lr){11-13}\cmidrule(lr){14-16}
Target ARL & 500 & 1000 & 2000 & 500 & 1000 & 2000 & 500 & 1000 & 2000 & 500 & 1000 & 2000 & 500 & 1000 & 2000 \\
\midrule
PL-CUSUM & \textbf{0.00} & \textbf{0.00} & \textbf{0.00} & \textbf{0.01} & \textbf{0.01} & \textbf{0.01} & \textbf{0.00} & \textbf{0.00} & \textbf{0.00} & 0.09 & 0.09 & 0.09 & \textbf{0.00} & \textbf{0.00} & \textbf{0.00} \\
OK-CUSUM & 18.70 & 18.89 & 24.06 & 16.28 & 17.33 & 18.34 & 12.38 & 13.33 & 14.27 & 20.25 & 21.52 & 22.69 & 12.73 & 13.65 & 14.60 \\
Scan B & 15.80 & 17.20 & 18.92 & 21.30 & 23.25 & 25.14 & 19.24 & 21.22 & 23.25 & 25.80 & 28.24 & 30.56 & 19.22 & 21.31 & 23.12 \\
KCUSUM & 53.38 & 61.95 & 66.58 & 44.81 & 50.53 & -- & 43.68 & 48.41 & 54.87 & -- & -- & -- & 42.46 & 49.61 & 54.54 \\
NEWMA & -- & -- & -- & -- & -- & -- & -- & -- & -- & -- & -- & -- & -- & -- & -- \\
PCA-CUSUM & -- & -- & -- & 0.73 & 0.91 & 1.07 & 2.57 & 3.08 & 3.53 & \textbf{0.01} & \textbf{0.01} & \textbf{0.02} & -- & -- & -- \\
Hotelling $T^2$ & -- & -- & -- & 12.23 & 14.32 & 16.03 & 4.75 & 5.51 & 6.17 & 23.12 & 26.54 & 29.02 & 5.21 & 5.94 & 6.95 \\
RAW-CUSUM & -- & -- & -- & 0.25 & 0.44 & 0.67 & 0.58 & 0.84 & 1.15 & 0.38 & 0.58 & 0.78 & 0.03 & 0.18 & 0.45 \\
GraphScan-kNN & -- & -- & -- & -- & -- & -- & 12.52 & 13.77 & -- & -- & -- & -- & 10.96 & 12.18 & 13.48 \\
PCA-BOCPD & -- & -- & -- & -- & -- & -- & -- & -- & -- & -- & -- & -- & -- & -- & -- \\
E-GaussianBet & -- & -- & -- & -- & -- & -- & 60.76 & -- & -- & 44.02 & 48.11 & 51.49 & 66.43 & -- & -- \\
\bottomrule
\end{tabular}
}
\end{table}
\FloatBarrier

\begin{table}[!htbp]
\centering
\caption{Success/Failure counts for SWaT and the four Electric Motor
conditions. Each cell is based on 1000 change-point paths. The number of
pre-change alarms is $1000-\mathrm{Success}-\mathrm{Failure}$. Bold values
indicate the largest Success count within each task and target ARL.}
\label{tab:appendix_real_counts_matrix}
\setlength{\tabcolsep}{2.4pt}
\renewcommand{\arraystretch}{0.92}
\resizebox{\linewidth}{!}{%
\begin{tabular}{l*{15}{r}}
\toprule
Method & \multicolumn{3}{c}{SWaT} & \multicolumn{3}{c}{Motor: no load} & \multicolumn{3}{c}{Motor: load} & \multicolumn{3}{c}{Motor: background} & \multicolumn{3}{c}{Motor: load + background} \\
\cmidrule(lr){2-4}\cmidrule(lr){5-7}\cmidrule(lr){8-10}\cmidrule(lr){11-13}\cmidrule(lr){14-16}
Target ARL & 500 & 1000 & 2000 & 500 & 1000 & 2000 & 500 & 1000 & 2000 & 500 & 1000 & 2000 & 500 & 1000 & 2000 \\
\midrule
PL-CUSUM & \textbf{1000/0} & \textbf{1000/0} & \textbf{1000/0} & \textbf{1000/0} & \textbf{1000/0} & \textbf{1000/0} & \textbf{1000/0} & \textbf{1000/0} & \textbf{1000/0} & \textbf{998/0} & \textbf{998/0} & \textbf{998/0} & \textbf{1000/0} & \textbf{1000/0} & \textbf{1000/0} \\
OK-CUSUM & 821/0 & 914/0 & 944/0 & 836/0 & 904/0 & 956/0 & 839/0 & 915/0 & 960/0 & 839/0 & 916/0 & 955/0 & 827/0 & 909/0 & 960/0 \\
Scan B & 817/0 & 910/0 & 958/0 & 839/0 & 919/0 & 954/0 & 860/0 & 933/0 & 967/0 & 850/0 & 916/0 & 961/0 & 873/0 & 934/0 & 972/0 \\
KCUSUM & 825/67 & 768/190 & 602/386 & 598/245 & 544/389 & 462/506 & 682/177 & 634/295 & 563/402 & 471/381 & 379/545 & 287/673 & 686/162 & 644/287 & 558/410 \\
NEWMA & 0/661 & 0/833 & 0/922 & 0/671 & 0/828 & 0/909 & 0/712 & 0/831 & 0/916 & 0/693 & 0/818 & 0/894 & 0/689 & 0/822 & 0/925 \\
PCA-CUSUM & 0/855 & 0/932 & 0/969 & 832/0 & 912/0 & 949/0 & 818/0 & 904/0 & 950/0 & 830/0 & 925/0 & 972/0 & 86/708 & 47/851 & 20/918 \\
Hotelling $T^2$ & 1/822 & 0/924 & 1/964 & 848/0 & 927/0 & 967/0 & 871/0 & 935/0 & 962/0 & 863/0 & 934/0 & 967/0 & 843/0 & 929/0 & 967/0 \\
RAW-CUSUM & 0/832 & 0/909 & 0/960 & 856/0 & 928/0 & 972/0 & 919/0 & 975/0 & 994/0 & 842/0 & 919/0 & 962/0 & 845/0 & 926/0 & 968/0 \\
GraphScan-kNN & 5/701 & 5/839 & 3/913 & 356/477 & 233/691 & 117/841 & 684/152 & 594/323 & 454/498 & 462/387 & 357/571 & 239/719 & 734/114 & 673/237 & 540/417 \\
PCA-BOCPD & 0/849 & 0/989 & 0/989 & 13/837 & 8/920 & 4/957 & 38/786 & 20/885 & 13/942 & 187/649 & 114/805 & 61/894 & 55/764 & 31/873 & 14/936 \\
E-GaussianBet & 0/1000 & 0/1000 & 0/1000 & 404/587 & 254/742 & 175/824 & 560/417 & 414/574 & 304/689 & 962/24 & 930/63 & 874/121 & 581/411 & 410/586 & 282/715 \\
\bottomrule
\end{tabular}
}
\end{table}
\FloatBarrier

PL-CUSUM raises no alarms among the 2000 state-0 paths for SWaT, Motor under
no load, Motor under load, and Motor under combined load and background
vibration. In the Motor background-vibration condition, 22 state-0 paths
raise an alarm, giving an estimated $\mathrm{ARL}_0$ of 18,082. Among the
1000 change-point paths, two raise a pre-change false alarm and the remaining
998 yield successful detections. The conditional EDD is 0.091.

\noindent\textit{Method trajectories and PL increment distributions.}
For each task, we select the change-point path whose PL-CUSUM detection delay
is closest to its median. The 11 methods in
Table~\ref{tab:real_arl1000} are then compared on this path. Let $S_{r,t}$ and
$\eta_r$ denote the monitoring statistic and control limit for method $r$.
Figure~\ref{fig:appendix_method_trajectories} plots
$\widetilde S_{r,t}=(S_{r,t}-\eta_r)/\max_u|S_{r,u}-\eta_r|$.

\begin{figure}[!htbp]
\centering
\includegraphics[width=\linewidth]{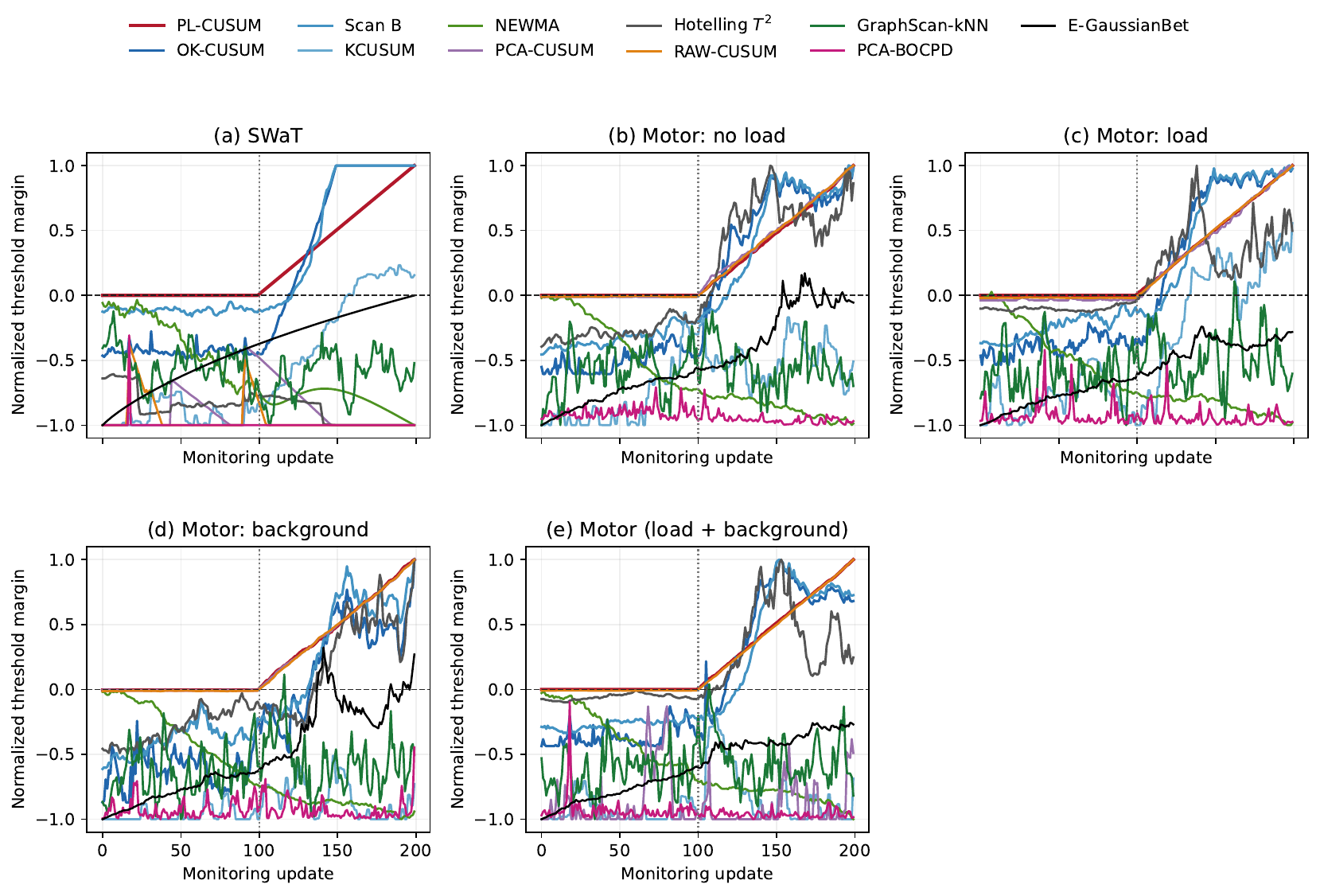}
\caption{Normalized monitoring statistics for 11 methods on SWaT and the four
Electric Motor conditions. The methods in each panel are evaluated on the
same selected change-point path. The vertical dashed line marks the change
point, and the horizontal dashed line marks the control limit.
A value $\widetilde S_{r,t}>0$ indicates that the statistic has crossed its
control limit.}
\label{fig:appendix_method_trajectories}
\end{figure}

PL-CUSUM crosses the control limit at the first post-change window on all five
selected paths. Among the remaining 50 method--task trajectories, two cross
before the change, 28 cross within 0--62 updates after the change, and 20 do
not cross within the monitoring horizon. The overall results for all methods
are reported in Tables~\ref{tab:appendix_real_arl_matrix}--%
\ref{tab:appendix_real_counts_matrix}.
\FloatBarrier

Figure~\ref{fig:appendix_pl_increments} compares the distributions of the
one-step PL increments under states 0 and 1.

\begin{figure}[!htbp]
\centering
\includegraphics[width=\linewidth]{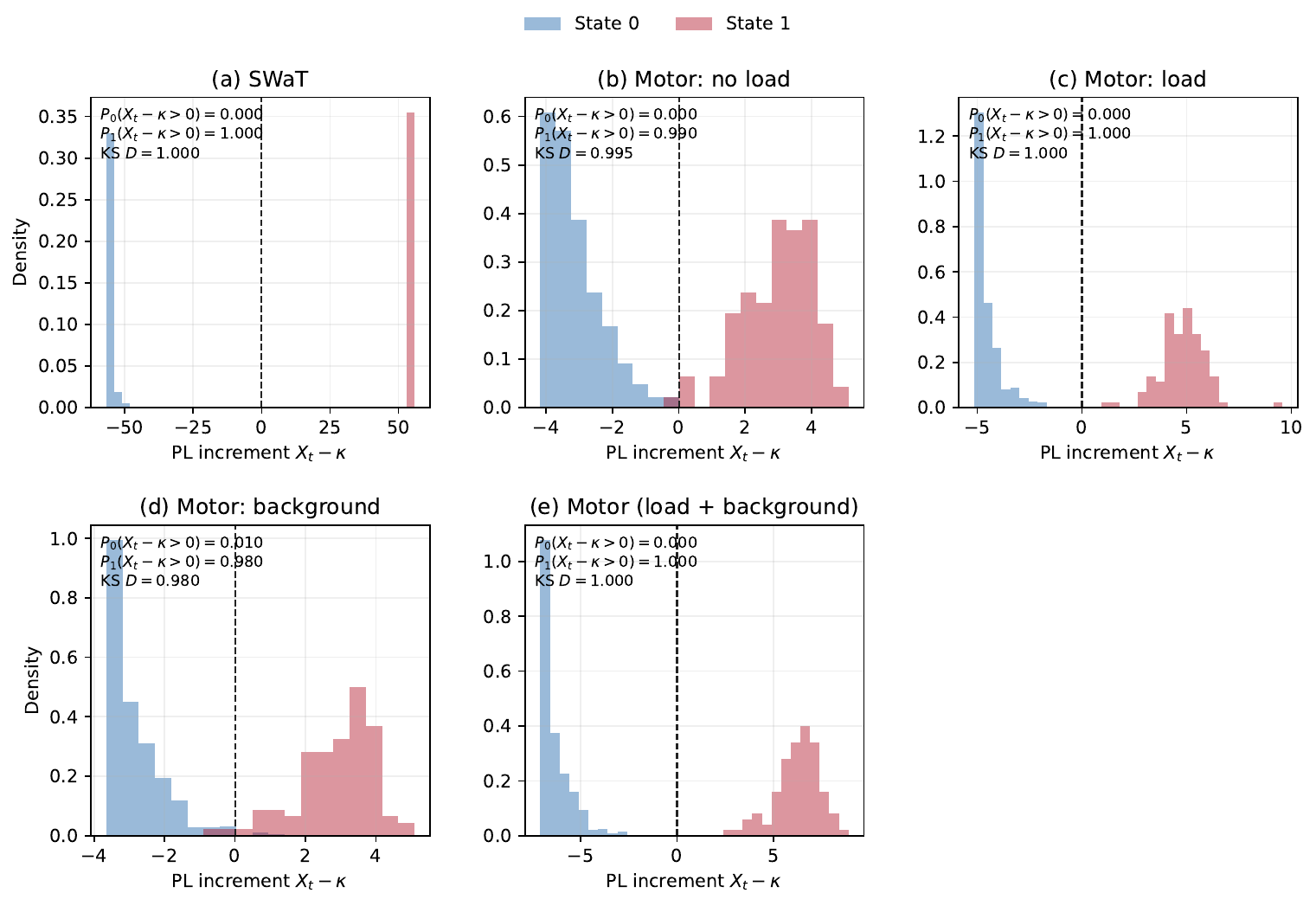}
\caption{Distributions of the one-step PL increment $X_t-\kappa$ under states
0 and 1 for SWaT and the four Electric Motor conditions. SWaT uses 500 state-0
windows, and each Motor condition uses 400 state-0 windows. Each task uses 100
state-1 windows. The dashed line marks zero increment. Each panel reports the
proportions of positive increments and the two-sample KS statistic.}
\label{fig:appendix_pl_increments}
\end{figure}
\FloatBarrier

The KS statistic is 1.000 for SWaT, Motor under load, and Motor under combined
load and background vibration. For Motor under no load, the proportions of
positive increments under states 0 and 1 are 0 and 0.99, and the KS statistic
is 0.995. For Motor under background vibration, the corresponding values are
0.01, 0.98, and 0.980. The state-0 increments lie mainly below zero, whereas
the state-1 increments lie mainly above zero. This separation explains the
low false-alarm rates and short detection delays of PL-CUSUM on these tasks.

\subsection{Experimental Protocol and Method Parameters}
\label{app:protocol_details}

Table~\ref{tab:appendix_data_protocol} summarizes the window dimensions, pool
sizes, reference windows, bootstrap block lengths, and numbers of calibration,
state-0 evaluation, and change-point paths. The Electric Motor row applies to
all four operating conditions. The data partitions are formed at the raw
sequence level before the windows are generated.

\begin{table}[!htbp]
\centering
\caption{Protocol for the real-data experiments. Phase~I 0/1 gives the numbers
of training windows under states 0 and 1. Cal.\ pool and Alt.\ pool denote the
state-0 calibration pool and state-1 evaluation pool, respectively. Ref.\
gives the number of state-0 reference windows used by the comparison methods.}
\label{tab:appendix_data_protocol}
\resizebox{\linewidth}{!}{%
\begin{tabular}{lrrrrrrrrrr}
\toprule
Data & $w$ & $d$ & Stride & Phase I 0/1 & Cal. pool & Alt. pool & Ref. & Block & $M_{\rm cal}$ & State-0/Change \\
\midrule
SWaT & 32 & 51 & 16 & 100/100 & 500 & 100 & 750 & 2 & 2000 & 2000/1000 \\
Motor (four settings) & 64 & 3 & 64 & 100/100 & 400 & 100 & 750 & 1 & 2000 & 2000/1000 \\
\bottomrule
\end{tabular}
}
\end{table}
\FloatBarrier

Table~\ref{tab:appendix_method_parameters} lists the parameters used by
PL-CUSUM and the comparison methods.

\begingroup
\scriptsize
\setlength{\LTleft}{0pt}
\setlength{\LTright}{0pt}
\setlength{\tabcolsep}{4pt}
\renewcommand{\arraystretch}{0.96}
\begin{longtable}{L{0.22\linewidth}L{0.72\linewidth}}
\caption{Method parameters for the real-data experiments.}
\label{tab:appendix_method_parameters}\\
\toprule
Method or setting & Parameters \\
\midrule
\endfirsthead
\toprule
Method or setting & Parameters \\
\midrule
\endhead
\bottomrule
\endfoot
\multicolumn{2}{l}{\textit{Task-specific PL-CUSUM configurations}} \\
\addlinespace[0.25em]
SWaT & $L=1$; full PL input; $J=31$; 8 scale pairs; selected pair $(14.051,2906.380)$; ridge multiplier $2$ \\
Motor: no load & $L=6$; full PL input; $J=58$; 4 scale pairs; selected pair $(6.405,7.209)$; ridge multiplier $0.125$ \\
Motor: load & $L=4$; full PL input; $J=60$; 8 scale pairs; selected pair $(6.353,6.990)$; ridge multiplier $0.125$ \\
Motor: background & $L=4$; full PL input; $J=60$; 4 scale pairs; selected pair $(5.866,6.924)$; ridge multiplier $0.125$ \\
Motor: load + background & $L=3$; full PL input; $J=61$; 4 scale pairs; selected pair $(5.112,6.178)$; ridge multiplier $0.125$ \\
\midrule
\multicolumn{2}{l}{\textit{Baseline parameter rules}} \\
\addlinespace[0.25em]
OK-CUSUM / Scan B & $B_{\min}=2$; $B_{\max}=w=50$; $N=15$; block step 2; task-specific median-distance Gaussian RBF bandwidth \\
KCUSUM & drift $0.02$; task-specific median-distance Gaussian RBF bandwidth \\
RBF bandwidths & In the order SWaT, Motor no load, load, background, and load + background: OK-CUSUM/Scan B use $(100317.596,17.480,17.142,18.194,17.459)$; KCUSUM uses $(101702.437,17.498,17.180,18.192,17.466)$ \\
NEWMA & $B=50$; fast and slow forgetting factors $0.03363$ and $0.01010$; 130 random features; bandwidth parameter determined by the upstream formula \\
PCA-CUSUM & 20 principal components; ridge $0.01$ \\
Hotelling $T^2$ & window length 50; ridge multiplier $0.01$ \\
RAW-CUSUM & standardized window-vector energy; ridge $10^{-6}$ \\
GraphScan-kNN & $k=4$; history length $L=16$; recent length $w=8$; standardized coordinates \\
PCA-BOCPD & four principal components; hazard mean 50; maximum run length 256; alarm radius 5 \\
E-GaussianBet & automatic projection; two-sided Gaussian bet; SR recursion; $\eta\in\{0.05,0.1,0.2,0.4,0.8\}$ \\
\end{longtable}

\endgroup

\end{appendices}

\end{document}